\documentclass[11pt,twoside]{article} 
\usepackage[utf8]{inputenc} \usepackage[T1]{fontenc} 
\usepackage[a4paper, left=25mm, right=25mm, top=25mm, bottom=25mm]{geometry} 
\usepackage{graphicx} \usepackage{setspace} \usepackage{xcolor} \usepackage{colortbl} 

\usepackage[switch,mathlines]{lineno}

\usepackage{cite} 
\usepackage{amsmath,amssymb,amsthm} \usepackage{mathtools} \mathtoolsset{showonlyrefs} \usepackage{mathrsfs} \usepackage{bm} 
\usepackage{circuitikz} 
\usepackage{tikz-cd} \usetikzlibrary{cd} 
\usepackage{enumerate}
\usepackage{marginnote}

\usepackage{imakeidx} \makeindex[intoc]

\usepackage[nottoc]{tocbibind} 

\usepackage[colorlinks=true]{hyperref} \hypersetup{urlcolor=blue, citecolor=red, linkcolor=blue} 

\theoremstyle{plain}
\newtheorem{theorem}{Theorem}[section]
\newtheorem{proposition}[theorem]{Proposition} 
\newtheorem{lemma}[theorem]{Lemma}

\theoremstyle{definition}
\newtheorem{definition}[theorem]{Definition}

\theoremstyle{remark}
\newtheorem{remark}[theorem]{Remark}
\newtheorem{example}[theorem]{Example} 

\newcommand{\R}{\mathbb{R}}
\newcommand{\dd}{\mathrm{d}}
\renewcommand{\d}{\mathrm{d}}
\newcommand{\Cinfty}{\mathscr{C}^\infty}
\newcommand{\T}{\mathrm{T}}
\newcommand{\cT}{\mathrm{T}^\ast}

\newcommand*{\inn}[1]{\iota_{#1}}
\newcommand{\Lie}{\mathscr{L}}

\DeclareMathOperator{\Ima}{Im}

\DeclareMathOperator{\rk}{rank}
\DeclareMathOperator{\cork}{corank}

\DeclareMathOperator{\pr}{pr}
\newcommand\restr[2]{\left.#1\right|_{#2}}     \newcommand{\parder}[2]{\frac{\partial #1}{\partial #2}}

\usepackage{fancyhdr}

\begin{document} \parskip=3pt 

\vspace{5em} {\huge\sffamily\raggedright \begin{spacing}{1.1} Hamilton--Jacobi theory for non-conservative field theories in the $k$-contact framework \end{spacing} } \vspace{2em} 
{\Large\raggedright\sffamily Javier de Lucas }\vspace{1mm}\newline {\raggedright Centre de Recherches Math\'ematiques, Universit\'e de Montr\'eal,\\ Succ. Centre-Ville, CP6128,  Montréal (Québec), Canada H3C 3J7.\\ Department of Mathematical Methods in Physics, University of Warsaw,\\ ul. Pasteura 5, 02-093 Warszawa, Poland.\\ e-mail: \href{mailto:javier.de.lucas@fuw.edu.pl}{javier.de.lucas@fuw.edu.pl} --- orcid: \href{https://orcid.org/0000-0001-8643-144X}{0000-0001-8643-144X} } \bigskip \\{\Large\raggedright\sffamily Julia Lange }\vspace{1mm}\newline {\raggedright Department of Mathematical Methods in Physics, University of Warsaw,\\ ul. Pasteura 5, 02-093 Warszawa, Poland.\\ e-mail: \href{mailto:j.lange2@uw.edu.pl}{j.lange2@uw.edu.pl} --- orcid: \href{https://orcid.org/0000-0001-6516-0839}{0000-0001-6516-0839} } \bigskip \\{\Large\raggedright\sffamily Xavier Rivas }\vspace{1mm}\newline {\raggedright Department of Computer Engineering and Mathematics, Universitat Rovira i Virgili,\\ Avinguda Països Catalans 26, 43007 Tarragona, Spain.\\ e-mail: \href{mailto:xavier.rivas@urv.cat}{xavier.rivas@urv.cat} --- orcid: \href{https://orcid.org/0000-0002-4175-5157}{0000-0002-4175-5157} } \bigskip \\{\Large\raggedright\sffamily Cristina Sardón }\vspace{1mm}\newline {\raggedright Departamento de Matemática Aplicada, Universidad Politécnica de Madrid,\\ Av. Juan de Herrera 6, 28040 Madrid, Spain.\\ e-mail: \href{mailto:mariacristina.sardon@upm.es}{mariacristina.sardon@upm.es} --- orcid: \href{https://orcid.org/0000-0001-9237-4373}{0000-0001-9237-4373} } \vspace{2em}

{\large\bfseries\raggedright Abstract}\vspace{1mm}\newline
{\raggedright
This article develops a Hamilton--Jacobi theory for non-conservative classical field theories, with particular emphasis on dissipative systems, in the framework of co-oriented \(k\)-contact geometry. We introduce evolution \(k\)-contact \(k\)-vector fields, extending the contact evolution formalism to field theories, and analyse the corresponding Hamilton--De Donder--Weyl equations. Moreover, we develop two distinct families of Hamilton--Jacobi theories: a \(z\)-independent approach, based on the reconstruction of the dynamics from an integrable \(k\)-vector field defined on the base manifold of \(\bigoplus^k\T^*Q\times\mathbb{R}^k\to Q\), and a \(z\)-dependent approach, where the integrable \(k\)-vector field is defined on the base manifold of \(\bigoplus^k\T^*Q\times\mathbb{R}^k\to Q\times\mathbb{R}^k\). We develop in detail the important case of Hamiltonian functions with affine dependence on the dissipative variables, show how quadratic dependence on these variables can be used structurally to enlarge the range of applications, and recover the ordinary contact Hamilton--Jacobi theory as the particular case \(k=1\), while removing some technical assumptions appearing in previous formulations. Our theory is illustrated through several representative examples, including the telegrapher/Klein--Gordon equation, a dissipative Hunter--Saxton equation, a simple dissipative non-regular first-order field model, and a relativistic thermodynamic model.
}

\bigskip


\medskip

\newpage


{\setcounter{tocdepth}{2}
\def\baselinestretch{1}
\small
\def\addvspace#1{\vskip 1pt}
\parskip 0pt plus 0.1mm
\tableofcontents
}

\pagestyle{fancy}

\fancyhead[LE]{Hamilton--Jacobi theory for non-conservative field theories in the $k$-contact framework} 
\fancyhead[RE]{}
\fancyhead[RO]{J. de Lucas, J. Lange, X. Rivas and C. Sardón} 
\fancyhead[LO]{}    

\fancyfoot[L]{}     
\fancyfoot[C]{\thepage}                  
\fancyfoot[R]{}            


\renewcommand{\headrulewidth}{0.1pt}  
\renewcommand{\footrulewidth}{0pt}    

\renewcommand{\headrule}{%
    \vspace{3pt}                
    \hrule width\headwidth height 0.4pt 
    \vspace{0pt}                
}

\setlength{\headsep}{30pt}  

 
\section{Introduction}

Classical field theories have a long history in physics, since they describe the evolution of physical fields such as electromagnetic or gravitational ones. From a geometric viewpoint, one may regard them as a generalisation of the Hamiltonian formalism of classical mechanics obtained by allowing for several independent variables \cite{Pri_14,GRR_20,LMM_03,KT_79,BC_23,CRS_00,LSV_15}. Field theory was later adapted to the quantum realm \cite{BrunettiFredenhagenVerch2003,Fewster2015,Rejzner2016}, but in this work we focus exclusively on classical field theories. Nonetheless, establishing rigorous geometric foundations for classical field theory is far from trivial, and modern differential-geometric formulations of variational calculus and field theory were developed gradually, with gauge theories and related geometric methods providing a major stimulus \cite{AA_80,LR_85,LSV_15,KT_79}.

Among geometric formalisms for classical field theories, the $k$-symplectic (also called $k$-polysymplectic) framework ($k\geq 1$) plays a distinguished role \cite{Awa_92,Gun_87,LSV_15}. It generalises the standard symplectic formalism of autonomous mechanics to field theories \cite{Gun_87}, and is particularly well suited to models whose Hamiltonian or Lagrangian functions do not depend explicitly on the space-time coordinates \cite{Awa_92,Awa_94,AG_00}. In this framework, the analogue of Hamilton equations is given by systems of partial differential equations known as the Hamilton--De Donder--Weyl (HdDW) equations. Moreover, $k$-symplectic geometry is also useful beyond classical field theory; for instance, it has been employed in the study of ordinary differential equations \cite{LV_15,LS_20}.

A characteristic feature of field theory is that its geometric dynamics is naturally described by means of $k$-vector fields. However, although integral sections of $k$-vector fields solving the geometric HdDW equations provide solutions of the corresponding HdDW equations, the converse need not hold due to the lack of integrability of certain $k$-vector fields. This phenomenon already appears in other geometric descriptions of field theories, such as multisymplectic geometry \cite{EMR_99,ELSZ_21}, and it also arises in the framework considered here, namely that of co-oriented $k$-contact manifolds \cite{LRS_25,GGMRR_20,GGMRR_21,GRR_22,Riv_22}.

This work aims to study Hamilton--Jacobi theories for field theories in the $k$-contact geometric setting. 
A co-oriented $k$-contact manifold is a manifold $M$ endowed with a $k$-contact form $\bm\eta=\sum_{\alpha=1}^k\eta^\alpha\otimes e_\alpha$, namely an $\mathbb{R}^k$-valued one-form on $M$ with $\{e_1,\ldots,e_k\}$ being the canonical basis of $\mathbb{R}^k$ such that
\(
\ker \bm\eta\oplus \ker \d\bm\eta=\T M,
\)
where $\ker\bm\eta=\cap_{\alpha=1}^k\ker \eta^\alpha$ is assumed to be a nonzero distribution of corank $k$ and $ \ker\dd\bm\eta=\cap_{\alpha=1}^k\ker \dd \eta^\alpha$. The case $k=1$ corresponds to co-oriented contact manifolds, which have been extensively studied in the literature \cite{LGLRR_20,LL_20,LL_19,ELLM_21,GGMRR_20a,SMLL_21}. In this sense, $k$-contact geometry may be understood as a natural extension of contact geometry.

As contact geometry can be used to study dynamical systems \cite{BCT_17,KT_10,LL_19}, $k$-contact geometry can be viewed as a modern generalisation of contact geometry to analyse field theories \cite{LL_19,GGMRR_20,LLM_21}. In this context, $k$-contact geometry provides the appropriate framework to incorporate dissipative effects in a genuinely field-theoretic manner. Its first systematic appearance can be traced to the $k$-contact Hamiltonian formalism for dissipative fields developed in \cite{GGMRR_20}, which showed that contact-type dissipation admits a natural multi-time geometric description. Moreover, \cite{Riv_22}  and related works provide a solid geometric background linking contact mechanics and field theory with the original motivations of the subject. 

More recently, the geometric foundations of $k$-contact geometry have been clarified in \cite{LRS_25}, where the basic structure of co-oriented $k$-contact manifolds, their Reeb distributions, $k$-contact distributions, characteristics for Lie symmetries, Darboux coordinates, and the role of polarisations are studied in detail. 

Current developments show that $k$-contact geometry is already relevant for concrete applications: it underlies geometric models for dissipative field equations, pseudo-gauge transformations and constitutive structures in relativistic hydrodynamics, and thermodynamics \cite{GGMRR_20,GGMRR_21,HLM_26}, and it may also be applied to develop methods for standard ordinary differential equations and control systems \cite{CAR_25, deLucasRivasSobczak2026,LRS_25,SF_25}. Altogether, these works indicate that $k$-contact geometry is becoming a natural language for certain classes of classical field theories, especially dissipative ones.

To illustrate the applications of $k$-contact geometry, one can describe geometrically several non-conservative field theories. A relevant guiding example is the damped vibrating membrane \cite{Rao2007,Graff1975,MorseIngard1968},
\begin{equation*} 
u_{tt}-\nu^2(u_{xx}+u_{yy})+\lambda u_t+\kappa u=0,
\end{equation*}
where $\nu$ is the wave speed in the medium, $\lambda$ and $\kappa$ are scalar real parameters, and $u=u(t,x,y)$ represents the displacement of the membrane. This equation serves as an introductory example and motivates the role of Hamiltonians that are affine in the variables that will hereafter be called dissipative variables and that naturally appear in $k$-contact geometry.

The concept of contact Hamiltonian vector field can be extended to $k$-contact geometry in different manners \cite{LRS_25,Riv_22}. In particular, we first study classical $k$-contact Hamiltonian $k$-vector fields in co-oriented $k$-contact manifolds \cite{Riv_22}. 
A relevant point, which becomes crucial for Hamilton--Jacobi theory, is that different $k$-contact Hamiltonian $k$-vector fields may be associated with the same Hamiltonian function. This non-uniqueness \cite{HLM_26} is described here by the kernel of a natural vector bundle morphism
\[
\textstyle\chi\colon (v_1,\dotsc,v_k)\in 
\bigoplus^k \T M\longmapsto \left(\sum_{\alpha=1}^k\iota_{v_\alpha}\d\eta^\alpha,\sum_{\alpha=1}^k\iota_{v_\alpha}\eta^\alpha\right)\in \T^*M\times\mathbb R\,,
\]
which is injective in the standard co-oriented contact case, where $k=1$. This motivates us to call $k$-vector fields taking values in $\ker\chi$ {\it gauge} $k$-vector fields or ${\bm\eta}$-gauge $k$-vector fields to emphasise the co-oriented $k$-contact manifold we are considering. 

This work also extends to co-oriented $k$-contact manifolds the notion of evolution Hamiltonian vector field, which appears in the contact-geometric approach to thermodynamic systems \cite{LL_19}. More precisely, we introduce the notion of evolution $k$-contact $k$-vector field and study its associated Hamilton--De Donder--Weyl equations. This yields a new type of Hamilton--De Donder--Weyl equations.

It is worth noting that the development of our $k$-contact Hamilton--Jacobi theories in this work requires the introduction of new $k$-contact geometric tools, such as maximally coisotropic subspaces and suitable regularity conditions for Hamiltonians on $k$-contact manifolds extending the ideas in \cite{GGMRR_21,LLM_21a}. In particular, this work shows that the formulation of Hamilton--Jacobi theory in the $k$-contact setting must take into account not only the HdDW equations themselves, but also integrability, projectability, and the role of gauge $k$-contact $k$-vector fields. These issues will become especially transparent for Hamiltonians affine in the dissipative variables. 

Roughly speaking, in the classical geometric Hamilton--Jacobi setting, the full dynamics is described by integral curves of a Hamiltonian vector field on a phase space \cite{AM_78}. Since integrating this vector field directly is often difficult, one tries to reconstruct the dynamics from a vector field defined on the base manifold of the phase bundle. To do so, one looks for a section of that bundle such that the projected dynamics and the full dynamics at points of the section are suitably related. This section, which is required to satisfy certain appropriate conditions, is interpreted as a solution of the Hamilton--Jacobi equation and, in the classical symplectic case, it is locally given by the differential of a generating function.

This classical picture can be represented by the commutative diagram
\begin{equation}\label{diagramintro}
\begin{tikzcd}[row sep=huge, column sep=huge]
E  \arrow[r, "X_h"] & \T E \arrow[d, swap, "\T \rho"] \\
M \arrow[r, "X_h^\gamma"] \arrow[u,  "\gamma"] & \T M \arrow[u, bend right=40, swap, "\T \gamma"]
\end{tikzcd}
\end{equation}
where $\gamma\colon M\to E$ is a section of the bundle $\rho\colon E\to M$, and $X^\gamma_h$ is constructed via $X_h$ as ${\rm T\rho}\circ X_h\circ \gamma$. Then, the diagram is commutative if and only if the dynamics on $\gamma(M)$ can be reconstructed from the projected vector field $X_h^\gamma$, namely when
\(
\T\gamma\bigl(X^\gamma_h\bigr)=X_h\circ \gamma.
\) In particular, the Hamilton--Jacobi problem relies on determining conditions on $\gamma$ to ensure that the above diagram is commutative.
In the classical symplectic case, $X_h$ is a Hamiltonian vector field on $E=\T^*Q$, while $M=Q$, and $\gamma = \d S$ locally, where $S$ is a generating function, is such that $h \circ \gamma$ is constant.

The above diagrammatic viewpoint is flexible enough to encompass several geometric settings. In the $k$-symplectic setting \cite{LSV_15} one has $E=\bigoplus^k \T^*Q$, which is endowed with an exact and non-degenerate two-form $\bm \omega$ taking values in $\mathbb{R}^k$, we set $M=Q$, and $\gamma$ is a family of $k$ closed one-forms. Meanwhile, ${\bm X}_h\colon E\to \bigoplus^k\T E$ is a $k$-symplectic Hamiltonian $k$-vector field and ${\bm X}^\gamma_h\colon M\to \bigoplus^k\T M$ is a $k$-vector field obtained from  ${\bm X}_h$ by projecting along ${\rm Im}\,\gamma$ onto $M$. The fact that $h$ does not uniquely determine ${\bm X}_h$ is not a very problematic issue. If the projection of ${\bm X}_h$ over ${\rm Im}\,\gamma$ onto $Q$ gives rise to an integrable $k$-vector field ${\bm X}^\gamma_h$, there is a Hamilton--Jacobi theorem ensuring that the lift of integral sections of ${\bm X}^\gamma_h$ solve the HdDW equations for ${\bm X}_h$. Finally, the Hamilton--Jacobi equations reduce to ensuring that $\dd(h\circ \gamma)=0$.

For $k$-contact manifolds, the relevant total space is
\(
E=\mathcal{J}_{Q,k}=\bigoplus^k \T^*Q\times\mathbb R^k\,,
\)
where the variables $z^1,\dots,z^k$, obtained by extending to $E$ from canonical coordinates in $\mathbb{R}^k$, are used to encode a dissipative behaviour in field theories. That is why they are hereafter called {\it dissipative variables}. Moreover $M=Q$ while the manifold $\mathcal{J}_{Q,k}$ has a natural $k$-contact form ${\bm \eta}_{Q,k}$, which is due to the fact that $\mathcal{J}_{Q,k}$ can be considered as a one-jet manifold \cite{LRS_25} and the kernel of ${\bm \eta}_{Q,k}$ locally behaves as the Cartan distribution of a first-order jet manifold of order $k$.  Moreover, it makes sense to define holonomic sections as those of the form $Q\ni q\rightarrow (q,d_qz^\alpha,z^\alpha(q))\in \mathcal{J}_{Q,k}$. Here, again, $h\in\Cinfty(E)$ does not uniquely determine an $\bm\eta_{Q,k}$-Hamiltonian $k$-vector field ${\bm X}_h$. This non-uniqueness is due to the non-trivial kernel of the morphism $\chi$ for $k>1$, while the dependence on the dissipative variables makes this gauge freedom especially visible in the $z$-dependent approach. This implies that the Hamilton--Jacobi theories must be formulated in terms of classes of ${\bm \eta}_{Q,k}$-Hamiltonian $k$-vector fields related to the same $h$. 

As in the Hamilton--Jacobi theory for contact Hamiltonian systems \cite{LLM_21a,LLGR_23}, we offer two different types of Hamilton--Jacobi theories concerning the choice of the base manifold in the bundle $\mathcal{J}_{Q,k}\rightarrow M$. In the first one, the base is $M=Q$, which yields the $z$-independent theory. In the second one, the base is $M=Q\times\mathbb R^k$, which yields the $z$-dependent theory. Both approaches are developed in this article, both for standard ${\bm \eta}_{Q,k}$-Hamiltonian $k$-vector fields and for ${\bm \eta}_{Q,k}$-evolution   $k$-vector fields. Due to its relation to physical theories and variational principles, the $z$-independent and $z$-dependent formulations could be called action-independent and action-dependent Hamilton--Jacobi formulations \cite{LPAF2018,LPF2019,LPAF2017}.

A further relevant point is that, in the $z$-independent approach, the natural sections are holonomic and their images are Legendrian submanifolds relative to $(\mathcal{J}_{Q,k},\bm\eta_{Q,k})$. Nevertheless, it is not in general possible to capture the full dynamics of the initial $\bm\eta_{Q,k}$-Hamiltonian/evolution $k$-vector field due to its dependence on the dissipative variables. At most, we can recover another $\bm\eta_{Q,k}$-Hamiltonian/evolution $k$-vector field related to the  HdDW equations for the same $h$. In any case, projected $k$-vector fields are independent of the representative $\bm\eta_{Q,k}$-Hamiltonian $k$-vector field related to $h$. As our real aim is to solve HdDW equations for the initial $h$, this is not a significant issue. 

In the $z$-dependent approach, an additional problem appears. The projection of the initial $k$-vector field onto $Q\times \mathbb{R}^k$ depends on the representative $\bm\eta_{Q,k}$-Hamiltonian/evolution $k$-vector field. Hence, the Hamilton--Jacobi equation is described for classes of $k$-vector fields on $Q\times \mathbb{R}^k$ coming from classes of ${\bm \eta}_{Q,k}$-Hamiltonian/evolution $k$-vector fields on $\mathcal J_{Q,k}$ associated with the same function. In the case $k=1$, we retrieve the standard contact Hamilton--Jacobi theories in \cite{LLM_21a} and weaken some technical assumptions appearing in this work. Under adequate assumptions on $h$, we can also retrieve the Hamilton--Jacobi theories for $k$-symplectic manifolds as a particular case of our results \cite{DELEON2010HAMILTONJACOBITHEORIES}. Indeed, it is worth noting that $k$-symplectic theories also deal with the fact that a function $h$ does not determine uniquely an associated $k$-symplectic  Hamiltonian $k$-vector field \cite{DELEON2010HAMILTONJACOBITHEORIES} and this can be seen in the corresponding Hamilton--Jacobi equations. 

Most applications of $k$-contact geometry deal with Hamiltonian models that are affine in the dissipative variables, which is one of the reasons why the theory has been closely connected with dissipation and why its relevance in physical applications has already been discussed \cite{Riv_22,RSS_24}.
 Only a few models, such as the parachute model in \cite{GGMRR_20a}, involve other types of dependence in the dissipative variables. In this work, we show that Hamiltonians with a quadratic term on the dissipative variables also arise naturally and broaden the scope of the theory, in particular towards non-autonomous models.  Moreover, the damped Klein--Gordon equation and several other physical field-theory models are analysed. The approach is general enough to be extended to broader settings.

The applications developed in this work show that the proposed $k$-contact Hamilton--Jacobi framework is flexible enough to cover several qualitatively different classes of non-conservative field theories. To the best of our knowledge, most examples in this work have not been previously analysed in the literature. First, the damped telegrapher/Klein--Gordon equation provides a natural test for the formalism. Physically, it describes damped propagation phenomena and includes, as a particular case, the classical transmission-line equation with clear physical parameters coming from resistance, inductance, capacitance, and leakage effects \cite{RWD_94}. Second, the dissipative Hunter--Saxton equation shows that the theory also applies to nonlinear dissipative wave equations arising in the geometry and dynamics of director fields and in related integrable and weakly dissipative models \cite{HS_91,HZ_94,Len_07,WY_11,Wei_12,LZ_11}. Third, the simple dissipative first-order field model  clarifies how the standard and evolution Hamilton--De Donder--Weyl formulations may produce the same projected field dynamics while encoding dissipation differently through the contact variables. It also shows how the lack of regularity affects the way they describe second-order systems of PDEs in a natural manner. Finally, the relativistic thermodynamic example indicates that the formalism is not restricted to scalar PDEs, but also accommodates balance laws and constitutive relations of extended irreversible thermodynamics, where entropy is naturally described by fluxes and the geometric structure becomes genuinely $k$-contact \cite{BP_93,Bra_18,HLM_26}.

The paper is organised as follows. Section~\ref{sec:2} reviews the geometric background needed in the sequel, including $k$-vector fields, $k$-contact manifolds, $k$-contact Hamiltonian systems, the associated Hamilton--De Donder--Weyl equations, and the new notion of evolution $k$-contact  $k$-vector field. It also contains several new structural results on HdDW equations, regular Hamiltonians, Hamiltonians affine in the dissipative variables, and maximally coisotropic subspaces. Sections~\ref{sec:HJzindependent} and~\ref{sec:HJzdependent} contain the main theoretical contributions of the paper. More precisely, Section~\ref{sec:HJzindependent} develops the Hamilton--Jacobi theory in the $z$-independent approach, both in the standard setting of Section \ref{Sec::zindclass} and in the evolution setting  of Section \ref{sec:zindevolution}, while Section~\ref{sec:HJzdependent} presents the corresponding $z$-dependent theory related to the standard and evolution perspectives in Sections \ref{sec:zdepclassical} and \ref{sec:HJzdepEvolution}, respectively. Section \ref{sec:Integrable} discusses integrable $k$-contact Hamiltonian systems and  Section~\ref{sec:Applications} is devoted to applications, namely the telegrapher/Klein--Gordon equation, the dissipative Hunter--Saxton equation, a simple dissipative first-order field model, and a relativistic thermodynamic model. Finally, Section~\ref{sec:ConclusionsAndOutlook} contains the conclusions and outlook.

\section{{\it k}-contact manifolds and {\it k}-contact Hamiltonian systems}
\label{sec:2}

Let us describe the fundamental notions and results used hereafter. Recall that $\mathbb{R}^k$ admits a canonical basis $\{e_1,\dotsc,e_k\}$, where $k$ denotes hereafter a natural positive number. Direct sums of vector bundles are considered to be Whitney sums unless otherwise stated.  We also assume $M$ to be an $m$-dimensional manifold and $\T M$ stands for its tangent bundle. Unless otherwise stated, structures are considered to be smooth.

The structure of this section is as follows. After clarifying some very standard concepts on distributions, codistributions, and $k$-vector fields, we recall the definition of $k$-contact manifolds and $k$-contact Hamiltonian systems. Relevantly, we introduce the notion of evolution $k$-contact $k$-vector field, which is a natural extension of the evolution Hamiltonian vector field in contact geometry \cite{LL_19} to the realm of field theories. Maximally coisotropic subspaces are also defined for the first time, as this notion is important for defining $k$-contact Hamilton--Jacobi theories in the next sections. We also study how our formalism applies to regular Hamiltonian functions with an affine dependence on the dissipative variables and how they canonically yield second-order partial differential equations, which is relevant for applications.

\subsection{Differential forms taking values in vector spaces, distributions, and \texorpdfstring{$k$}{}-vector fields}

 Let us recall some basic definitions regarding $k$-contact vector fields, $k$-contact geometry, and other related notions to be used hereafter.

A \textit{distribution} on $M$ is a subset of $\T M$ such that $D_x = D\cap\T_xM$ is a vector subspace of $\T_xM$ at every $x\in M$. We call $\dim D_x$ the {\it rank} of $D$ at $x\in M$. A distribution $D$ is \textit{smooth} if there exists, for every $x\in M$, a family of vector fields $X_1, \dotsc, X_r$ defined in a neighbourhood $U$ of $x$ such that $D_{x'}=\langle X_1(x'), \dotsc, X_r(x')\rangle$ for every $x'\in U$. Note that $r$ does not need to match the dimension of $D_x$.  A distribution $D$ is \textit{regular} if it is smooth and of constant rank. A \textit{codistribution} on $M$ is a subset $C\subset\cT M$ such that $C_x = C\cap\cT_xM$ is a vector subspace of $\cT_x M$ for every $x\in M$. The {\it annihilator} of a distribution $D$ is the codistribution $D^\circ=\bigsqcup_{x\in M} D^\circ_x$, where $D_x^\circ = \{\alpha \in \T_x^*M \mid \alpha(v)=0 \text{ for all } v \in D_x\}$. The rank of the codistribution $D^\circ$ at $x\in M$ is called the {\it corank} of $D$ at $x$.

If $D$ has constant rank $k$, we write $\rk D = k$. If $D$ has constant corank $p$, we denote it by writing $\cork D = p$. If $D$ is not regular, $D^\circ$ may not be smooth. Using the usual identification $E^{\ast\ast} = E$ for a finite-dimensional vector space $E$, it follows that $(D^\circ)^\circ = D$.  The space of vector fields on a manifold $M$ taking values in a distribution $D$ on $M$ is denoted by $\Gamma(D)$. Note that we use the term `distribution' instead of `generalised distribution', as in part of the standard literature, to simplify the notation and because it does not lead to any misunderstanding.

Given a fibre bundle $\pi:E\rightarrow M$, we write $V(\pi)=\ker \T\pi$ for its vertical bundle. 
Consider the Whitney sum of $k$ copies of the tangent bundle of $M$, namely $\bigoplus^k\T M = \T M\oplus_M\overset{(k)}{\dotsb}\oplus_M \T M$, and the natural projections
\begin{equation*}
    \textstyle\pr^\alpha\colon\bigoplus^k\T M\to\T M\,, \qquad \pr_M\colon\bigoplus^k\T M\to M\,, \qquad \alpha=1,\dotsc, k\, ,
\end{equation*}
where $\pr^\alpha$ denotes the projection onto the $\alpha$-th component of the Whitney sum.

 Note that $k$-vector fields are of great use in the geometric study of systems of partial differential equations via different types of geometric structures related to field equations \cite{LSV_15,MRSV_15}. A {\it $k$-vector field} on $M$ is a section ${\bf Z}\colon M\to\bigoplus^k\T M$ of $\pr_M$, i.e. a map making the diagram
\begin{equation}
\label{diag:projection}
    \begin{tikzcd}[row sep=huge, column sep=huge]
        & \bigoplus^k\T M \arrow[d, "\pr^\alpha"]\\
        M \arrow[r, "Z_\alpha"] \arrow[ur, "{\bf Z}"] & \T M
    \end{tikzcd}
\end{equation}
commutative. Throughout the paper, $\mathfrak{X}^k(M)$ denotes the space of $k$-vector fields on $M$, while we simply write $\mathfrak{X}(M)=\mathfrak{X}^1(M)$.

Taking into account \eqref{diag:projection}, a $k$-vector field ${\bf Z}\in\mathfrak{X}^k(M)$ amounts to a family of vector fields $Z_1,\dotsc,Z_k\in\mathfrak{X}(M)$ given by $Z_\alpha = \pr^\alpha\circ{\bf Z}$ with $\alpha=1,\dotsc,k$. With this in mind, one can denote ${\bf Z} = (Z_1,\dotsc, Z_k)$. A $k$-vector field ${\bf Z}$ induces a decomposable contravariant skew-symmetric tensor field, $Z_1\wedge\dotsb\wedge Z_k$, which is a section of the bundle $\bigwedge^k\T M\to M$. Moreover, $\mathbf{Z}$ also induces a distribution $D^\mathbf{Z}$ on $M$ spanned by $Z_1, \dotsc, Z_k$.

    Given a map $\phi\colon U\subset\R^k\to M$, its {\it first prolongation} to $\bigoplus^k\T M$ is the map $\phi'\colon U\subset\R^k\to\bigoplus^k\T M$ 
    defined by
    $$ \phi'(t) = \left( \phi(t); \T_t\phi\left( \parder{}{t^1}\bigg\vert_t \right),\dotsc,\T_t\phi\left( \parder{}{t^k}\bigg\vert_t \right) \right)\,, \qquad t\in\mathbb{R}^k, $$
    where $t^1,\dotsc,t^k$ are the canonical Cartesian coordinates on $\R^k$.
 
An {\it integral section} of a $k$-vector field ${\bf Z} = (Z_1,\dotsc,Z_k)\in\mathfrak{X}^k(M)$ is a map $\phi\colon U\subset\R^k\to M$ such that
    $ \phi' = {\bf Z}\circ\phi\,, $
    that is, $\T\phi \left(\parder{}{t^\alpha}\right) = Z_\alpha\circ\phi$ for $\alpha=1,\dotsc,k$. A $k$-vector field ${\bf Z}\in\mathfrak{X}^k(M)$ is {\it integrable} if and only if every point of $M$ is in the image of an integral section of ${\bf Z}$.

Consider a $k$-vector field ${\bf Z} = (Z_1,\dotsc, Z_k)$ on $M$ with local expression $ Z_\alpha = \sum_{i=1}^nZ_\alpha^i\parder{}{x^i}\,$ for $\alpha=1,\dotsc,k$. 
Then, $\phi\colon U\subset\R^k\to M$ is an integral section of ${\bf Z}$ if, and only if, its coordinates form a solution of the system of partial differential equations
\begin{equation}
\label{eq:PDEs}
\parder{\phi^i}{t^\alpha} = Z_\alpha^i\circ\phi\,,\qquad i=1,\dotsc,n\,,\qquad \alpha=1,\dotsc,k\,.
\end{equation}

Then, ${\bf Z}$ is integrable if, and only if, $[Z_\alpha,Z_\beta] = 0$ for $\alpha<\beta=1,\dotsc,k$. These are precisely the necessary and sufficient conditions for the integrability of \eqref{eq:PDEs} (see \cite{Lee_13} for details). 
The notion of integral section of a $k$-vector field ${\bf Z}$ is stronger than the notion of integral section of the distribution $D^{\bf Z}$. Indeed,  $D^{\bf Z}$ is integrable if and only if $[Z_\alpha,Z_\beta] = \sum_{\gamma=1}^kf_{\alpha\beta}^\gamma Z_\gamma$ for certain functions $f_{\alpha\beta}^\gamma$ and $\alpha,\beta=1,\ldots,k$. On the other hand, a $k$-vector field is integrable if, and only if, $f_{\alpha\beta}^\gamma = 0$ for every $\alpha,\beta,\gamma=1,\ldots,k$ in the previous condition. It is worth noting that ${\bm Z}$ will be used to denote general $k$-vector fields. Meanwhile, other bold letters are to refer to particular types of $k$-vector fields, such as Hamiltonian or evolution $k$-contact $k$-vector fields, to be defined in forthcoming sections.

Each differential $\ell$-form taking values in $\mathbb{R}^k$, let us say $\bm\theta\in \Omega^\ell(M,\R^k)$, can be represented uniquely as $\bm\theta=\sum_{\alpha=1}^k\theta^\alpha\otimes e_\alpha$ for some differential $\ell$-forms $\theta^1,\dotsc,\theta^k \in \Omega^\ell(M)$. 
Then, $\bm\theta\in \Omega^\ell(M,\R^k)$ is {\it nondegenerate} if $\ker\bm\theta:=\bigcap^k_{\alpha=1}\ker\theta^\alpha=0$.
The contractions of a $k$-vector field ${\bm Z}$ and a vector field $Z\in \mathfrak{X}(M)$ on $M$ with $\bm \theta\in \Omega^\ell(M,\mathbb{R}^k)$ are given, respectively, by
$$
\iota_{\bm Z}\bm \theta=\sum_{\alpha=1}^k\iota_{Z_\alpha}\theta^\alpha \in \Omega^{\ell-1}(M)\,,\qquad
\iota_{Z}\bm \theta=\sum_{\alpha=1}^k\iota_{Z}\theta^\alpha\otimes e_\alpha\in \Omega^{\ell-1}(M,\mathbb{R}^k)\,.
$$
It is worth noting that we will use the above intrinsic contractions when it helps keeping notation and results clear and more concise. Meanwhile,  the summation over contractions is to be employed when it helps to understand particular calculations.
\subsection{\texorpdfstring{$k$}--contact geometry}

Let us present the basic notions of $k$-contact geometry \cite{LRS_25,GGMRR_20,GGMRR_21}. 

\begin{definition}\label{def:k-form}
A \textit{$k$-contact form on an open subset $U \subset M$} is a differential one-form on $U$ taking values in $\mathbb{R}^k$, that is $\bm{\eta} \in \Omega^1(U,\mathbb{R}^k)$, such that
 $\ker \bm{\eta} \subset \T U$ is a regular non-zero distribution of corank $k$ and $\ker \bm{\eta} \oplus \ker \dd \bm{\eta} = \T U$.

If $\bm{\eta}$ is a $k$-contact form defined on $M$, the pair $(M,\bm{\eta})$ is called a \textit{co-oriented $k$-contact manifold}, and $\ker \dd \bm{\eta}$ is called the \textit{Reeb distribution} of $(M,\bm{\eta})$. Moreover, if $\dim M = n + nk + k$ for some $n,k \in \mathbb{N}$ and there exists an integrable distribution $\mathcal{V} \subset \ker \bm{\eta}$ with $\rk \mathcal{V} = nk$, then $(M,\bm{\eta},\mathcal{V})$ is called a \textit{polarised co-oriented $k$-contact manifold}. The distribution $\mathcal{V}$ is referred to as a \textit{polarisation} of $(M,\bm{\eta})$.
\end{definition} 

The previous formalism recovers contact forms for $k=1$. Nevertheless, Definition~\ref{def:k-form} requires $\ker \bm{\eta}\neq 0$ to avoid studying the case $k=1$ on a one-dimensional manifold. In this instance, $\ker \eta=0$ is sometimes said to be maximally non-integrable \cite{GG_22}, although it is an integrable distribution. Moreover, accepting $\ker\bm\eta=0$ involves dealing with several technical nuances in $k$-contact geometry, which is undesirable. Despite this, one should note that $\ker \bm{\eta}\neq 0$ is a mere technical condition and excludes an instance of no relevance in contact or $k$-contact geometry. 

\begin{theorem}
\label{thm:k-contact-Reeb}
Let $\big(M,\bm\eta = \sum_{\alpha=1}^k \eta^\alpha \otimes e_\alpha\big)$ be a co-oriented $k$-contact manifold. There exists a unique family of  vector fields $R_1,\dotsc,R_k \in \mathfrak{X}(M)$ such that
\begin{equation}\label{eq:k-contact-Reeb} 
\inn{R_\alpha} \eta^\beta = \delta_\alpha^\beta\,,\qquad 
\inn{R_\alpha} \dd \eta^\beta = 0\,, \qquad \alpha,\beta = 1,\dotsc,k\,.
\end{equation}
The vector fields $R_1,\dotsc,R_k$ commute with each other, that is,
\(
[R_\alpha, R_\beta] = 0
\)
for $\alpha,\beta = 1,\dotsc,k$. Moreover, these vector fields form a basis of the Reeb distribution, namely
\(
\ker \dd \bm\eta = \langle R_1,\dotsc,R_k \rangle .
\)
\end{theorem}

Since the distribution $\ker \dd \bm\eta$ is the intersection of the kernels of closed two-forms and has constant rank $k$ by assumption, it is integrable. 

\begin{definition}
Given a $k$-contact manifold $(M,\bm\eta)$, the \textit{Reeb $k$-vector field} of $(M,\bm\eta)$ is the integrable $k$-vector field
\(
\mathbf{R} = (R_1,\dotsc,R_k)
\) 
on $M$ whose components are described in Theorem~\ref{thm:k-contact-Reeb}. The vector fields $R_1,\dotsc,R_k$ are called the \textit{Reeb vector fields} of the co-oriented $k$-contact manifold $(M,\bm\eta)$.
\end{definition}

The following example shows one of the main $k$-contact manifolds to be used hereafter. It is also the canonical local model for polarised $k$-contact manifolds as they are all locally of this type \cite{LRS_25}. Moreover, it is related to first-order jet bundles and many examples in physics \cite{LRS_25}. 

\begin{example} 
\label{ex:canonical-k-contact-structure}
\rm
The manifold $\mathcal{J}_{Q,k} = \bigoplus^k \cT Q \times \mathbb{R}^k$ carries a canonical $k$-contact manifold structure defined by the $k$-contact form
\(
\bm\eta_{Q,k} = \sum_{\alpha=1}^k (\dd z^\alpha - \theta^\alpha) \otimes e_\alpha.
\)
Here, $\{z^1,\dotsc,z^k\}$ are the pull-back of linear coordinates on $\mathbb{R}^k$ to $\mathcal{J}_{Q,k}$ via its canonical projection onto $\mathbb{R}^k$.  These coordinates are hereafter called {\it dissipative variables} due to their relation to dissipation in field theories \cite{Riv_22}. Note that $\mathcal{J}_{Q,k}$ also admits a natural projection onto $\bigoplus^k\cT Q$. Moreover, each $\theta^\alpha$ is the pull-back to $\mathcal{J}_{Q,k}$ of the Liouville one-form $\theta$ on the $\alpha$-th copy of the cotangent bundle $\cT Q$ via the natural projection $\mathcal{J}_{Q,k}\to \bigoplus^k\cT Q\stackrel{{\rm pr}_\alpha}{\to} \cT Q$. Additionally, $\mathcal{J}_{Q,k}$ admits a vertical distribution $\mathcal{V}$ relative to the natural projection onto $Q \times \mathbb{R}^k$   of rank $k \cdot \dim Q$ contained in $\ker \bm\eta_{Q,k}$. Therefore, $\left(\mathcal{J}_{Q,k}, \bm\eta_{Q,k}, \mathcal{V}\right)$ is a polarised co-oriented $k$-contact manifold. From now on, we assume $Q$ to be $n$-dimensional.

Any choice of local coordinates $\{q^1,\dotsc,q^n\}$ on $Q$, together with linear coordinates on $\mathbb{R}^k$, induces a natural coordinate system $\{q^i, p_i^\alpha, z^\alpha\}$ on $\mathcal{J}_{Q,k}$, with $\alpha = 1,\dotsc,k$. In these coordinates,
\[
{\bm \eta}_{Q,k}= \sum_{\alpha=1}^k\left(\dd z^\alpha - \sum_{i=1}^n p_i^\alpha \dd q^i\right)\otimes e_\alpha\,,\qquad
\dd {\bm \eta}_{Q,k}= \sum_{\alpha=1}^k\left(\sum_{i=1}^n\dd q^i \wedge \dd p_i^\alpha\right)\otimes e_\alpha\,,\qquad
R_\beta = \parder{}{z^\beta}\,,
\]
for \(\beta = 1,\dotsc,k\), and
\[
\ker \bm\eta_{Q,k} =
\left\langle
\parder{}{p_i^\alpha},
\parder{}{q^i} + \sum_{\beta=1}^k p_i^\beta \parder{}{z^\beta}
\right\rangle_{\stackrel{\alpha=1,\ldots,k}{i = 1,\dotsc,n}}\!\!\!\!\!\!\!\!\!\!\!\!,\quad
\ker \dd \bm\eta_{Q,k} =
\left\langle
R_1,\dotsc,R_k
\right\rangle,
\quad
\mathcal{V}_{Q,k}=\left\langle\parder{}{p^\alpha_i}\right\rangle_{\substack{i=1,\dotsc,n\\\alpha=1,\dotsc,k}} .
\]
\end{example}

\begin{example}\label{ex:contactification-k-symplectic-manifold}
Let us study the $k$-contactification of an exact $k$-symplectic manifold, namely the $k$-contact manifold associated with a $k$-symplectic manifold $(P,\bm \omega)$, where $\bm \omega\in \Omega^2(P,\mathbb{R}^k)$ is a closed two-form taking values in $\mathbb{R}^k$ with $\ker \bm\omega=0$. More specifically, consider an exact $k$-symplectic manifold $(P,\bm\omega = \sum_{\alpha=1}^k \omega^\alpha \otimes e_\alpha)$ with $\bm \omega  = -\dd\bm \theta$, and the product manifold $M = P \times \mathbb{R}^k$. Let $\{z^1,\dotsc,z^k\}$ denote the pull-back to $M$ of linear coordinates on $\mathbb{R}^k$, and let $\theta_M^\alpha$ be the pull-back of $\theta^\alpha$ from $P$ to $M$ in the canonical manner.

Define the $\mathbb{R}^k$-valued one-form
\[
\bm\eta_{\bm \theta} = \sum_{\alpha=1}^k ( \dd z^\alpha - \theta_M^\alpha)\otimes e_\alpha \in \Omega^1(M,\mathbb{R}^k).
\]
Then, $(M,\bm\eta_{\bm \theta})$ becomes a co-oriented $k$-contact manifold. Indeed, $\ker \bm\eta_{\bm \theta}$ has corank $k$ and is non-zero, while $\dd \bm\eta_{\bm \theta} = -\dd \bm\theta_M$ for $\bm \theta_M=\sum_{\alpha=1}^k\theta^\alpha_M\otimes e_\alpha$. Then,
\[
\ker \dd \bm\eta_{\bm \theta} = \left\langle \parder{}{z^1},\dotsc,\parder{}{z^k} \right\rangle
\]
has rank $k$, since $(P,\bm\omega)$ is $k$-symplectic. Hence, $\bm\eta_{\bm \theta}$ becomes a $k$-contact form on $M$.

The canonical $k$-contact form $\bm\eta_{Q,k}$ described in Example~\ref{ex:canonical-k-contact-structure} is precisely the $k$-contactification of the exact $k$-symplectic manifold $(P = \bigoplus^k \cT Q, \bm\omega_{Q,k}=\sum_{\alpha=1}^k \omega_{Q}^\alpha \otimes e_\alpha)$, where $\omega_{Q}^\alpha$ denotes the pull-back to $P$ of the canonical symplectic form from the $\alpha$-th copy of $\cT Q$ via the natural projection from $\bigoplus^k\cT Q$ onto its $\alpha$-th copy of $\cT Q$.  Note that one can construct a canonical one-form $\bm \theta_Q=\sum_{\alpha=1}^k\theta_Q^\alpha\otimes e_\alpha\in \Omega^1(P,\mathbb{R}^k)$ from the pull-back $\theta^\alpha_Q$ to $P$ of the canonical Liouville form $\theta_Q$ on the $\alpha$-th copy of $\T^*Q$ in $P$. Moreover, $\bm\omega_{Q,k}=-\dd\bm\theta_Q$.

Let us analyse the above construction for $k=1$. Let $P$ be an exact symplectic manifold and consider $M = P \times \mathbb{R}$ endowed with the one-form $\eta = \dd z - \theta$. In this case, $\ker \dd \eta = \ker (-\dd \theta)$ has rank one since $\dd \theta$ is the pull-back to $M$ of a symplectic form on $P$. Under these assumptions, $(M,\eta)$ is a (co-oriented) contact manifold.
\end{example}

\begin{theorem}[$k$-contact Darboux Theorem \cite{Riv_22}]
\label{thm:kcontDarboux}
Let $(M,\bm\eta,\mathcal{V})$ be a polarised co-oriented $k$-contact manifold. Around every point of $M$, there exist local coordinates $\{q^i,p_i^\alpha,z^\alpha\}$, with $1 \leq i \leq n$ and $1 \leq \alpha \leq k$, such that
\[
\bm \eta = \sum_{\alpha=1}^k \left(\dd z^\alpha - \sum_{i=1}^np_i^\alpha \dd q^i\right)\otimes e_\alpha\,,
\quad\!\!
\ker \dd \bm\eta = \left\langle R_1 = \parder{}{z^1} ,\dotsc,R_k = \parder{}{z^k} \right\rangle,
\quad\!\!
\mathcal{V} = \left\langle \parder{}{p_i^\alpha} \right\rangle_{\substack{i=1,\dotsc,n\\\alpha=1,\dotsc,k}}\!.
\]
These coordinates are called \textit{Darboux coordinates} of $(M,\bm\eta,\mathcal{V})$.
\end{theorem}

Theorem~\ref{thm:kcontDarboux} allows us to regard the manifold $\mathcal{J}_{Q,k}$ introduced in Example~\ref{ex:canonical-k-contact-structure} as the canonical local model of polarised co-oriented $k$-contact manifolds (see \cite{LRS_25} for details). Moreover, every $k$-contact manifold obtained as the $k$-contactification of $\bigoplus^k \cT Q$ endowed with its canonical exact $k$-symplectic structure (see Example~\ref{ex:contactification-k-symplectic-manifold} for $P=\bigoplus^k \cT Q$) admits Darboux coordinates.

Let us define a morphism that is very useful to study co-oriented $k$-contact manifolds and their associated Hamiltonian systems, as shown in following parts of this work. Given a co-oriented $k$-contact manifold $(M,\bm\eta)$, we can define a vector bundle morphism over $M$ of the form $\chi\colon \bigoplus^k\T M \to \cT M\times \mathbb{R}$ by
\[
\chi({\bm Z}) = (\iota_{\bm Z} \dd \bm\eta,\iota_{\bm Z}\bm\eta).
\]  
This morphism is injective if, and only if, $k=1$. In this case, a co-oriented one-contact manifold $(M,\bm\eta)$ is a contact manifold. Indeed, this fact will play a fundamental role in the study of $k$-contact Hamiltonian systems, since it implies that the $k$-contact Hamilton--De Donder--Weyl equations for $k$-vector fields admit solutions for every Hamiltonian function, but these solutions are not unique when $k>1$ (see Proposition~\ref{prop:k-contact-HdDW-have-solutions} below and \cite{HLM_26}).

\begin{proposition}\label{prop:kernel-flat-k-contact-Darboux}
Let \(\mathcal{J}_{Q,k}=\bigoplus^k \T^*Q\times \mathbb{R}^k\) be endowed with Darboux coordinates \((q^i,p_i^\alpha,z^\alpha)\). If a $k$-tangent vector \(\mathbf Z=(Z_1,\dots,Z_k)\) at a point of \( \mathcal{J}_{Q,k}\) has the components
\(
Z_\alpha=\sum_{i=1}^n Z_\alpha^i\frac{\partial}{\partial q^i}+\sum_{i=1}^n\sum_{\beta=1}^k (Z_{\alpha })_i^{\beta}\frac{\partial}{\partial p_i^\beta}+\sum_{\beta=1}^k Z_{\alpha}^{\beta}\frac{\partial}{\partial z^\beta},
\) for $\alpha=1,\ldots,k$, 
then 
\[
\ker\chi=\left\{\mathbf Z=(Z_1,\dots,Z_k)\in {\textstyle\bigoplus^k} V(\pi_Q) \ \Big\vert\  \sum_{\alpha=1}^k(Z_{\alpha })_i^{\alpha}=0\,,\ \ \sum_{\alpha=1}^k Z_\alpha^\alpha=0\,,\,\ i=1,\ldots,n\right\},
\]
where \(\pi_Q\colon\bigoplus^k\T^*Q\times\mathbb R^k\to Q\) is the canonical projection and \(V(\pi_Q)=\ker \T\pi_Q\). In particular,
\(
\ker\chi\subset \bigoplus^k V(\pi_Q),
\)
and \(\dim \ker\chi=(n+1)(k^2-1)\).
\end{proposition}

\begin{proof}
For the canonical $k$-contact form $\bm\eta_{Q,k}$, the components of the morphism $\chi$ read 
\[
\iota_{\mathbf Z}\dd\bm\eta_{Q,k}=\sum_{i=1}^n\sum_{\alpha=1}^k(Z_\alpha^i\,\dd p_i^\alpha-(Z_{\alpha})_i^{\alpha}\,\dd q^i)\,,\qquad
\iota_{\mathbf Z}\bm\eta_{Q,k}=\sum_{\alpha=1}^k(Z_\alpha^\alpha-\sum_{i=1}^np_i^\alpha Z_\alpha^i)\,.
\]
By linear independence of the forms \(\dd p_i^\alpha\) and \(\dd q^i\), the $k$-vector field \(\mathbf Z\) takes values in \(\ker\chi\) if, and only if,
\(
Z_\beta^i=0,\ \sum_{\alpha=1}^kZ_{\alpha i}^{\alpha}=0,\ \sum_{\alpha=1}^k Z_{\alpha}^{\alpha}=0,
\) for $\beta=1,\ldots,k, i=1,\ldots,n$. 
The conditions \(Z_\beta^i=0\) mean exactly that each \(Z_\alpha\) is vertical with respect to \(\pi_Q\), and \(\mathbf Z\) belongs to \(\bigoplus^kV(\pi_Q)\). The dimension count is immediate: the coefficients \(Z_{\alpha i}^{\beta}\) and \(Z_\alpha^\beta\) give \(nk^2+k^2\) parameters, and the above relations impose \(n+1\) independent conditions.
\end{proof}

\subsection{On \texorpdfstring{$k$}{}-contact Hamiltonian systems}
\label{sub:k-contact-Hamiltonian-systems}

After introducing the geometric framework of $k$-contact geometry, we now address its Hamiltonian formulation of field theories. 

\begin{definition}
The {\it $k$-contact Hamilton--De Donder--Weyl equations} for a $k$-vector field
${\bf X} = (X_\alpha) \in \mathfrak{X}^k(M)$ associated with $h\in \Cinfty(M)$ and $(M,\bm \eta)$ are
\begin{equation}\label{eq:k-contact-HdDW-fields1}
\iota_{\bm X}\dd\bm\eta=\sum_{\alpha=1}^k \inn{X_\alpha} \dd \eta^\alpha
= \dd h - \sum_{\alpha=1}^k (\Lie_{R_\alpha} h)\, \eta^\alpha\,,\qquad
\iota_{\bm X}\bm\eta=\sum_{\alpha=1}^k \inn{X_\alpha} \eta^\alpha
= -h \,.
\end{equation}
A $k$-vector field ${\bm X}$ satisfying these equations for some $h$ is called an
\textit{${\bm \eta}$-Hamiltonian $k$-vector field}. Meanwhile, a triple $(M,\bm \eta,h)$ is called an {\it $\bm\eta$-Hamiltonian system} and $h \in \Cinfty(M)$ is its \textit{Hamiltonian function}.
\end{definition}

The space of all $\bm \eta$-Hamiltonian $k$-vector fields is a vector space. The $\bm \eta$-Hamiltonian $k$-vector fields related to the zero function are called {\it $\bm\eta$-gauge $k$-vector fields}. They also span a vector space. From now on, $(M,\bm\eta)$ always stands for a $k$-contact manifold, and $(M,\bm\eta,h)$ is an ${\bm \eta}$-Hamiltonian system. 

It is worth stressing that every $(M,\bm \eta,h)$ induces an ${\bm \eta}$-Hamiltonian $k$-vector field \cite{HLM_26,Riv_22}. Indeed, one has the following result (see \cite[Theorem 3.6]{HLM_26} for more details).

\begin{proposition}\label{prop:k-contact-HdDW-have-solutions}
The $k$-contact Hamilton--De Donder--Weyl equations
\eqref{eq:k-contact-HdDW-fields1} admit solutions for every $h\in \Cinfty(M)$. These solutions are not unique when $k>1$.
\end{proposition}

Let ${\bf X} = (X_1,\dotsc,X_k) \in \mathfrak{X}^k(M)$ be a $k$-vector field with local expression in Darboux coordinates
\[
X_\alpha
=
\sum_{i=1}^n (X_\alpha)^i \parder{}{q^i}
\;+\;
\sum_{\beta=1}^k \sum_{i=1}^n (X_\alpha)^\beta_i \parder{}{p_i^\beta}
\;+\;
\sum_{\beta=1}^k (X_\alpha)^\beta \parder{}{z^\beta}\,,
\qquad \alpha=1,\dotsc,k\,.
\]
Note that the existence of Darboux coordinates for some $(M,\bm\eta)$, namely putting $\bm\eta$ locally of the form ${\bm \eta}=\sum_{\alpha=1}^k (\dd z^\alpha -\sum_{i=1}^np^\alpha_i\dd q^i)\otimes e_\alpha$, yields that $(M,\bm\eta)$ induces, locally, a polarised $k$-contact manifold. Then, equations \eqref{eq:k-contact-HdDW-fields1} imply
\begin{equation}\label{eq:k-contact-HdDW-fields-Darboux-coordinates1}
\begin{dcases}
(X_\beta)^i = \parder{h}{p_i^\beta}\,,\\
\sum_{\alpha=1}^k (X_\alpha)^\alpha_i
= -\left( \parder{h}{q^i} + \sum_{\alpha=1}^k p_i^\alpha \parder{h}{z^\alpha} \right)\,,\\
\sum_{\alpha=1}^k (X_\alpha)^\alpha
= \sum_{\alpha=1}^k \sum_{j=1}^n p_j^\alpha \parder{h}{p_j^\alpha} - h\,,
\end{dcases}\qquad i=1,\dotsc,n\,,\qquad \beta=1,\dotsc,k\,.
\end{equation}

The components $(X_\alpha)^\beta$ and $(X_\alpha)^\beta_i$ with $\alpha \neq \beta$ do not appear explicitly in equations \eqref{eq:k-contact-HdDW-fields-Darboux-coordinates1}. Hence, given an $\bm \eta$-Hamiltonian $k$-vector field associated with $(M,\bm \eta,h)$, these components may be chosen arbitrarily without affecting the ${\bm \eta}$-Hamiltonian character of ${\bf X}$. Nevertheless, such choices may be relevant for ensuring the integrability of ${\bf X}$.

\begin{definition}\label{dfn:k-contact-Hamiltonian-system}
Given a map $\psi \colon D \subset \mathbb{R}^k \to M$, the \textit{$k$-contact Hamilton--De Donder--Weyl equations} related to the $\bm \eta$-Hamiltonian system $(M,\bm \eta,h)$ are
\begin{equation}\label{eq:k-contact-HdDW}
\sum_{\alpha=1}^k \inn{\psi'_\alpha} \dd \eta^\alpha
= \left( \dd h - \sum_{\alpha=1}^k (\Lie_{R_\alpha} h)\, \eta^\alpha \right) \circ \psi\,,\qquad
\sum_{\alpha=1}^k \inn{\psi'_\alpha} \eta^\alpha
= - h \circ \psi\,.
\end{equation}
\end{definition}

Note that a solution for the $k$-contact HdDW equations \eqref{eq:k-contact-HdDW} gives rise to a $k$-vector field on ${\rm Im}\,\psi$, which is naturally diffeomorphic to $D$.

In \textit{Darboux coordinates} for a polarised $k$-contact manifold $(M,\bm \eta,\mathcal{V})$, a section $\psi\colon\mathbb{R}^k\rightarrow M$ takes the form $\psi(t) = (q^i(t),p_i^\alpha(t), z^\alpha(t))$, while equations \eqref{eq:k-contact-HdDW} read
\begin{equation}\label{eq:k-contact-HdDW-Darboux-coordinates}
\begin{dcases}
\parder{q^i}{t^\beta} = \parder{h}{p_i^\beta} \circ \psi\,,\\
\sum_{\alpha=1}^k \parder{p_i^\alpha}{t^\alpha}
= -\left( \parder{h}{q^i} + \sum_{\alpha=1}^k p_i^\alpha \parder{h}{z^\alpha} \right) \circ \psi\,,\qquad i=1,\dotsc,n \qquad \beta=1,\dotsc,k, \\
\sum_{\alpha=1}^k \parder{z^\alpha}{t^\alpha}
= \left( \sum_{\alpha=1}^k \sum_{j=1}^n p_j^\alpha \parder{h}{p_j^\alpha} - h \right) \circ \psi\,.
\end{dcases}
\end{equation}

Furthermore, the existence of solutions of equations \eqref{eq:k-contact-HdDW-fields1}
does not imply the integrability of the associated $k$-vector fields. Indeed, as in the $k$-symplectic case, equations \eqref{eq:k-contact-HdDW} and
\eqref{eq:k-contact-HdDW-fields1} are not fully equivalent, since a $k$-vector field solving \eqref{eq:k-contact-HdDW-fields1} does not need to give rise to a solution
of \eqref{eq:k-contact-HdDW}. For instance, this happens when $X_1,\dotsc,X_k$ do not commute between themselves at any point of $M$ (see \cite{GGMRR_20} for further details).  The second line of equations in \eqref{eq:k-contact-HdDW-Darboux-coordinates} are called the \textit{balance equations} of the $k$-contact Hamilton--De Donder--Weyl equations, since they are related to the balance of momenta and dissipation in field theories \cite{Riv_22}.

\begin{definition} \label{def:IsotropicDistribution}
The {\it $k$-contact orthogonal} of a vector subspace $E_x\subset \ker \bm \eta_x$, for $x\in M$ and a $k$-contact manifold $(M,\bm \eta)$, is the vector subspace
\begin{equation}\label{eq:kConOrth}
E_x^{\perp_{\bm\eta}}:=\{v_x\in \ker \bm \eta_x \mid \dd\bm\eta(v_x,w_x)=0\,,\ \forall w_x\in E_x\}\,.
\end{equation}
Then, $E_x$ is {\it isotropic} if $E_x\subset E_x^{\perp_{\bm\eta}}$. Meanwhile, $E_x$ is {\it Legendrian} if $E_x$ is isotropic and it admits a complement $F_x$ such that $F_x\oplus E_x=\ker \bm \eta_x$ and $\dd\bm\eta\vert_{F_x\times F_x}=0$. A distribution $E\subset \T M$ is {\it isotropic} (resp. {\it Legendrian}) if every $E_x$, with $x\in M$, is isotropic (resp. Legendrian). A submanifold of $M$ is called isotropic (resp. Legendrian) if its tangent space at each point is isotropic (resp. Legendrian). A subspace $W_x\subset \T_xM$ is called {\it coisotropic} if $(W_x\cap \ker \bm\eta_x)^{\perp_{\bm \eta}} \subset W_x\cap \ker\bm\eta_x$ and {\it maximally coisotropic} if $(W_x\cap \ker \bm\eta_x)^{\perp_{\bm \eta}} = W_x\cap \ker\bm\eta_x$. 
\end{definition}

The following lemma is necessary so as to stress why our results retrieve natural Hamilton--Jacobi equations for contact Hamiltonian systems when $k=1$. 
\begin{lemma}\label{lem:coisotropic-max-coisotropic-contact-case}
Let $(M,\eta)$ be a contact manifold with $\dim M=2n+1$ and let $N\subset M$ be a submanifold of dimension $n+1$ such that $\restr{\eta}{\T N}$ is nowhere vanishing. Then, $N$ is coisotropic if and only if it is maximally coisotropic.
\end{lemma}

\begin{proof}
Fix $x\in N$ and set $W_x\coloneqq \T_xN\cap\ker\eta_x$. Since $\restr{\eta_x}{\T_xN}$ is a nonzero linear form on the $(n+1)$-dimensional space $\T_xN$, one has that $\T_xN+\ker \eta_x$ has dimension $2n+1$ and $\dim W_x=n$. On the other hand, $(\ker\eta_x,\dd\eta_x|_{\ker\eta_x})$ is a symplectic linear space of dimension $2n$, so $\dim W_x^{\perp}=2n-\dim W_x=n$. If $N$ is coisotropic, then $W_x^{\perp_\eta}\subset W_x$, and since both spaces have dimension $n$, they must be equal. The converse is immediate.
\end{proof}

\begin{example}
Let us illustrate our theory by studying a damped vibrating membrane on an elastic foundation. Consider coordinates $(t,x,y)$ on $\mathbb{R}^3$ and let $Q=\mathbb{R}$. The phase space is \(\mathcal{J}_{\mathbb{R},3}=\big(\cT\mathbb{R}\oplus \cT\mathbb{R}\oplus \cT\mathbb{R}\big)\times\mathbb{R}^3\) with coordinates $(u,p^t,p^x,p^y,z^t,z^x,z^y)$. The variable $u$ represents the transverse displacement of the membrane, $\lambda\geq 0$ is a damping coefficient, and $\kappa\geq 0$ is the stiffness of a linear elastic foundation. Consider the Hamiltonian
\begin{equation}\label{eq:Hamiltonian-damped-membrane}
h(u,p^t,p^x,p^y,z^t,z^x,z^y)=\frac12\Big((p^t)^2-c^{-2}(p^x)^2-c^{-2}(p^y)^2\Big)+\frac{\kappa}{2}u^2+\lambda z^t ,
\end{equation} 
where \(c>0\) is the propagation speed. The Hamilton--De Donder--Weyl equations for the associated $\bm\eta_{\mathbb{R},3}$-Hamiltonian three-vector field read
\[
\begin{dcases}
\parder{u}{t}=p^t,\\
\parder{u}{x}=-c^{-2}p^x,\\
\parder{u}{y}=-c^{-2}p^y,\\
\parder{p^t}{t}+\parder{p^x}{x}+\parder{p^y}{y}=-\kappa u-\lambda p^t,\\
\parder{z^t}{t}+\parder{z^x}{x}+\parder{z^y}{y}=\frac12\Big((p^t)^2-c^{-2}(p^x)^2-c^{-2}(p^y)^2\Big)-\frac{\kappa}{2}u^2-\lambda z^t .
\end{dcases}
\]
Using the first three equations in the fourth one, we obtain
\begin{equation}\label{eq:damped-membrane-foundation}
u_{tt}-c^2(u_{xx}+u_{yy})+\lambda u_t+\kappa u=0,
\end{equation}
that is,
\[
u_{tt}-c^2\Delta u+\lambda u_t+\kappa u=0,
\qquad
\Delta=\partial^2_{x}+\partial^2_{y}.
\]
Equation \eqref{eq:damped-membrane-foundation} describes a damped vibrating membrane on a linear elastic foundation. When $\kappa=0$, one recovers the standard damped membrane equation.
\end{example}

\subsection{Evolution \texorpdfstring{$k$}{}-contact \texorpdfstring{$k$}{}-vector fields}
\label{subsec:evolution-k-contact-k-vector-fields}

Let us develop a generalisation of contact evolution vector fields \cite{LL_19} to the realm of $k$-contact geometry. This will help us develop a theory of Hamilton--Jacobi equations in a new realm for $k$-contact manifolds. 

On a co-oriented one-contact manifold $(M,\eta)$, an evolution vector field is a vector field $E$ on $M$ that satisfies the equation
\[
\iota_E\dd\eta = \dd h - (\Lie_R h)\eta\,, \qquad \iota_E\eta = 0\,,
\]
for a certain $h\in \Cinfty(M)$. In fact, $E$ is uniquely determined from $h$ because $\ker \eta\oplus\ker \d\eta=\T M$. 
Meanwhile, for the $k$-contact setting with $k>1$, an evolution $k$-contact $k$-vector field is defined as follows. 

\begin{definition}
An {\it evolution $k$-vector field} for a co-oriented $k$-contact manifold $(M,\bm \eta)$ is a $k$-vector field ${\bm E}$ on $M$ satisfying
\begin{equation}\label{eq:k-contact-evolution-HdDW-fields}
\iota_{\bm E}\dd\bm\eta=\sum_{\alpha=1}^k\inn{E_\alpha}\dd\eta^\alpha = \dd h - \sum_{\alpha=1}^k(\Lie_{R_\alpha}h)\eta^\alpha\,,\qquad
\iota_{\bm E}\bm \eta=\sum_{\alpha=1}^k\inn{E_\alpha}\eta^\alpha = 0\,,
\end{equation}
for some function $h\in \Cinfty (M)$. The above equations for a fixed $h$ are called the {\it evolution $k$-contact Hamilton--De Donder--Weyl equations} for $k$-contact vector fields. We also call $(M,\bm \eta,h)_{\bm E}$ an {\it evolution ${\bm \eta}$-Hamiltonian system} and ${\bm E}$ an $\bm \eta$-evolution $k$-vector field.
\end{definition}

The main difference between the previous system \eqref{eq:k-contact-evolution-HdDW-fields} and the $k$-contact HdDW equations for ${\bm \eta}$-Hamiltonian $k$-vector fields relies on the fact that $\iota_{\bm E}\bm \eta=0$ instead of $\iota_{\bm E}\bm \eta=-h$. The following result is a simple modification of \cite[Proposition 3.5]{HLM_26} to the evolution $k$-contact realm.

\begin{proposition}\label{prop:k-contact-evolution-HdDW-have-solutions}
Evolution $k$-contact Hamilton--De Donder--Weyl equations \eqref{eq:k-contact-evolution-HdDW-fields} admit solutions. They are not unique if $k > 1$.
\end{proposition}

\begin{proof}
Consider the mapping
$
\textstyle\rho\colon {\bm E}\in\bigoplus^k\T M\longmapsto \iota_{\bm E}\dd\bm \eta\in \cT M\,.
$
One has that $\iota_{R_\alpha}\rho({\bm E})=0$ for every $R_1,\dotsc,R_k$. Hence, ${\rm Im}\,\rho\subset (\ker \dd \bm \eta)^\circ$. Moreover, ${\rm Im}\,\rho=(\ker \dd \bm \eta)^\circ$, otherwise, $({\rm Im}\,\rho)^\circ$ would have dimension bigger than $k$ and $({\rm Im}\,\rho)^\circ\cap \ker \bm \eta\neq 0$. But  ${\rm Im} \rho)^\circ\subset \ker\d\bm\eta$ and
$$
({\rm Im}\rho)^\circ\cap\ker \bm \eta\subset \ker\dd\bm\eta\cap\ker\bm\eta=0\,,
$$ 
which is a contradiction. Therefore, for every $h$ there exists some $k$-vector field ${\bm E}$ such that $\rho({\bm E})=\dd h-\sum_{\alpha=1}^k(\Lie_{R_\alpha}h)\eta^\alpha$. Then, $(E_1-(\iota_{\bm E}\bm \eta)R_1,\dotsc,E_k)$ is an ${\bm \eta}$-evolution $k$-vector field for $h$. 
    
If $k>1$ and ${\bm E}$ is a solution to evolution $k$-contact HdDW equations, one has that functions $f_1,\dotsc,f_k\in \Cinfty(M)$ such that $\sum_{\alpha=1}^kf_\alpha=0$ yield that ${\bm E}+(f_1R_1,\dotsc,f_kR_k)$ is another solution to evolution $k$-contact HdDW equations. 
\end{proof}


Consider a $k$-vector field ${\bf E} = (E_1,\dotsc,E_k)\in\mathfrak{X}^k(M)$. Its local expression in \textit{Darboux coordinates} in a co-oriented polarised $k$-contact manifold $(M,\bm \eta,\mathcal{V})$ reads 
\[
E_\alpha
=
\sum_{i=1}^n (E_\alpha)^i \frac{\partial}{\partial q^i}
\;+\;
\sum_{\beta=1}^k \sum_{i=1}^n (E_\alpha)^\beta_i \frac{\partial}{\partial p_i^\beta}
\;+\;
\sum_{\beta=1}^k (E_\alpha)^\beta \frac{\partial}{\partial z^\beta}\,,
\qquad \alpha=1,\dotsc,k\,.
\]
Moreover, equation \eqref{eq:k-contact-evolution-HdDW-fields} implies that
\begin{equation}\label{eq:k-contact-HdDW-fields-Darboux-coordinates}
\begin{dcases}
(E_\beta)^i = \parder{h}{p_i^\beta}\,,\\
\sum_{\alpha=1}^k  (E_\alpha)^\alpha_i = -\left( \parder{h}{q^i} + \sum_{\alpha=1}^k  p_i^\alpha\parder{h}{z^\alpha} \right)\,,\\
\sum_{\alpha=1}^k  (E_\alpha)^\alpha = \sum_{j=1}^n  \sum_{\alpha=1}^k  p_j^\alpha\parder{h}{p_j^\alpha}\,.
\end{dcases}\qquad i=1,\dotsc,n\,,\qquad \beta=1,\dotsc,k\,.
\end{equation}
There is therefore a significant difference relative to the standard $k$-contact Hamiltonian case: the sums of coefficients of ${\bm E}$ in the left-hand side of the above equations are always proportional to linear combinations of first-order derivatives of $h$, while $k$-contact Hamiltonian $k$-vector fields are not.

The following observation is immediate from the local expressions above, so we only record the statement.

\begin{lemma}
Let $\mathbf{E}=(E_1,\dots,E_k)$ be an $\bm \eta$-evolution $k$-vector field. Then, $\mathbf{E}'=(E'_1,\dots,E'_k)$ with
\[
E'_\alpha
=
E_\alpha
-
\sum_{\beta\neq\alpha}(E_\alpha)^\beta R_\beta\,,
\qquad \alpha=1,\dots,k\,,
\]
satisfies the same $k$-contact Hamilton--De Donder--Weyl equations as $\mathbf{E}$. In particular, ${\bm E'}$ is again an $\bm \eta$-evolution $k$-vector field for the same Hamiltonian function.
\end{lemma}

More frequently, one is concerned with the integral sections of an evolution $k$-vector field. In this sense, the interest is focused on the system of PDEs given below.

\begin{definition}\label{dfn:evolution-k-contact-Hamiltonian-system}
Given a $k$-contact Hamiltonian system $(M,\bm\eta,h)$, the \textit{evolution $k$-contact Hamilton--De Donder--Weyl (HdDW) equations} are
\begin{equation}\label{eq:k-contact-HdDW2}
\begin{dcases}
\sum_{\alpha=1}^k\inn{\psi_\alpha'}\dd\eta^\alpha = \left( \dd h - \sum_{\alpha=1}^k(\Lie_{R_\alpha}h)\eta^\alpha \right)\circ\psi\,,\\
\sum_{\alpha=1}^k \inn{\psi_\alpha'}\eta^\alpha = 0\,,
\end{dcases}
\end{equation}
where solutions are maps $\psi\colon D\subset\R^k\to M$, where $D$ is an open subset of $\R^k$. 
\end{definition}

In \textit{Darboux coordinates} for some $(M,\bm\eta,\mathcal{V})$, one may write $\psi(t) = (q^i(t),p_i^\alpha(t),z^\alpha(t))$ and equations \eqref{eq:k-contact-HdDW2} read
\begin{equation}\label{eq:k-contact-HdDW-Darboux-coordinates2}
\begin{dcases}
\parder{q^i}{t^\beta} = \parder{h}{p_i^\beta}\circ\psi\,,\\
\sum_{\alpha=1}^k\parder{p^\alpha_i}{t^\alpha} = -\left( \parder{h}{q^i} +\sum_{\alpha=1}^k p_i^\alpha\parder{h}{z^\alpha} \right)\circ\psi\,,\\
\sum_{\alpha=1}^k\parder{z^\alpha}{t^\alpha} =\sum_{\alpha=1}^k \sum_{j=1}^n p_j^\alpha\parder{h}{p_j^\alpha}\circ\psi\,,
\end{dcases}
\end{equation}
for $ i=1,\dotsc,n$ and $\beta=1,\dotsc,k$.

 \begin{lemma}\label{lem:canonical-evolution-lift}
Let \((P,\bm\omega=\sum_{\alpha=1}^k\omega^\alpha\otimes e_\alpha)\) be an exact \(k\)-symplectic manifold, with \(\bm \omega=-d\bm \theta\), and let \((M=P\times\mathbb R^k,\bm\eta=\sum_{\alpha=1}^k\eta^\alpha\otimes e_\alpha)\) be its induced \(k\)-contactification. Let \(\rho\colon M\to P\) be the canonical projection and \(H\in \Cinfty(P)\), set \(h\coloneqq H\circ\rho\), and let \(\mathbf X=(X_1,\dots,X_k)\) be a \(k\)-symplectic Hamiltonian \(k\)-vector field for \(H\), i.e.
\(
\sum_{\alpha=1}^k\iota_{X_\alpha}\omega^\alpha=\dd H.
\)
Then, there exists a canonical \(k\)-vector field \(\mathbf E=(E_1,\dots,E_k)\) on \(M\) such that \(\rho_*E_\alpha=X_\alpha\), \(\iota_{E_\alpha}\eta^\alpha=0\) and $\iota_{E_\alpha}\dd z^\beta =0$ for every \(\alpha,\beta=1,\dots,k\) with $\alpha\neq \beta$. It is given by
\begin{equation}\label{eq:Lift}
E_\alpha=\widetilde X_\alpha+\rho^*(\theta^\alpha(X_\alpha))\frac{\partial}{\partial z^\alpha}\,,\qquad \alpha=1,\dots,k\,,
\end{equation}
where \(\widetilde X_\alpha\) denotes the lift of \(X_\alpha\) to \(M\) tangent to the \(P\)-factor. Moreover, \(\mathbf E\) is an ${\bm \eta}$-evolution \(k\)-contact \(k\)-vector field for \(h\).
\end{lemma}

\begin{proof}
Since \(\eta^\alpha = \d z^\alpha-\rho^*\theta^\alpha\), one has \(\d\eta^\alpha=\rho^*\omega^\alpha\). Moreover, \(\partial/\partial z^\beta\in\ker \d\eta^\alpha\) for every \(\alpha,\beta\). Hence, for \eqref{eq:Lift}, 
we obtain
\[
\sum_{\alpha=1}^k\iota_{E_\alpha}\d\eta^\alpha=\sum_{\alpha=1}^k\iota_{\widetilde X_\alpha}\rho^*\omega^\alpha=\rho^*\!\left(\sum_{\alpha=1}^k\iota_{X_\alpha}\omega^\alpha\right)=\rho^*(\d H) = \d h\,.
\]
On the other hand,
\[
\iota_{E_\alpha}\eta^\alpha=\iota_{\widetilde X_\alpha}(\d z^\alpha-\rho^*\theta^\alpha)+\rho^*(\theta^\alpha(X_\alpha))\,\iota_{\partial/\partial z^\alpha}\eta^\alpha\,.
\]
Since \(\widetilde X_\alpha\) is tangent to \(P\), one has \(\iota_{\widetilde X_\alpha}\d z^\alpha=0\), while \(\iota_{\widetilde X_\alpha}\rho^*\theta^\alpha=\rho^*(\theta^\alpha(X_\alpha))\) and \(\iota_{\partial/\partial z^\alpha}\eta^\alpha=1\). Therefore \(\iota_{E_\alpha}\eta^\alpha=0\) for every \(\alpha\), and thus \(\sum_{\alpha=1}^k\iota_{E_\alpha}\eta^\alpha=0\). Hence \(\mathbf E\) is an evolution \(k\)-contact \(k\)-vector field for \(h\).

Finally, let \(\mathbf F=(F_1,\dots,F_k)\) be another \(k\)-vector field on \(M\) such that \(\rho_*F_\alpha=X_\alpha\) and \(\iota_{F_\alpha}\eta^\alpha=0\) for every \(\alpha\). Then \(F_\alpha-\widetilde X_\alpha\) is vertical with respect to \(\rho\), so \(F_\alpha=\widetilde X_\alpha+f_\alpha^\beta\partial/\partial z^\beta\) for certain functions \(f_\alpha^\beta\). Since \(\eta^\alpha(\partial/\partial z^\beta)=\delta^\alpha_\beta\), the conditions \(\iota_{F_\alpha}\eta^\alpha=0\), for $\alpha=1,\ldots,k$, yield
\[
0=\eta^\alpha(F_\alpha)=-\rho^*(\theta^\alpha(X_\alpha))+f_\alpha^\alpha\,,\qquad \alpha=1,\ldots,k,
\]
Thus, \(f_\alpha^\alpha=\rho^*(\theta^\alpha(X_\alpha))\) for $\alpha=1,\ldots,k$.  The conditions \(\iota_{F_\alpha}\dd z^\beta=0\) for \(\beta\neq\alpha\) force all off-diagonal coefficients to vanish, that is, \(f_\alpha^\beta=0\) whenever \(\beta\neq\alpha\). Then, the condition \(\iota_{F_\alpha}\eta^\alpha=0\) determines uniquely the remaining diagonal coefficients. 
Hence, \(F_\alpha=E_\alpha\).
\end{proof}

\subsection{HdDW equations for special types of Hamiltonian functions and general PDEs}\label{Sec:SpecialTypes}

The aim of this section is to study a particular class of 
$k$-contact Hamiltonian systems, their HdDW equations, and their use in the study of higher-order systems of PDEs with relevant applications \cite{Riv_22}.

The following definition is a generalisation of the notion of regularity for Hamiltonian functions in the symplectic and contact cases to the $k$-contact realm. It is also a natural analogue of the notion of regularity for Lagrangian functions in the $k$-contact setting appearing in \cite{Riv_22}.

\begin{definition}\label{def:regular_hamiltonian}
We say that \((\mathcal{J}_{Q,k},\bm \eta_{Q,k},h)\) is \emph{regular} if the fibre derivative \(\mathbb Fh\colon\bigoplus^k\T^*Q\times\mathbb R^k\to \bigoplus^k\T Q\times\mathbb R^k\), given in canonical coordinates by \(\mathbb Fh(q^i,p_i^\alpha,z^\alpha)=(q^i,\partial h/\partial p_i^\alpha,z^\alpha)\), is a local diffeomorphism. If \(\mathbb Fh\) is a global diffeomorphism, then \((\mathcal{J}_{Q,k},\bm \eta_{Q,k},h)\) is said to be \emph{hyperregular}.
\end{definition}

In a slight abuse of notation, the terms `regular' and `hyperregular' are only used in reference to $h$ when it is understood from context that the notion is with respect to $(\mathcal{J}_{Q,k},\bm\eta_{Q,k},h)$. 

An analogue of the following proposition for the Lagrangian setting can be found in \cite[Proposition 3.2]{Riv_22}.

\begin{proposition}\label{prop:equivalent_regular_hamiltonian}
Given \((\mathcal{J}_{Q,k},\bm\eta_{Q,k}, h)\), the following conditions are equivalent: 
\begin{enumerate}
\item \(h\) is regular.
\item The Hessian matrix \(\bigl(\partial^2 h/\partial p_i^\alpha\partial p_j^\beta\bigr)\) is non-singular at every point.
\item The relations \(v_\alpha^i=\partial h/\partial p_i^\alpha\) can be solved locally for the variables \(p_i^\alpha\) as smooth functions of \((q^i,v_\alpha^i,z^\alpha)\).
\end{enumerate}
\end{proposition}

Let us now formalise for our setting a result that describes how $k$-contact geometry has been used to study higher-order systems of PDEs in the literature. In particular, we will focus on the case of Hamiltonian functions that are affine in the dissipative variables, which is a common situation in many models due to its simplicity and applications (cf. \cite{Riv_22,Riv_23,RSS_24,GRR_22}). The following result states that, for such Hamiltonians, every integral section of an \(\bm\eta_{Q,k}\)-Hamiltonian or \(\bm\eta_{Q,k}\)-evolution $k$-vector field associated with the same Hamiltonian function projects onto a map whose coordinates satisfy the same second-order system of PDEs. It also explains certain additional properties of this case, which will be important for the applications of $k$-contact geometry to the study of PDEs in the literature and the rest of our paper.

\begin{theorem}\label{thm:linear_z_same_second_order_system}
Let \(h\in \Cinfty(\mathcal J_{Q,k})\) be of the form \(h(q^i,p_i^\alpha,z^\alpha)=g(q^i,p_i^\alpha)+\sum_{\alpha=1}^kA_\alpha z^\alpha\), where \(A_1,\dots,A_k\in\mathbb R\) are constants and \(h\) is regular. Every integral section \(\psi(t)=(q^i(t),p_i^\alpha(t),z^\alpha(t))\) of an \(\bm\eta_{Q,k}\)-Hamiltonian \(k\)-vector field \({\bm X}_{h}\) projects onto a map \(q(t)=(q^1(t),\dots,q^n(t))\) satisfying a second-order system of partial differential equations. More precisely, if \(p_i^\alpha=P_i^\alpha(q ,v_\beta^j)\) denotes the local inverse of the fibre derivative, determined by \(v_\beta^i=\partial h/\partial p_i^\beta=\partial g/\partial p_i^\beta\), then the functions \(q^i(t)\) satisfy the second-order system of PDEs of the form
\begin{equation}\label{eq:second_order_system_linear_z}
\sum_{\alpha=1}^k \frac{\partial}{\partial t^\alpha}\Bigl(P_i^\alpha(q,v_\beta)\Bigr)
+\frac{\partial g}{\partial q^i}\Bigl(q,P_j^\beta(q,v_\gamma)\Bigr)
+\sum_{\alpha=1}^k A_\alpha P_i^\alpha(q,v_\beta)=0\,,\qquad i=1,\dotsc,n\,,
\end{equation}
where \(v_\beta =\partial q /\partial t^\beta\) and \eqref{eq:second_order_system_linear_z} is independent of the dissipative variables. Moreover, the same second-order system is obtained from every integral section of an \(\bm\eta_{Q,k}\)-evolution $k$-vector field associated with the same Hamiltonian \(h\). In particular, \eqref{eq:second_order_system_linear_z} depends only on \(h\) and not on the particular \(\bm\eta_{Q,k}\)-Hamiltonian or \(\bm\eta_{Q,k}\)-evolution $k$-vector field chosen.
\end{theorem}

\begin{proof}
Let us consider the case of an integral section of an ${\bm\eta}_{Q,k}$-Hamiltonian $k$-vector field. 
Let \(\psi(t)=(q^i(t),p_i^\alpha(t),z^\alpha(t))\) be an integral section of an \({\bm\eta}_{Q,k}\)-Hamiltonian \(k\)-vector field ${\bm X}_h$. By the local expression of the \(k\)-contact HdDW equations for ${\bm X}_h$, one has \(\partial q^i/\partial t^\beta=(\partial h/\partial p_i^\beta)\circ\psi\) and \(\sum_{\alpha=1}^k\partial p_i^\alpha/\partial t^\alpha=-\bigl(\partial h/\partial q^i+\sum_{\alpha=1}^k p_i^\alpha\,\partial h/\partial z^\alpha\bigr)\circ\psi\) for $i=1,\dotsc,n$. Since \(h=g+\sum_{\alpha=1}^kA_\alpha z^\alpha\), it follows that \(\partial h/\partial p_i^\beta=\partial g/\partial p_i^\beta\), \(\partial h/\partial q^i=\partial g/\partial q^i\), and \(\partial h/\partial z^\alpha=A_\alpha\). Hence,
\[
\frac{\partial q^i}{\partial t^\beta}=\frac{\partial g}{\partial p_i^\beta}\circ\psi\,,
\qquad 
\sum_{\alpha=1}^k\frac{\partial p_i^\alpha}{\partial t^\alpha}
=
-\left(\frac{\partial g}{\partial q^i}+\sum_{\alpha=1}^k A_\alpha p_i^\alpha\right)\circ\psi\,,\qquad i=1,\dotsc,n\,.
\]

Since \(h\) is regular, Proposition~\ref{prop:equivalent_regular_hamiltonian} ensures that the relations \(v_\beta^i=\partial g/\partial p_i^\beta\) can be solved locally as \(p_i^\alpha=P_i^\alpha(q,v_\beta)\), which is independent of the values of the dissipative variables. Along \(\psi\), the first equations above yield \(p_i^\alpha=P_i^\alpha(q,v_\beta)\). Substituting these expressions into the second equations, we obtain \eqref{eq:second_order_system_linear_z}.

Assume now that \(\psi\) is an integral section of an \(\bm\eta_{Q,k}\)-evolution $k$-vector field associated with the same Hamiltonian \(h\). By the local form of the evolution HdDW equations and the special form of $h$, the equations for \(q^i\) and \(p_i^\alpha\) are again \(\partial q^i/\partial t^\beta=(\partial g/\partial p_i^\beta)\circ\psi\) and \(\sum_{\alpha=1}^k\partial p_i^\alpha/\partial t^\alpha=-\bigl(\partial g/\partial q^i+\sum_{\alpha=1}^k p_i^\alpha\,A_\alpha\bigr)\circ\psi\). Therefore, for Hamiltonians of the form \(h=g+\sum_{\alpha=1}^kA_\alpha z^\alpha\), one obtains exactly the same first-order system for \((q^i,p_i^\alpha)\), and, using the regularity of $h$, the same second-order system of PDEs \eqref{eq:second_order_system_linear_z} for the variables \(q^i\). 
\end{proof}

The system of second-order PDEs \eqref{eq:damped-membrane-foundation} is the particular case of the second-order system \eqref{eq:second_order_system_linear_z} for the Hamiltonian function \eqref{eq:Hamiltonian-damped-membrane}, which is affine in the dissipative variables in the manner given in Theorem~\ref{thm:linear_z_same_second_order_system}. Hence, Theorem~\ref{thm:linear_z_same_second_order_system} ensures that every integral section of an \(\bm\eta_{Q,k}\)-Hamiltonian or \(\bm\eta_{Q,k}\)-evolution $k$-vector field associated with the Hamiltonian \eqref{eq:Hamiltonian-damped-membrane} projects onto a solution of the damped wave equation \eqref{eq:damped-membrane-foundation}.

\begin{proposition}\label{thm:difference_Hamiltonian_kernel_flat}
Let \(h\in \Cinfty(\mathcal J_{Q,k})\) and let \(\mathbf X=(X_1,\dots,X_k)\) and \(\mathbf Y=(Y_1,\dots,Y_k)\) be two \(\bm\eta_{Q,k}\)-Hamiltonian \(k\)-vector fields associated with the same Hamiltonian function \(h\). Then, \(\mathbf X-\mathbf Y\) takes values in \(\ker\chi\). In particular, \(\mathbf X-\mathbf Y\) is \(\pi_Q\)-vertical. The same conclusion holds if \(\mathbf X\) and \(\mathbf Y\) are two \(\bm\eta_{Q,k}\)-evolution $k$-vector fields associated with the same Hamiltonian function \(h\).
\end{proposition}

\begin{proof}
Since \(\mathbf X\) and \(\mathbf Y\) are \(\bm\eta_{Q,k}\)-Hamiltonian for the same function \(h\), they satisfy the same defining equations and hence \(\chi(\mathbf X)=\chi(\mathbf Y)\). By fibrewise linearity of \(\chi\), it follows that \(\chi(\mathbf X-\mathbf Y)=0\), so \(\mathbf X-\mathbf Y\) takes values in \(\ker\chi\). Moreover, by Proposition~\ref{prop:kernel-flat-k-contact-Darboux}, one has \(\ker\chi\subset\bigoplus^kV(\pi_Q)\), where \(V(\pi_Q)=\ker \T\pi_Q\). Therefore, \(\mathbf X-\mathbf Y\) is \(\pi_Q\)-vertical. The proof for \(\bm\eta_{Q,k}\)-evolution $k$-vector fields is identical.
\end{proof}

Proposition 3.5 of \cite{HLM_26} and Proposition \ref{prop:k-contact-evolution-HdDW-have-solutions} show that every $(\mathcal{J}_{Q,k},\bm \eta_{Q,k},h)$ is related to non-empty families of ${\bm \eta}_{Q,k}$-Hamiltonian and ${\bm \eta}_{Q,k}$-evolution $k$-vector fields. Using this, the following proposition is immediate. 

\begin{proposition}
The space of $\bm\eta_{Q,k}$-Hamiltonian $k$-vector fields associated with
$h\in \Cinfty(\mathcal{J}_{Q,k})$ is an affine
space modelled on $\Gamma(\ker\chi)$. The same holds for $\bm\eta_{Q,k}$-evolution
$k$-vector fields.
\end{proposition}

Higher-order systems of PDEs can be written as first-order systems of PDEs by considering higher-order partial derivatives of the coordinates of solutions as new variables. In this manner, the resulting system can be described by our $k$-contact formalism.

\section{The \texorpdfstring{$k$}{}-contact Hamilton--Jacobi equations in the \texorpdfstring{$z$}{}-independent framework}
\label{sec:HJzindependent}


Let us develop two types of Hamilton--Jacobi theories. The main idea is to analyse the HdDW equations for a $k$-contact Hamiltonian system $(\mathcal{J}_{Q,k}=\bigoplus^k \T^*Q\times \R^k,\bm \eta_{Q,k},h)$ by means of the values of an ${\bm \eta}_{Q,k}$-Hamiltonian $k$-vector field along a submanifold of $\mathcal{J}_{Q,k}$ that is projectable onto $Q$ in a suitable sense. Nevertheless, the ${\bm \eta}_{Q,k}$-Hamiltonian $k$-vector field does not need to be tangent to the submanifold for $k>1$. It is worth noting that our formalism yields a generalisation to the $k$-contact setting of a classical result on jet manifolds described by means of a contact form: a submanifold is Legendrian if and only if it can be understood as the image of a  holonomic section.

Consider a section of the bundle \(
\pi_Q\colon \bigoplus\nolimits^k \cT Q \times \mathbb{R}^k \to Q
\) locally given in adapted coordinates by
\[
\begin{array}{rccl}
\gamma\colon & Q & \longrightarrow & \bigoplus\nolimits^k \cT Q \times \mathbb{R}^k \\
& q=(q^i) & \longmapsto & \left(q^i,\gamma_i^\alpha(q),\gamma^\alpha(q)\right).
\end{array}
\]
Given a smooth map $f\colon M\to N$, define
\[\textstyle
\bigoplus\nolimits^k \T f\colon \bigoplus\nolimits^k \T M \longrightarrow \bigoplus\nolimits^k \T N\,,
\qquad
(v_1,\dotsc,v_k)\longmapsto (\T f(v_1),\dotsc,\T f(v_k))\,.
\]

Let $\mathbf{Z}=(Z_\alpha)$ be a $k$-vector field on
$\bigoplus^k \cT Q \times \mathbb{R}^k$. Since $\gamma$ is a section of $\pi_Q$, let us define the projected $k$-vector
field on $Q$ by
\[
\textstyle\mathbf{Z}^\gamma
=\left(\bigoplus\nolimits^k \T\pi_Q\right)\circ (\mathbf{Z}\circ \gamma)\,.
\]
We denote by $\left[\cdot\right]\colon\bigoplus^k\T\mathcal{J}_{Q,k}\to \bigoplus^k\T\mathcal{J}_{Q,k}/\ker \chi$ the natural projection onto the quotient space. 
The difference between \(
\mathbf{Z}\circ \gamma\) and \(
\left(\bigoplus\nolimits^k \T\gamma\right)\circ \mathbf{Z}^\gamma
\) takes values in $\ker \chi$ if and only if the diagram
\[
\begin{tikzcd}[row sep=4.5em]
\mathcal{J}_{Q,k}
  \arrow[r,"{[\mathbf{Z}]}"]
&
\bigoplus^k \T\mathcal{J}_{Q,k}/\ker \chi
  \arrow[d,"\widetilde{\bigoplus^k \T\pi_Q}"']
\\
Q
  \arrow[r,"\mathbf{Z}^\gamma"]
  \arrow[u,"\gamma"]
&
\bigoplus^k \T Q
  \arrow[u,bend right=30,"{\left[(\bigoplus^k \T\gamma)\right]}"']
\end{tikzcd}
\]
is  commutative. Note also that the diagram is well-defined since $\ker \bigoplus^k\T\pi_Q\supset \ker \chi$. We write $\widetilde{\bigoplus^k \T\pi_Q}([{\bm Z}])=[\bigoplus^k\T\pi_Q ({\bf Z})]$. We say that $\mathbf{Z}$ and $\mathbf{Z}^\gamma$ are {\it almost $\gamma$-related} if the above diagram is commutative. Meanwhile, $\mathbf{Z}$ and $\mathbf{Z}^\gamma$ are $\gamma$-related if  $\mathbf{Z}\circ \gamma = (\bigoplus^k\T \gamma)\circ {\bm Z}^\gamma$.

For $k>1$, the almost $\gamma$-relation does not ensure that $\mathbf{Z}$ on ${\rm Im}\,\gamma$ is the lift via $ \bigoplus^k\T\gamma$  of $\mathbf{Z}^\gamma$. It is also worth noting that $\ker \chi=0$ if and only if $k=1$. Hence,  the almost $\gamma$-relation is equivalent to the $\gamma$-relation in the contact case.

Recall that ${\rm Im}\, \gamma$ is a submanifold in $\bigoplus^k \cT Q \times \mathbb{R}^k$.   Note that if ${\bf X}_h$ is an $\bm\eta_{Q,k}$-Hamiltonian $k$-vector field, then 
\begin{equation}\label{eq:ProjZgamma}
({\bf X}^\gamma)_\alpha(q)=\sum_{i=1}^n\frac{\partial h}{\partial p^\alpha_i}\circ\gamma(q)\frac{\partial}{\partial q^i}\,,\qquad \alpha=1,\dotsc,k\,,\qquad q\in Q\,.
\end{equation}
A similar expression is obtained from \eqref{eq:k-contact-HdDW-fields-Darboux-coordinates} if ${\bf Z}$ is an $\bm \eta_{Q,k}$-evolution  $k$-vector field.

Consider a Legendrian submanifold of $\bigoplus ^k\cT Q\times \R^k$ relative to its natural $k$-contact manifold structure induced by $\bm\eta_{Q,k}$. Let us describe the conditions for $\Ima\gamma$ to be a  Legendrian submanifold. 
The canonical diffeomorphism
$$ \textstyle J^1(Q,Q\times \mathbb{R}^k) \longrightarrow \bigoplus ^k\cT Q\times \mathbb{R}^k\,, \qquad  j^1_x\sigma\longmapsto (x,\dd \sigma^\mu(x),\sigma^\mu(x)) $$
allows us to understand $\bigoplus ^k\cT Q\times \mathbb{R}^k$  as the first-order jet manifold $J^1(Q,Q\times \mathbb{R}^k)$ of sections of the bundle $Q\times \mathbb{R}^k\rightarrow Q$ \cite{OLV_86}. This is what inspired us to denote \(\bigoplus^k\T^*Q\times \mathbb{R}^k\) by \(\mathcal{J}_{Q,k}\). This also justifies calling a section \(\gamma\colon Q\rightarrow \mathcal{J}_{Q,k}\) {\it holonomic} when it is of the form
$$
\gamma(q)=(q,\dd\gamma^\mu(q),\gamma^\mu(q))\,.
$$
The characterisation of holonomic sections is accomplished by means of the following result, which has applications in the further Hamilton--Jacobi theory. It is worth noting that one can also define the symplectic orthogonal of a vector subspace $W_x\subset \T_xM$ relative to a two-form $\bm\Omega$ on a manifold $M$ taking values in $\mathbb{R}^k$ as follows
$$
W_x^{\perp_{\bm\Omega}}=\{v_x\in \T_xM\mid\bm\Omega(v_x,w_x)=0,\forall w_x\in W_x\}.
$$

\begin{theorem}\label{thm:k-coisotropic-condition}
    Given a section $\gamma\colon \left(q^i\right) \in Q  \mapsto \left(q^i,\gamma_i^\alpha(q), \gamma^\alpha(q) \right)\in\mathcal{J}_{Q,k}=\bigoplus^k\cT Q\times \R^k$ in adapted coordinates \((q^i,p_i^\alpha,z^\alpha)\), the tangent space to $\Ima \gamma$ is 
        \[
\T\Ima\gamma
=
\left\langle
\frac{\partial}{\partial q^i}
+
\sum_{\alpha=1}^k \sum_{j=1}^n
\frac{\partial \gamma^\alpha_j}{\partial q^i}
\frac{\partial}{\partial p^\alpha_j}
+
\sum_{\alpha=1}^k
\frac{\partial \gamma^\alpha}{\partial q^i}
\frac{\partial}{\partial z^\alpha}
\right\rangle_{i=1,\dotsc,n}\,.
\]
The submanifold $\gamma(Q)$ is Legendrian if and only if $\gamma$ is holonomic. In particular, the condition \({\rm Im}\,\T\gamma\subset \ker \bm \eta\) implies that \(\gamma\) is holonomic. Meanwhile,  the orthogonal of $\T\!\Ima \gamma$ relative to $\dd \bm\eta_{Q,k}$ contains $\T\!\Ima\gamma$ if and only if  
$$
\frac{\partial \gamma^\mu_i}{\partial q^j} = \frac{\partial \gamma^\mu_j}{\partial q^i}\,, \qquad \mu=1,\dotsc,k\,, \qquad i,j =1, \dotsc, n\,.
$$
\end{theorem}

\begin{proof}
The canonical $k$-contact form on $\mathcal{J}_{Q,k}=\bigoplus^k\cT Q\times \R^k$ reads
${\bm \eta}_{Q,k}=\sum_{\mu=1}^k(\mathrm{d}z^\mu - \sum_{i=1}^np_i^\mu\,\mathrm{d}q^i)\otimes e_\mu$, and hence $\mathrm{d}\bm \eta_{Q,k} = -\sum_{\mu=1}^k\sum_{i=1}^n\,\mathrm{d}p_i^\mu\wedge \mathrm{d}q^i\otimes e_\mu$. Since $\gamma(q)=(q,\gamma^\mu_i(q),\gamma^\mu(q))$, the condition ${\rm Im}\,\T\gamma \subset \ker \bm \eta_{Q,k}$, which gives
$$
\gamma^*\bm \eta_{Q,k}=0\quad\Longleftrightarrow\quad \dd\gamma^\mu=\sum_{i=1}^n\gamma^\mu_i(q)\dd q^i\,,\qquad \mu=1,\ldots,k\,.
$$
In other words, 
\begin{equation}\label{eq:ConHol} \gamma^\mu_i(q)=\parder{\gamma^\mu}{q^i}(q)\,,\qquad  \mu=1,\dotsc,k\,,\qquad i=1,\dotsc,n\,.
\end{equation}
Hence, it makes sense to write $\gamma(q)=(q,\dd\gamma^\mu (q),\gamma^\mu(q))$ and $\gamma$ is holonomic. In other words, one has that the section can be understood as the one-jet prolongation of a section of $Q\times \mathbb{R}^k\rightarrow Q$ (see \cite{Saunders1989}). To verify that $\gamma(Q)$ is also Legendrian, one may prove that it admits a complement  in $\ker \bm \eta_{Q,k}$ that is also isotropic. But this is immediate since $\mathcal{V}_{Q,k}$ spans an isotropic subbundle in $\ker \bm \eta_{Q,k}$ that complements $\T{\rm Im}\gamma$. The converse is immediate.

Moreover, 
\[
\gamma^*\dd\bm \eta_{Q,k}= -\,\frac 12\sum_{i,j=1}^n\sum_{\mu=1}^k\frac{\partial \gamma_i^\mu}{\partial q^j}\,
\mathrm{d}q^j\wedge \mathrm{d}q^i\otimes e_\mu\,,
\]
and the orthogonal of  $\T\Ima\gamma=$ relative to $\dd \bm\eta_{Q,k}$ includes $\T\Ima\,\gamma$ if and only if these terms vanish, i.e.
\begin{equation}\label{eq:closedgamma2}
\frac{\partial \gamma_i^\mu}{\partial q^j} =
\frac{\partial \gamma_j^\mu}{\partial q^i}\,,\qquad \mu=1,\dotsc,k\,,\qquad 1\leq i<j\leq n\,,
\end{equation}
which is automatically satisfied under condition \eqref{eq:ConHol}. Again, the converse is immediate.
\end{proof}

\subsection{The classical Hamiltonian approach}
\label{Sec::zindclass}
Let us introduce the $z$-independent $k$-contact Hamilton--Jacobi theorem for $k$-contact Hamiltonian $k$-vector fields and analyse its consequences and applications. 

\begin{theorem}\label{theo:ClassicalHJzindependent}
Let $\mathbf{X}_h$ be an ${\bm \eta}_{Q,k}$-Hamiltonian $k$-vector field associated with $h\in\Cinfty(\bigoplus^k \cT Q \times \mathbb{R}^k)$, let $\gamma$ be a holonomic section, and assume $\mathbf{X}_h^\gamma$ to be integrable. Then, the following statements are equivalent:
\begin{enumerate}
\item Every integral section of $\mathbf{X}_h^\gamma$ lifts to a solution of the $z$-independent $\bm\eta_{Q,k}$-Hamiltonian HdDW equations for $h$,
\item $h\circ\gamma=0$.
\end{enumerate}
\end{theorem}

\begin{proof}
Let us first prove that \(1.\) implies \(2.\) Suppose that \((\gamma \circ \sigma)(t)=(\sigma^i(t),\gamma_i^\alpha(\sigma(t)),\gamma^\alpha(\sigma(t)))\) is a solution to the $z$-independent $\bm\eta_{Q,k}$-Hamiltonian HdDW equations. Then,
\[
\begin{dcases}
\restr{\dfrac{\partial \sigma^i}{\partial t^\beta}}{t}=\restr{\dfrac{\partial h}{\partial p_i^\beta}}{\gamma\circ \sigma(t)}\,,\qquad \beta=1,\dotsc,k  \\[2mm]
\sum_{\alpha=1}^k \restr{\dfrac{\partial(\gamma_i^\alpha\circ \sigma)}{\partial t^\alpha}}{t}
=
-\left(\restr{\dfrac{\partial h}{\partial q^i}}{\gamma\circ \sigma(t)}
+\sum_{\alpha=1}^k (\gamma_i^\alpha\circ \sigma)(t)\restr{\dfrac{\partial h}{\partial z^\alpha}}{\gamma\circ \sigma(t)}\right)\,,\qquad i=1,\dotsc,n\,,\\[2mm]
\sum_{\alpha=1}^k \restr{\dfrac{\partial(\gamma^\alpha\circ \sigma)}{\partial t^\alpha}}{t}
=
\sum_{\alpha=1}^k\sum_{j=1}^n (\gamma_j^\alpha\circ \sigma)(t)\restr{\dfrac{\partial h}{\partial p_j^\alpha}}{\gamma\circ \sigma(t)}
-h\circ \gamma\circ \sigma(t)\,.
\end{dcases}
\]
Since $\gamma(Q)$ is Legendrian because $\gamma$ is holonomic, one has \(\dfrac{\partial \gamma^\alpha}{\partial q^i}=\gamma_i^\alpha\) for \(i=1,\dotsc,n\) and \(\alpha=1,\dotsc,k\). Using the first equation above, it follows that
\[
\sum_{\alpha=1}^k \frac{\partial (\gamma^\alpha \circ \sigma)}{\partial t^\alpha}
=
\sum_{\alpha=1}^k \sum_{i=1}^n \frac{\partial \gamma^\alpha}{\partial q^i}(\sigma(t)) \frac{\partial \sigma^i}{\partial t^\alpha}
=
\sum_{\alpha=1}^k \sum_{i=1}^n (\gamma_i^\alpha\circ \sigma)(t)\,\frac{\partial h}{\partial p^\alpha_i}\circ \gamma \circ \sigma(t)\,.
\]
Comparing this expression with the third HdDW equation, we obtain \(h\circ \gamma\circ \sigma=0\). Since \(\mathbf{X}_h^\gamma\) is integrable, through every point \(q\in Q\) there exists a local integral section \(\sigma\) of \(\mathbf{X}_h^\gamma\) with \(\sigma(t_0)=q\) for some \(t_0\). Hence, \(h\circ \gamma(q)=0\) for every \(q\in Q\), and therefore \(h\circ \gamma=0\).

Let us prove the converse. Assume that \(h\circ \gamma =0\) and let \(\sigma\) be an integral section of \(\mathbf{X}_h^\gamma\). Let us show that \(\gamma \circ \sigma\) is a solution to the $z$-independent $\bm\eta_{Q,k}$-Hamiltonian HdDW equations for \(h\). In local coordinates on \(Q\), and since \(\gamma\) is holonomic, it follows that \eqref{eq:closedgamma2} holds. Since \(h\circ \gamma=0\), we also have \(\dd(h\circ \gamma)=0\), that is
\begin{equation}
\label{eq:d(hgamma)=02}
0 = \sum_{i=1}^n\left(\frac{\partial h}{\partial q^i}\circ \gamma + \sum_{\alpha=1}^k\sum_{j=1}^n\left(\frac{\partial h}{\partial p^\alpha_j}\circ \gamma\right)\frac{\partial \gamma^\alpha_j}{\partial q^i} + \sum_{\alpha=1}^k\left(\frac{\partial h}{\partial z^\alpha}\circ \gamma\right)\frac{\partial \gamma^\alpha}{\partial q^i}\right)\dd q^i.
\end{equation}
Since \(\mathbf{X}_h\) is a $k$-contact Hamiltonian $k$-vector field related to \(h\), equations \eqref{eq:k-contact-HdDW-fields-Darboux-coordinates1} and \eqref{eq:ProjZgamma} yield
\[
(\mathbf X_h^\gamma)_\alpha = \sum_{i=1}^n\left(\frac{\partial h}{\partial p_i^\alpha}\circ \gamma\right)\frac{\partial}{\partial q^i}\,,\qquad \alpha=1,\dotsc,k\,.
\]
On the other hand, since \(\sigma\) is an integral section of \(\mathbf X_h^\gamma\), we have
\begin{equation}
\label{eq:sigmaHDW2}
\restr{\frac{\partial \sigma^i}{\partial t^\alpha}}{t} = \frac{\partial h}{\partial p_i^\alpha}\circ \gamma \circ \sigma(t)\,,\qquad \alpha=1,\dotsc,k\,,\qquad i=1,\dotsc,n\,.
\end{equation}
Using \eqref{eq:closedgamma2}, \eqref{eq:d(hgamma)=02}, \eqref{eq:sigmaHDW2}, and the fact that \(\dfrac{\partial \gamma^\alpha}{\partial q^i}= \gamma_i^\alpha\) for all \(i\) and \(\alpha\), we obtain
\begin{align}
\sum_{\alpha=1}^k \restr{\frac{\partial(\gamma_i^\alpha \circ \sigma)}{\partial t^\alpha}}{t}
& =
\sum_{\alpha=1}^k\sum_{j=1}^n \restr{\frac{\partial \gamma_i^\alpha}{\partial q^j}}{\sigma(t)} \restr{\frac{\partial \sigma^j}{\partial t^\alpha}}{t}
=
\sum_{\alpha=1}^k\sum_{j=1}^n \restr{\frac{\partial \gamma_j^\alpha}{\partial q^i}}{\sigma(t)} \restr{\frac{\partial h}{\partial p_j^\alpha}}{\gamma(\sigma(t))} \\
&=
- \restr{\frac{\partial h}{\partial q^i}}{\gamma(\sigma(t))}
- \sum_{\alpha=1}^k \restr{\frac{\partial \gamma^\alpha}{\partial q^i}}{\sigma(t)} \restr{\frac{\partial h}{\partial z^\alpha}}{\gamma(\sigma(t))} \\
&=
- \restr{\frac{\partial h}{\partial q^i}}{\gamma(\sigma(t))}
- \sum_{\alpha=1}^k (\gamma_i^\alpha\circ \sigma)(t)\restr{\frac{\partial h}{\partial z^\alpha}}{\gamma(\sigma(t))}\,,
\end{align}
for $i=1,\dotsc,n$.  This proves the second family of HdDW equations. Finally,
\[
\sum_{\alpha=1}^k \restr{\frac{\partial (\gamma^\alpha \circ \sigma)}{\partial t^\alpha}}{t}
=
\sum_{\alpha=1}^k\sum_{i=1}^n \restr{\frac{\partial \gamma^\alpha}{\partial q^i}}{\sigma(t)} \restr{\frac{\partial \sigma^i}{\partial t^\alpha}}{t}
=
\sum_{\alpha=1}^k\sum_{i=1}^n (\gamma_i^\alpha \circ \sigma)(t)\restr{\frac{\partial h}{\partial p_i^\alpha}}{\gamma\circ\sigma(t)}.
\]
Since \(h\circ \gamma=0\), this is exactly the third HdDW equation for \(\gamma\circ \sigma\). The first family of HdDW equations follows directly from \eqref{eq:sigmaHDW2}. Hence, \(\gamma\circ \sigma\) is a solution to the $z$-independent $\bm\eta_{Q,k}$-Hamiltonian HdDW equations. 
\end{proof}

Note that the previous theorem, and its corresponding Hamilton--Jacobi equation \(h\circ\gamma=0\), ensure that \(\mathbf X_h^\gamma\) gives rise, via \(\gamma\), to a \(k\)-vector field tangent to \({\rm Im}\,\gamma\). Indeed, if \(\sigma\) is an integral section of \(\mathbf X_h^\gamma\), then \(\gamma\circ\sigma\) is a solution of the $z$-independent $\bm\eta_{Q,k}$-Hamiltonian HdDW equations for \(h\), and therefore its first prolongation is tangent to \({\rm Im}\,\gamma\).

If \(k=1\), then \(\ker\chi=0\), so the almost \(\gamma\)-relation reduces to the usual \(\gamma\)-relation. Hence, the lift of \(\mathbf X_h^\gamma\) coincides with \(\mathbf X_h\) along \({\rm Im}\,\gamma\). Nevertheless, for \(k>1\), the lift does not need to agree with \(\mathbf X_h\) on \({\rm Im}\,\gamma\), since the difference may be a \(k\)-vector field taking values in \(\ker\chi\). For instance, if \(\mathbf X_h\) is an \({\bm \eta}_{Q,k}\)-Hamiltonian \(k\)-vector field, then so is \(\mathbf X'_h\coloneqq\mathbf X_h+\mathbf Z\), where \(\mathbf Z=(Z_1,\dots,Z_k)\) is a \(k\)-vector field whose components are given in Darboux coordinates by
\[
Z_\alpha=\sum_{\beta\neq \alpha}Z_{\alpha}^\beta\frac{\partial}{\partial z^\beta}\,,\qquad \alpha=1,\dots,k\,,
\]
for arbitrary functions \(Z_\alpha^\beta\in \Cinfty(\mathcal J_{Q,k})\). Since \(\mathbf Z\) takes values in \(\ker\chi\), the \(k\)-vector fields \(\mathbf X_h\) and \(\mathbf X'_h\) are associated with the same Hamiltonian function \(h\) and induce the same projected \(k\)-vector field on \(Q\), while they do not coincide on \({\rm Im}\,\gamma\) in general.

As a final remark, note that our Hamilton--Jacobi theorem uses only the integrability of ${\bm X}_h^\gamma$. This is important in applications, where one is interested in finding solutions and some mild integrability condition is used. In this sense, the integrability of ${\bm X}_h^\gamma$ is much easier to ensure than the integrability of ${\bm X}_h$. 

\subsection{The evolution approach}
\label{sec:zindevolution}

Let us now give a version of the previous section for ${\bm \eta}_{Q,k}$-evolution $k$-vector fields in the $z$-independent case.

\begin{theorem}\label{theo:evolutionHJzindependent}
Let $\mathbf{E}_h$ be an ${\bm \eta}_{Q,k}$-evolution $k$-vector field associated with $h\in \Cinfty(\mathcal{J}_{Q,k})$, let $\gamma$ be a holonomic section, and assume $\mathbf{E}_h^\gamma$ to be integrable. Then, the following statements are equivalent:
\begin{enumerate}
\item Every integral section $\sigma$ of $\mathbf{E}_h^\gamma$ lifts to a solution $\gamma\circ\sigma$ of the $z$-independent ${\bm \eta}_{Q,k}$-evolution HdDW equations,
\item $\dd(h\circ\gamma)=0$.
\end{enumerate}
\end{theorem}

\begin{proof}
Suppose that $(\gamma \circ \sigma)(t) = (\sigma^i(t), \gamma_i^\alpha(\sigma(t)), \gamma^\alpha(\sigma(t)))$ is a solution to the $z$-independent HdDW equations for ${\bm \eta}_{Q,k}$-evolution  $k$-vector fields. Then,
\[
\begin{dcases}
\restr{\dfrac{\partial \sigma^i}{\partial t^\alpha}}{t} = \restr{\dfrac{\partial h}{\partial p_i^\alpha}}{\gamma\circ \sigma(t)}\,,\\[2mm]
\sum_{\alpha=1}^k \restr{\dfrac{\partial(\gamma_i^\alpha\circ \sigma)}{\partial t^\alpha}}{t}
=
-\left( \restr{\dfrac{\partial h}{\partial q^i}}{\gamma\circ \sigma(t)} + \sum_{\alpha=1}^k\left(\gamma_i^\alpha \circ \sigma(t) \right) \restr{\dfrac{\partial h}{\partial z^\alpha}}{\gamma\circ \sigma(t)}\right)\,,  \\[2mm]
\sum_{\alpha=1}^k \restr{\dfrac{\partial(\gamma^\alpha\circ \sigma)}{\partial t^\alpha}}{t}
=
\sum_{\alpha=1}^k\sum_{j=1}^n (\gamma^\alpha_j\circ \sigma)(t) \restr{\dfrac{\partial h}{\partial p_j^\alpha}}{\gamma\circ \sigma(t)},
\end{dcases}
\]
for $\alpha=1,\dotsc,k,\quad i=1,\dotsc,n\,.$ On the other hand,
\[
\dd(h \circ \gamma)
=
\sum_{i=1}^n\left( \frac{\partial h}{\partial q^i} \circ \gamma + \sum_{\alpha=1}^k\sum_{j=1}^n\left( \frac{\partial h}{\partial p_j^\alpha}\circ \gamma \right) \frac{\partial \gamma_j^\alpha}{\partial q^i} + \sum_{\alpha=1}^k\left( \frac{\partial h}{\partial z^\alpha}\circ \gamma \right) \frac{\partial \gamma^\alpha}{\partial q^i} \right) \dd q^i.
\]
Since $\gamma$ is holonomic, it follows that \eqref{eq:closedgamma2} holds.
Using \eqref{eq:closedgamma2} and the $z$-independent $\bm\eta_{Q,k}$-evolution HdDW equations, we have
\begin{align*}
\restr{\dd(h\circ \gamma)}{\sigma(t)}
&=
\sum_{i=1}^n
\Bigg(
  \restr{\frac{\partial h}{\partial q^i}}{\gamma(\sigma(t))}
  + \sum_{\alpha=1}^k \sum_{j=1}^n
    \restr{\frac{\partial h}{\partial p^\alpha_j}}{\gamma(\sigma(t))}
    \restr{\frac{\partial \gamma^\alpha_j}{\partial q^i}}{\sigma(t)}\\[1ex]&
  + \sum_{\alpha=1}^k
    \restr{\frac{\partial h}{\partial z^\alpha}}{\gamma(\sigma(t))}
    \restr{\frac{\partial \gamma^\alpha}{\partial q^i}}{\sigma(t)}
\Bigg)
\, \dd q^i(\sigma(t))
\\[1ex]
&=
\sum_{i=1}^n
\Bigg(
  - \sum_{\alpha=1}^k
    \left(\gamma_i^\alpha\circ \sigma\right)(t)
    \restr{\frac{\partial h}{\partial z^\alpha}}{\gamma(\sigma(t))}
  - \sum_{\alpha=1}^k
    \restr{\frac{\partial(\gamma_i^\alpha\circ\sigma)}{\partial t^\alpha}}{t}
\\
&\qquad\qquad
  + \sum_{\alpha=1}^k \sum_{j=1}^n
    \restr{\frac{\partial \sigma^j}{\partial t^\alpha}}{t}
    \restr{\frac{\partial \gamma^\alpha_i}{\partial q^j}}{\sigma(t)}
  + \sum_{\alpha=1}^k
    \restr{\frac{\partial h}{\partial z^\alpha}}{\gamma(\sigma(t))}
    \restr{\frac{\partial \gamma^\alpha}{\partial q^i}}{\sigma(t)}
\Bigg)
\, \dd q^i(\sigma(t))
\\[1ex]
&=
\sum_{i=1}^n
\sum_{\alpha=1}^k
\left(
  - \gamma_i^\alpha(\sigma(t))
  + \restr{\frac{\partial \gamma^\alpha}{\partial q^i}}{\sigma(t)}
\right)
\restr{\frac{\partial h}{\partial z^\alpha}}{\gamma(\sigma(t))}
\, \restr{\dd q^i}{\sigma(t)} = 0\,,
\end{align*}
where we used in the last line the fact that $\dfrac{\partial \gamma^\alpha}{\partial q^i} = \gamma^\alpha_i$ for all possible indices. Hence, $\dd(h\circ \gamma)(q) = 0$ for every $q\in \Ima \sigma$. Since $\mathbf{E}_h^\gamma$ is integrable by assumption, every $q\in Q$ induces an integral section $\sigma$ such that $\sigma(0) = q$. Therefore, $\dd(h\circ\gamma)(q) = 0$ for every $q\in Q$.

Let us prove the converse. Assume that $\dd (h \circ \gamma) = 0$, and let $\sigma$ be an integral section of $\mathbf E_h^\gamma$. Let us prove that $\gamma \circ \sigma$ is a solution to the $z$-independent ${\bm \eta}_{Q,k}$-evolution HdDW equations. From $\dd(h\circ \gamma)=0$, we have
\begin{equation}
\label{eq:d(hgamma)=0-evol}
0
=
\sum_{i=1}^n
\Bigg(
  \frac{\partial h}{\partial q^i} \circ \gamma
  + \sum_{\alpha=1}^k \sum_{j=1}^n
    \left(\frac{\partial h}{\partial p^\alpha_j} \circ \gamma \right)
    \frac{\partial \gamma^\alpha_j}{\partial q^i}
  + \sum_{\alpha=1}^k
    \left(\frac{\partial h}{\partial z^\alpha} \circ \gamma \right)
    \frac{\partial \gamma^\alpha}{\partial q^i}
\Bigg)
\, \dd q^i.
\end{equation}
Since $\mathbf E_h$ is an evolution ${\bm \eta}_{Q,k}$-Hamiltonian $k$-vector field, equations \eqref{eq:ProjZgamma} and \eqref{eq:k-contact-HdDW-fields-Darboux-coordinates} yield
\[
(\mathbf E_h^\gamma)_\alpha
=
\sum_{i=1}^n
\left(
  \frac{\partial h}{\partial p_i^\alpha}
  \circ \gamma
\right)
\frac{\partial}{\partial q^i}\,,
\qquad \alpha=1,\dotsc,k\,.
\]
Since $\sigma$ is an integral section of ${\bf E}_h^\gamma$, it follows that
\begin{equation}
\label{eq:sigmaHDW-evol}
\restr{\frac{\partial \sigma^i}{\partial t^\alpha}}{t} = \frac{\partial h}{\partial p_i^\alpha}\circ \gamma \circ \sigma(t)\,,\qquad i=1,\dotsc,n\,,\qquad \alpha=1,\dotsc,k\,.
\end{equation}
From \eqref{eq:closedgamma2}, \eqref{eq:d(hgamma)=0-evol}, and \eqref{eq:sigmaHDW-evol}, one has
\begin{align}
\sum_{\alpha=1}^k
\restr{\frac{\partial(\gamma_i^\alpha \circ \sigma)}{\partial t^\alpha}}{t}
&=
\sum_{\alpha=1}^k \sum_{j=1}^n
\restr{\frac{\partial \gamma_i^\alpha}{\partial q^j}}{\sigma(t)}
\restr{\frac{\partial \sigma^j}{\partial t^\alpha}}{t}
\\
&=
\sum_{\alpha=1}^k \sum_{j=1}^n
\restr{\frac{\partial \gamma_i^\alpha}{\partial q^j}}{\sigma(t)}
\restr{\frac{\partial h}{\partial p_j^\alpha}}{\gamma(\sigma(t))}
=
\sum_{\alpha=1}^k \sum_{j=1}^n
\restr{\frac{\partial \gamma_j^\alpha}{\partial q^i}}{\sigma(t)}
\restr{\frac{\partial h}{\partial p_j^\alpha}}{\gamma(\sigma(t))}
\\
&=
- \restr{\frac{\partial h}{\partial q^i}}{\gamma(\sigma(t))}
-
\sum_{\alpha=1}^k
\restr{\frac{\partial \gamma^\alpha}{\partial q^i}}{\sigma(t)}
\restr{\frac{\partial h}{\partial z^\alpha}}{\gamma(\sigma(t))}
\\
&=
- \restr{\frac{\partial h}{\partial q^i}}{\gamma(\sigma(t))}
-
\sum_{\alpha=1}^k
\gamma_i^\alpha (\sigma(t))\,
\restr{\frac{\partial h}{\partial z^\alpha}}{\gamma(\sigma(t))}.
\end{align}
Moreover,
\[
\sum_{\alpha=1}^k
\restr{\frac{\partial (\gamma^\alpha \circ \sigma)}{\partial t^\alpha}}{t}
=
\sum_{\alpha=1}^k \sum_{i=1}^n
\restr{\frac{\partial \gamma^\alpha}{\partial q^i}}{\sigma(t)}
\restr{\frac{\partial \sigma^i}{\partial t^\alpha}}{t}
=
\sum_{\alpha=1}^k \sum_{i=1}^n
\gamma^\alpha_i(\sigma(t))\,
\restr{\frac{\partial h}{\partial p_i^\alpha}}{\gamma\circ\sigma(t)}.
\]
This concludes the proof.
\end{proof}

\begin{remark}
It is worth noting that, if \(k=1\), then \(\ker\chi=0\), so the almost \(\gamma\)-relation reduces to the usual \(\gamma\)-relation. Hence, along \(\Ima\gamma\), the lift of \(\mathbf E_h^\gamma\) coincides with \(\mathbf E_h\).

It is also worth observing that the evolution Hamilton--Jacobi equation \(\dd(h\circ\gamma)=0\) is different from the one in the standard \(z\)-independent formalism, but this has a natural geometric explanation. Indeed, a similar equation appears in the classical \(k\)-symplectic Hamilton--Jacobi theory. Let us explain this.

Recall that, given an exact \(k\)-symplectic manifold \((P,\bm\omega=\sum_{\alpha=1}^k\omega^\alpha\otimes e_\alpha)\), with \(\bm\omega=-\dd\bm\theta\), its \(k\)-contactification is \(P\times \mathbb{R}^k\), endowed with the \(k\)-contact form \(\bm\eta_{\bm \theta}=\sum_{\alpha=1}^k(\d z^\alpha-\rho^*\theta^\alpha)\otimes e_\alpha\), where \(\rho\colon P\times\mathbb{R}^k\to P\) is the canonical projection. If \(H\in \Cinfty(P)\) and \(\mathbf X_H\) is a \(k\)-symplectic Hamiltonian \(k\)-vector field associated with \(H\), then \(h\coloneqq H\circ\rho\) admits a canonically associated \(\bm\eta_{\bm\theta}\)-evolution \(k\)-vector field on \(P\times\mathbb{R}^k\) projecting onto \(\mathbf X_H\). Moreover, if \(\gamma_0:Q\to \bigoplus^k\T^*Q\) is a section given by closed one-forms, then, locally, \(\gamma_0^\alpha=\dd S^\alpha\) and one obtains a holonomic section \(\widehat\gamma\colon Q\to \mathcal J_{Q,k}=\bigoplus^k\T^*Q\times\mathbb{R}^k\) by
\[
\widehat\gamma(q)=\bigl(q,\gamma_{0\,i}^\alpha(q),S^\alpha(q)\bigr)\,.
\]
Since \(h=H\circ\rho\), one has \(h\circ \widehat\gamma = H\circ \gamma_0\). Therefore,
\[
\dd(h\circ \widehat\gamma)=0 \qquad\Longleftrightarrow\qquad \dd(H\circ \gamma_0)=0\,.
\]
Hence, the \(z\)-independent Hamilton--Jacobi equation in the evolution \(k\)-contact formalism recovers exactly the classical \(k\)-symplectic Hamilton--Jacobi equation under \(k\)-contactification. Moreover, the projection to $Q$ is exactly the same in both formalisms, so the integrability condition is also the same. This explains why the Hamilton--Jacobi equation in Theorem \ref{theo:evolutionHJzindependent} is different from the one in Theorem \ref{theo:ClassicalHJzindependent}, and why it is the natural one for evolution \(k\)-vector fields.
\end{remark}

\section{On \texorpdfstring{$z$}{}-dependent \texorpdfstring{$k$}{}-contact Hamilton--Jacobi equations}\label{sec:HJzdependent}
Let \((\mathcal{J}_{Q,k}=\bigoplus^k\cT Q\times\mathbb{R}^k,{\bm \eta}_{Q,k},h)\) be a $k$-contact Hamiltonian system. Consider a section, written in Darboux coordinates, of the form
\[
\gamma\colon Q\times\mathbb{R}^k\longrightarrow \mathcal{J}_{Q,k}\,,\qquad
(q^i,z^\alpha)\longmapsto (q^i,\gamma_i^\alpha(q,z),z^\alpha)\,.
\]

Let \(\mathbf Z=(Z_1,\dots,Z_k)\) be a $k$-vector field on \(\mathcal{J}_{Q,k}\). Writing
\[
Z_\alpha
=
\sum_{i=1}^n (Z_\alpha)^i\frac{\partial}{\partial q^i}
+\sum_{\beta=1}^k\sum_{i=1}^n (Z_\alpha)_i^\beta\frac{\partial}{\partial p_i^\beta}
+\sum_{\beta=1}^k (Z_\alpha)^\beta\frac{\partial}{\partial z^\beta}\,,\qquad \alpha=1,\dotsc,k\,,
\]
the projected $k$-vector field on \(Q\times\mathbb{R}^k\), defined by
\(
\textstyle\mathbf Z^\gamma:=\left(\bigoplus\nolimits^k \T\pi_{Q\times\mathbb{R}^k}\right)\circ (\mathbf Z\circ \gamma)\,,
\) 
reads
\[
Z^\gamma_\alpha
=
\sum_{i=1}^n (Z_\alpha)^i\circ\gamma\, \frac{\partial}{\partial q^i}
+\sum_{\beta=1}^k (Z_\alpha)^\beta\circ\gamma\, \frac{\partial}{\partial z^\beta}\,,\qquad \alpha=1,\dotsc,k\,,
\]
where \(\pi_{Q\times\mathbb{R}^k}\colon \mathcal{J}_{Q,k}\to Q\times\mathbb{R}^k\) is the canonical projection. Recall that \(\mathbf Z\) and \(\mathbf Z^\gamma\) are \emph{\(\gamma\)-related} if
\(\textstyle
\left(\bigoplus\nolimits^k \T\gamma\right)\circ \mathbf Z^\gamma=\mathbf Z\circ\gamma\,.
\) In the \(z\)-dependent setting, however, the natural projected object is only defined modulo the image of \(\ker\chi\) under \(\bigoplus^k\T\pi_{Q\times\mathbb{R}^k}\). Thus, by abuse of notation, we write
\[\textstyle
[\cdot]\colon \bigoplus\nolimits^k \T\mathcal{J}_{Q,k}\longrightarrow \bigoplus\nolimits^k \T\mathcal{J}_{Q,k}/\ker\chi
\]
and
\[
\textstyle[\cdot]\colon \bigoplus\nolimits^k \T(Q\times\mathbb{R}^k)\longrightarrow \bigoplus\nolimits^k \T(Q\times\mathbb{R}^k)\big/\big(\bigoplus\nolimits^k\T\pi_{Q\times\mathbb{R}^k}\big)(\ker\chi)
\]
for the corresponding canonical projections. Accordingly, the projected class of \(\mathbf Z\) along \(\gamma\) is
\[\textstyle
[\mathbf Z^\gamma]:=\left[\big(\bigoplus\nolimits^k \T\pi_{Q\times\mathbb{R}^k}\big)\circ (\mathbf Z\circ \gamma)\right].
\]

Although \(\ker(\bigoplus^k\T\pi_{Q\times\mathbb{R}^k})\) does not contain \(\ker\chi\) in general,  the map \(\bigoplus^k\T\pi_{Q\times\mathbb{R}^k}\) induces a well-defined quotient morphism
\[
\widetilde{\bigoplus^k\T\pi_{Q\times\mathbb{R}^k}}\colon
\bigoplus^k \T\mathcal{J}_{Q,k}/\ker\chi
\longrightarrow
\bigoplus^k \T(Q\times\mathbb{R}^k)\big/\big(\bigoplus^k\T\pi_{Q\times\mathbb{R}^k}\big)(\ker\chi)
\]
by
\[
\widetilde{\bigoplus^k\T\pi_{Q\times\mathbb{R}^k}}([\mathbf Z])
=
\left[
\left(\bigoplus^k\T\pi_{Q\times\mathbb{R}^k}\right)(\mathbf Z)
\right].
\]
Indeed, changing \(\mathbf Z\) by a section of \(\ker\chi\) changes its projection by a section of \((\bigoplus^k\T\pi_{Q\times\mathbb{R}^k})(\ker\chi)\), which is precisely the subbundle divided out in the target.
Accordingly, the projected class of \(\mathbf Z\) along \(\gamma\) is defined by \([\mathbf Z^\gamma]:=\widetilde{\bigoplus^k\T\pi_{Q\times\mathbb{R}^k}}\circ [\mathbf Z]\circ\gamma\).

We say that \([\mathbf Z]\) and \([\mathbf Z^\gamma]\) are \emph{almost \(\gamma\)-related in the \(z\)-dependent sense} if there exists $\bar{\bm Z}\in [\mathbf Z]$ and $\bar{\bm Z}^\gamma\in [\mathbf Z^\gamma]$ such that \([\mathbf{\bar{Z}}\circ\gamma]=[(\bigoplus^k \T\gamma)\circ \mathbf{ \bar{Z}}^\gamma]\).  Then, \([\mathbf Z]\) and \([\mathbf Z^\gamma]\) are almost \(\gamma\)-related if and only if the following diagram commutes:

\[
\begin{tikzcd}[row sep=4.8em, column sep=4.6em]
\mathcal{J}_{Q,k} \arrow[r, "{\mathbf{\overline{Z}}}"] 
& \bigoplus^k\T\mathcal{J}_{Q,k} \arrow[r, "{[\cdot]}"] \arrow[d, "{\bigoplus^k\T\pi_{Q\times\mathbb{R}^k}}"] 
& \bigoplus^k \T\mathcal{J}_{Q,k}/\ker\chi \arrow[d, "{\widetilde{\bigoplus^k\T\pi_{Q\times\mathbb{R}^k}}}"] \\
Q\times\mathbb{R}^k \arrow[r, "{\mathbf{\overline{Z}}^\gamma}"'] \arrow[u, "\gamma"] 
& \bigoplus^k \T(Q\times\mathbb{R}^k) \arrow[r, "{[\cdot]}"'] \arrow[ur, "{\big[\bigoplus^k \T\gamma\big]}"'] 
& \bigoplus^k \T(Q\times\mathbb{R}^k)\big/\big(\bigoplus^k\T\pi_{Q\times\mathbb{R}^k}\big)(\ker\chi)
\end{tikzcd}
\]

One of the keys to extend the theory of Hamilton--Jacobi equations to different geometric settings is to give adequate conditions on the section \(\gamma\) so as to obtain a Hamilton--Jacobi theorem with useful applications. In the \(k\)-contact realm, the following conditions play this role, as shown later in the paper.

\begin{lemma}
Consider \((\mathcal{J}_{Q,k}=\bigoplus^{k}\T^{*}Q\times\mathbb{R}^{k}, \bm\eta_{{Q,k}})\) and a section \(\gamma\colon Q\times\mathbb{R}^{k}\to \mathcal{J}_{Q,k}\), \(\gamma(q,z)=(q,\gamma_{i}^{\alpha}(q,z),z)\), of \({\rm pr}\colon\mathcal{J}_{Q,k}\rightarrow Q\times \mathbb{R}^k\). Then, \(\gamma\) is maximally coisotropic if and only if
\begin{equation}\label{eq:Con1}
\frac{\partial \gamma_{j}^{\alpha}}{\partial q^{i}}+\sum_{\beta=1}^k \gamma_{i}^{\beta}\frac{\partial \gamma_{j}^{\alpha}}{\partial z^{\beta}}=\frac{\partial \gamma_{i}^{\alpha}}{\partial q^{j}}+\sum_{\beta=1}^k \gamma_{j}^{\beta}\frac{\partial \gamma_{i}^{\alpha}}{\partial z^{\beta}}\,,\qquad \alpha=1,\dots,k\,,\qquad i,j=1,\dotsc,n\,.
\end{equation}
For each \(z\in\mathbb{R}^{k}\), the induced map \(\gamma_{z}\colon q\in Q\mapsto (q,\gamma^\alpha_i(q,z))\in \bigoplus^{k}\T^{*}Q\) is isotropic with respect to the canonical \(k\)-symplectic structure \(\bm\omega_{Q,k}\) on \(\bigoplus^{k}\T^{*}Q\) if and only if
\begin{equation}\label{eq:Con2}
\frac{\partial \gamma_{j}^{\alpha}}{\partial q^{i}}=\frac{\partial \gamma_{i}^{\alpha}}{\partial q^{j}}\,,\qquad \alpha=1,\dots,k\,,\qquad i,j=1,\dots,n\,.
\end{equation}
\end{lemma}

\begin{proof}
Consider adapted Darboux coordinates on \(\mathcal{J}_{Q,k}\) and set \(N={\rm Im}\,\gamma\), which is a submanifold of \(\mathcal{J}_{Q,k}\). A direct computation shows that \(\T N\cap\ker\bm \eta_{Q,k}\) is locally generated by the vector fields
\[
X_i=\frac{\partial}{\partial q^{i}}+\sum_{\alpha=1}^k \gamma_{i}^{\alpha}\frac{\partial}{\partial z^{\alpha}}+\sum_{\alpha=1}^k\sum_{j=1}^n\left(\frac{\partial \gamma_{j}^{\alpha}}{\partial q^{i}}+\sum_{\beta=1}^k \gamma_{i}^{\beta}\frac{\partial \gamma_{j}^{\alpha}}{\partial z^{\beta}}\right)\frac{\partial}{\partial p_{j}^{\alpha}}\,,\qquad i=1,\dots,n\,.
\]
Fix \(x\in N\) and let \(v\in\ker (\bm\eta_{Q,k})_x\). Since \(X_1,\dotsc,X_n\) are linearly independent, \(v\) admits a unique decomposition \(v=\sum_{j=1}^n a^j X_j+\sum_{\alpha=1}^k\sum_{j=1}^n u_j^\alpha\,\frac{\partial}{\partial p_j^\alpha}\), for certain constants \(a^j,u^\alpha_j\), with \(j=1,\dotsc,n\) and \(\alpha=1,\dotsc,k\).

Let us assume now that \(v\) satisfies \(\dd\eta^\alpha_{Q,k}(v,X_i)=0\) for \(\alpha=1,\dots,k\) and \(i=1,\dots,n\). Using \(\dd\eta^\alpha_{Q,k}=\sum_{r=1}^n\dd q^r\wedge \dd p_r^\alpha\) for \(\alpha=1,\dotsc,k\), one obtains
\[
\dd\eta^\alpha_{Q,k}(X_m,X_i)=A_{im}^\alpha-A_{mi}^\alpha\,,\qquad \dd\eta^\alpha_{Q,k}\!\left(\frac{\partial}{\partial p_j^\beta},X_i\right)=-\delta^\alpha_\beta\,\delta_{ij}\,,\qquad A_{ij}^\alpha:=\frac{\partial \gamma_{j}^{\alpha}}{\partial q^{i}}+\sum_{\beta=1}^k \gamma_{i}^{\beta}\frac{\partial \gamma_{j}^{\alpha}}{\partial z^{\beta}}\,.
\]
Hence, the equations \(\dd\eta^\alpha_{Q,k}(v,X_i)=0\) for \(\alpha=1,\dotsc,k\) and \(i=1,\dotsc,n\) are equivalent to
\[
u_i^\alpha=\sum_{m=1}^n a^m(A_{im}^\alpha-A_{mi}^\alpha)\,,\qquad \alpha=1,\dotsc,k,\qquad i=1,\dotsc,n\,.
\]
By definition, \((\T N\cap \ker \bm\eta_{Q,k})_x^{\perp_{\bm\eta_{Q,k}}}\subset\ker \bm\eta_{Q,k,x}\) consists of those \(v\) satisfying the above relations. The inclusion \((\T N\cap \ker \bm\eta_{Q,k})_x^{\perp_{\bm\eta_{Q,k}}}\subset (\T N\cap \ker \bm\eta_{Q,k})_x\) holds if and only if \(u_i^\alpha=0\) for all \(\alpha=1,\dotsc,k\) and \(i=1,\dotsc,n\). The equality is obtained if and only if
\[
A_{ij}^\alpha=A_{ji}^\alpha\,,\qquad \alpha=1,\dots,k\,,\qquad i,j=1,\dots,n\,.
\]
This proves the first claim, namely \eqref{eq:Con1}.

For the second statement, fix \(z\in\mathbb R^k\) and consider the map \(\gamma_z\colon Q\to\bigoplus^k \T^*Q\) given by \(\gamma_z(q)=(q,\gamma_i^\alpha(q,z))\). A direct pull-back computation yields
\[
\gamma_z^*\omega_Q^\alpha=\sum_{1\leq i<j\leq n}\left(\frac{\partial \gamma_i^\alpha}{\partial q^j}(q,z)-\frac{\partial \gamma_j^\alpha}{\partial q^i}(q,z)\right)\dd q^i\wedge \dd q^j\,,\qquad \alpha=1,\dotsc,k\,.
\]
Therefore, \(\gamma_z\) is isotropic with respect to \(\bm\omega_{Q,k}\) if and only if \eqref{eq:Con2} holds, which concludes the proof.
\end{proof}
In the case $k=1$, the coisotropy assumption appearing in the $z$-dependent contact Hamilton--Jacobi theory is equivalent to the maximal coisotropy condition used in the present $k$-contact setting. This is a consequence of the following lemma.

\begin{lemma}\label{lem:k1-maximal-coisotropic}
Let \(
\gamma\colon Q\times \mathbb R \to \mathcal J_{Q,1}=\T^*Q\times \mathbb R,
\,\,
\gamma(q,z)=(q,\gamma_i(q,z),z),
\) be a section of the canonical projection \(
{\rm pr}\colon \T^*Q\times \mathbb R\to Q\times \mathbb R
\), and let $N={\rm Im}\,\gamma$. Then, $ \restr{\eta_{Q,1}}{\T N}$ is nowhere vanishing. In particular, $N$ is coisotropic if and only if it is maximally coisotropic.
\end{lemma}

\begin{proof}
Let $(q^i,p_i,z)$ be adapted Darboux coordinates on $\T^*Q\times \mathbb R$, so that $\eta_{Q,1}=\d z-\theta_Q=\d z-\sum_{i=1}^np_i\,\d q^i$. Since $\gamma$ is a section of ${\rm pr}$, one has
\(
\gamma(q,z)=(q,\gamma_i(q,z),z),
\) and therefore
\[
\T\gamma\!\left(\frac{\partial}{\partial z}\right)
=
\frac{\partial}{\partial z}
+
\sum_{i=1}^n
\frac{\partial \gamma_i}{\partial z}\frac{\partial}{\partial p_i}\,,\qquad \eta_{Q,1}\!\left(\T\gamma\!\left(\frac{\partial}{\partial z}\right)\right)=1\,.
\]
Hence, $\restr{\eta_{Q,1}}{\T N}$ is nowhere vanishing and the result follows from Lemma \ref{lem:coisotropic-max-coisotropic-contact-case}.
\end{proof}

It is worth noting that conditions \eqref{eq:Con2} allow for the local existence of functions $W^\alpha(q,z)$ such that
\[
\gamma_i^\alpha=\frac{\partial W^\alpha}{\partial q^i}\,,\qquad \alpha=1,\dotsc,k\,,\qquad i=1,\dotsc,n\,.
\]

The following definition will be useful to describe our $k$-contact Hamilton--Jacobi theorems in the following subsections. Moreover, it is a natural extension of the definition given in \cite{LLM_21a} for the case of contact manifolds. It is indeed an operator that allows us to describe a condition for $h$ in Hamilton--Jacobi theory in the $k$-contact framework.  
\begin{definition}
Let  \(
\alpha \colon Q \times \mathbb{R}^k \longrightarrow 
\bigwedge\nolimits^\ell \cT Q \times \mathbb{R}^k
\) be a section, consider $\ell$ to be a non-negative integer, and fix $z\in\mathbb{R}^k$. Define
\begin{align*}\textstyle
\alpha_z \colon Q &\longrightarrow \bigwedge\nolimits^\ell \cT Q \\
x & \longmapsto \mathrm{pr}_{\Lambda^\ell \cT Q}(\alpha(x,z))\,,
\end{align*}
where $\mathrm{pr}_{\Lambda^\ell \cT Q}$ denotes the canonical projection. The \emph{exterior derivative of $\alpha$ at fixed $z$} is the section \(
\dd_Q \alpha \colon Q \times \mathbb{R}^k \longrightarrow
\bigwedge\nolimits^{\ell+1}\cT Q \times \mathbb{R}^k
\) defined by
\[
\dd_Q \alpha(x,z) = (\dd \alpha_z(x),\, z)\,.
\]
\end{definition}

\subsection{The classical approach}
\label{sec:zdepclassical}

Let us provide the Hamilton--Jacobi theorem for a \({\bm \eta}_{Q,k}\)-Hamiltonian system in the \(z\)-dependent framework. We now formulate the \(z\)-dependent Hamilton--Jacobi theorem intrinsically with respect to the gauge freedom given by \(\ker\chi\). Since \({\bm \eta}_{Q,k}\)-Hamiltonian \(k\)-vector fields associated with the same Hamiltonian \(h\) are only defined up to elements of \(\ker\chi\), their projections to \(Q\times \mathbb R^k\) are naturally defined only up to the projections of such gauge terms. Accordingly, the Hamilton--Jacobi equation must be formulated at the level of these projected classes.

    \begin{theorem}\label{thm:zdependentHJ-gauge} Consider a \(k\)-contact Hamiltonian system \((\mathcal{J}_{Q,k}=\bigoplus^{k}\T^{*}Q\times\mathbb{R}^{k},\bm \eta_{{Q,k}},h)\). Let \([\mathbf X_h]\) denote the class of \(\bm\eta_{Q,k}\)-Hamiltonian \(k\)-vector fields associated with \(h\), modulo \(\ker\chi\), and let \([\mathbf X_h^\gamma]\) be the projected class induced by \([\mathbf X_h]\) along \(\gamma\). Let \(\gamma\colon Q\times\mathbb{R}^{k}\to \mathcal{J}_{Q,k}\) be a section of \({\rm pr}\colon\mathcal{J}_{Q,k}\to Q\times \mathbb{R}^k\), and write \(N={\rm Im}\,\gamma\). Assume that \(N\) is maximally coisotropic. Then, the following statements are equivalent:
\begin{enumerate}[(a)]
    \item there exist a representative \(\mathbf X\in[\mathbf X_h]\) and a representative \(\overline{\mathbf X}^{\gamma}\in[\mathbf X_h^\gamma]\) such that \((\bigoplus^k\T\gamma)(\overline{\mathbf X}^{\gamma})\) and \(\mathbf X\circ\gamma\) determine the same \(z\)-dependent HdDW equations for \(h\);
    \item there exists a matrix of smooth functions \(C=(C_\alpha^\beta)\in \Cinfty(Q\times\mathbb R^k,\operatorname{Mat}_{k\times k})\) satisfying \(\sum_{\alpha=1}^k C_\alpha^\alpha=-(h\circ\gamma)\) and such that
    \begin{equation}\label{eq:kContactHJGauge}
    \dd_Q(h\circ\gamma)+\sum_{\beta=1}^k \Gamma_\beta\,\gamma^*\theta^\beta+\sum_{\alpha,\beta=1}^k C_\alpha^\beta\,i_{\partial/\partial z^\beta}\,\dd(\gamma^*\theta^\alpha)=0\,,
    \end{equation} 
    where $\Gamma_\beta:=\dd h|_\gamma\!\left(\T\gamma\!\left(\frac{\partial}{\partial z^\beta}\right)\right)$, for $\beta=1,\dots,k$.
\end{enumerate}
\end{theorem}
\begin{proof}
    Take adapted Darboux coordinates \((q^i,p_i^\alpha,z^\alpha)\) on \(\mathcal{J}_{Q,k}\) and write \(\gamma(q,z)=(q,\gamma_i^\alpha(q,z),z)\). Set \(U^i_\alpha:=\big(\frac{\partial h}{\partial p_i^\alpha}\big)\!\circ\gamma\) and \(\Gamma_\beta=\left(\frac{\partial h}{\partial z^\beta}\right)\!\circ\gamma+\sum_{\alpha=1}^k\sum_{i=1}^n U^i_\alpha\,\frac{\partial \gamma_i^\alpha}{\partial z^\beta}\) for $\beta=1,\dots,k$. By the projected gauge freedom induced by \(\ker\chi\), every representative \(\overline{\mathbf X}^{\gamma}\in[\mathbf X_h^\gamma]\) can be written locally as
    \[
        \overline X_\alpha^\gamma=\sum_{i=1}^n U^i_\alpha\,\frac{\partial}{\partial q^i}+\sum_{\beta=1}^k\left(\sum_{j=1}^n \gamma_j^\beta U^j_\alpha+C_\alpha^\beta\right)\frac{\partial}{\partial z^\beta}\,,\qquad \alpha=1,\dots,k\,,
    \]
    for some matrix \(C=(C_\alpha^\beta)\) satisfying \(\sum_{\alpha=1}^k C_\alpha^\alpha=-(h\circ\gamma)\). The trace condition is precisely the projected form of the identity \(\sum_{\alpha=1}^k\eta^\alpha({X}_\alpha)=-h\). Assume first (a). Now
    \[
        \T\gamma(\overline X_\alpha^\gamma)=\sum_{i=1}^n U^i_\alpha\,\T\gamma\!\left(\frac{\partial}{\partial q^i}\right)+\sum_{\beta=1}^k\left(\sum_{j=1}^n \gamma_j^\beta U^j_\alpha+C_\alpha^\beta\right)\T\gamma\!\left(\frac{\partial}{\partial z^\beta}\right),
    \]
    so the \(p_j^\alpha\)-components satisfy
    \begin{equation}\label{eq:First}
        \sum_{\alpha=1}^k\bigl(\T\gamma(\overline X_\alpha^\gamma)\bigr)_{p_j^\alpha}=\sum_{\alpha=1}^k\sum_{i=1}^n U^i_\alpha\,\frac{\partial \gamma_j^\alpha}{\partial q^i}+\sum_{\alpha,\beta=1}^k\left(\sum_{i=1}^n\gamma_i^\beta U^i_\alpha+C_\alpha^\beta\right)\frac{\partial \gamma_j^\alpha}{\partial z^\beta}\,. 
    \end{equation}
    Since \(\mathbf X\circ\gamma\) satisfies the same $z$-dependent HdDW equations as \((\bigoplus^k\T\gamma)(\overline{\mathbf X}^{\gamma})\), one has
    \begin{equation}\label{eq:Second}
        \sum_{\alpha=1}^k\bigl(\T\gamma(\overline X_\alpha^\gamma)\bigr)_{p_j^\alpha}=-\frac{\partial h}{\partial q^j}\!\circ\gamma-\sum_{\alpha=1}^k \gamma_j^\alpha\,\frac{\partial h}{\partial z^\alpha}\!\circ\gamma\,,\qquad j=1,\ldots,n\,.
    \end{equation}
    Using the coisotropy condition \eqref{eq:Con1}, namely
    \[
        \frac{\partial \gamma_{j}^{\alpha}}{\partial q^{i}}+\sum_{\beta=1}^k\gamma_{i}^{\beta}\frac{\partial \gamma_{j}^{\alpha}}{\partial z^{\beta}}=\frac{\partial \gamma_{i}^{\alpha}}{\partial q^{j}}+\sum_{\beta=1}^k\gamma_{j}^{\beta}\frac{\partial \gamma_{i}^{\alpha}}{\partial z^{\beta}}\,,\qquad i,j=1,\ldots,n\,,\qquad \alpha=1,\ldots,k\,,
    \]
    we can rewrite the first two sums in \eqref{eq:First} as
    \[
        \sum_{\alpha=1}^k\sum_{i=1}^n U^i_\alpha\,\frac{\partial \gamma_i^\alpha}{\partial q^j}+\sum_{\beta=1}^k\gamma_j^\beta\sum_{\alpha=1}^k\sum_{i=1}^n U^i_\alpha\,\frac{\partial \gamma_i^\alpha}{\partial z^\beta}\,,\qquad j=1,\ldots,n\,.
    \]
    Therefore,
    \[
        \frac{\partial(h\circ\gamma)}{\partial q^j}+\sum_{\beta=1}^k\Gamma_\beta\,\gamma_j^\beta+\sum_{\alpha,\beta=1}^k C_\alpha^\beta\,\frac{\partial \gamma_j^\alpha}{\partial z^\beta}=0\,,\qquad j=1,\dots,n\,.
    \]
    On the other hand, \(\gamma^*\theta^\alpha=\sum_{i=1}^n\gamma_i^\alpha\,\dd q^i\), so \(i_{\partial/\partial z^\beta}\,\dd(\gamma^*\theta^\alpha)=\sum_{i=1}^n\frac{\partial \gamma_i^\alpha}{\partial z^\beta}\,\dd q^i\). Thus, the previous system is exactly \eqref{eq:kContactHJGauge}, and (b) follows. 
    
    Conversely, assume (b) and choose a matrix \(C=(C_\alpha^\beta)\) satisfying the trace condition and \eqref{eq:kContactHJGauge}. Define a representative of the projected class by
    \[
        \overline X_\alpha^\gamma=\sum_{i=1}^n U^i_\alpha\,\frac{\partial}{\partial q^i}+\sum_{\beta=1}^k\left(\sum_{j=1}^n \gamma_j^\beta U^j_\alpha+C_\alpha^\beta\right)\frac{\partial}{\partial z^\beta}\,,\qquad \alpha=1,\dots,k\,.
    \]
    Its \(q\)-components are, by construction, the canonical ones, namely \(U^i_\alpha=(\partial h/\partial p_i^\alpha)\circ\gamma\), which solve its corresponding equations. Moreover,
    \[
        \sum_{\alpha=1}^k\eta^\alpha\!\left(\T\gamma(\overline X_\alpha^\gamma)\right)=\sum_{\alpha=1}^k C_\alpha^\alpha=-(h\circ\gamma)\,,
    \]
    so the \(z\)-equation is also satisfied. Finally, reversing the computation above and using \eqref{eq:kContactHJGauge} together with \eqref{eq:Con1}, we obtain
    \[
        \sum_{\alpha=1}^k\bigl(\T\gamma(\overline X_\alpha^\gamma)\bigr)_{p_j^\alpha}=-\frac{\partial h}{\partial q^j}\!\circ\gamma-\sum_{\alpha=1}^k \gamma_j^\alpha\,\frac{\partial h}{\partial z^\alpha}\!\circ\gamma\,,\qquad j=1,\dots,n\,.
    \]
    Hence, \((\bigoplus^k\T\gamma)(\overline{\mathbf X}^{\gamma})\) and any representative \(\mathbf X\in[\mathbf X_h]\) determine the same \(z\)-dependent HdDW equations for \(h\). This proves (a).
\end{proof}

The previous condition admits a geometric reformulation that makes the role of the projected gauge freedom more transparent. It is worth stressing that, in the \(z\)-dependent Hamilton--Jacobi theorem, the isotropy assumption on the slices \(\gamma_z\colon q\in Q\mapsto \gamma(q,z)\in \bigoplus^k\T^*Q\times \mathbb{R}^k\) relative to $\dd\bm\eta_{Q,k}$ is not needed for the proof of the Hamilton--Jacobi equation itself. The argument only uses the maximal coisotropy of \(N=\Ima\gamma\) in order to rewrite the momentum equations. The isotropy condition for $\gamma$ relative to $\dd \bm\eta_{Q,k}$ is only relevant if one wishes to introduce local generating functions, since it yields the symmetry conditions \(\partial\gamma_i^\alpha/\partial q^j=\partial\gamma_j^\alpha/\partial q^i\), and hence the local existence of functions \(W^\alpha(q,z)\) such that \(\gamma_i^\alpha=\partial W^\alpha/\partial q^i\). It is also worth noting that no integrability assumption on the associated \(k\)-vector fields is used in the proof. Integrability only becomes relevant when the theorem is interpreted in terms of integral sections or actual solutions of the corresponding Hamilton--De Donder--Weyl equations. In this case, one wants to obtain solutions of the lifted system and it is enough to ensure the integrability of $\bm {\overline{X}}^\gamma$. 

\begin{remark}\label{rem:geometric-HJ-condition}
    Let \(\{e_1,\ldots,e_k\}\) be the canonical basis of \(\mathbb R^k\) and let \(\{e^1,\ldots,e^k\}\) be its dual basis. Define the \(\mathbb R^k\)-valued \(1\)-form and function
    \[
        \gamma^*\bm\theta\coloneqq\sum_{\alpha=1}^k \gamma^*\theta^\alpha\otimes \,e_\alpha\in \Omega^1(Q\times\mathbb R^k,\mathbb R^k)\,,\qquad \Gamma\coloneqq\sum_{\alpha=1}^k \Gamma_\alpha\,e^\alpha\in \Cinfty(Q\times\mathbb R^k,(\mathbb R^k)^*)\,,
    \]
    while the \({\rm End}(\mathbb{R}^k)\)-valued \(1\)-form \(\mathcal K_\gamma\in \Omega^1(Q\times\mathbb R^k,{\rm End}(\mathbb R^k))\) is given by
    \[
        \bigl(\mathcal K_\gamma(Y)\bigr)(v):=\bigl(i_{v^V}\dd(\gamma^*\bm\theta)\bigr)(Y)\,,\qquad Y\in\mathfrak X(Q\times\mathbb R^k)\,,\qquad v\in\mathbb R^k\,,
    \]
    where \(v^V\) denotes the constant vertical vector field on \(Q\times\mathbb R^k\) associated with \(v\in\mathbb R^k\). In the canonical basis, one has \((\mathcal K_\gamma)^\alpha{}_\beta=i_{\partial/\partial z^\beta}\,\dd(\gamma^*\theta^\alpha)\). Set
    \[
        \Xi_\gamma\coloneqq\dd_Q(h\circ\gamma)+\langle \Gamma,\gamma^*\bm\theta\rangle\in \Omega^1(Q\times\mathbb R^k)\,.
    \]
    Then, condition (b) in Theorem~\ref{thm:zdependentHJ-gauge} is equivalent to saying that \(\Xi_\gamma\) belongs to the affine family
    \[
        \mathcal T_\gamma(h):=\big\{-\operatorname{tr}(\mathcal A\circ \mathcal K_\gamma)\ \mid\ \mathcal A\in \Cinfty(Q\times\mathbb R^k,{\rm End}(\mathbb R^k))\,,\ \operatorname{tr}(\mathcal A)=-(h\circ\gamma)\big\}\,.
    \]
    Thus, the defect of the naive equation \(\dd_Q(h\circ\gamma)+\langle \Gamma,\gamma^*\bm\theta\rangle=0\) is measured precisely by the affine family generated by \(\mathcal K_\gamma\) with prescribed trace. This is exactly the geometric manifestation of the projected gauge freedom induced by \(\ker\chi\). For \(k=1\), one has \({\rm End}(\mathbb R)\cong \mathbb R\), so the trace condition determines \(\mathcal A\) uniquely, namely \(\mathcal A=-(h\circ\gamma)\). Hence, \(\mathcal T_\gamma(h)\) consists of a single \(1\)-form and the above condition reduces to the \(z\)-dependent contact Hamilton--Jacobi equation presented in \cite{LLM_21a}. Therefore, \eqref{eq:kContactHJGauge} reduces to
    \[
        \dd_Q(h\circ\gamma)+\Gamma\,\gamma^*\theta-(h\circ\gamma)\,i_{\partial/\partial z}\,\dd(\gamma^*\theta)=0\,,
    \]
    which is exactly the \(z\)-dependent contact Hamilton--Jacobi equation in \cite{LLM_21a}. Moreover, by Lemma~\ref{lem:k1-maximal-coisotropic}, in the case \(k=1\) the condition ``maximally coisotropic'' is equivalent to the usual coisotropy condition for the image of a section \(Q\times\mathbb R\to \T^*Q\times\mathbb R\). As said before, the isotropy condition for \(\gamma_z\) is not relevant for the proof of the Hamilton--Jacobi equation itself as used in \cite{LLM_21a}, but it is relevant if one wishes to introduce local generating functions, as performed in the latter cited work.    
\end{remark}

\subsection{The evolution case}
\label{sec:HJzdepEvolution}

The next result is the $z$-dependent Hamilton--Jacobi theorem for evolution $\bm\eta_{Q,k}$-Hamiltonian systems. In this case, the projected gauge term has vanishing trace, so the Hamilton--Jacobi equation is the natural evolution counterpart of the general $z$-dependent one.
\begin{theorem}\label{thm:zdependentHJ-evolution}
Consider the $k$-contact Hamiltonian system 
$(\mathcal J_{Q,k}=\bigoplus^k \T^*Q\times\mathbb R^k,\bm\eta_{Q,k},h)$.
Let $[\mathbf E_h]_{\mathrm{ev}}$ denote the class of  $\bm \eta_{Q,k}$-evolution
  $k$-vector fields associated with $h$, modulo gauge terms
taking values in $\ker\chi$ and preserving the evolution condition. Let
$\gamma\colon Q\times\mathbb R^k\to \mathcal J_{Q,k}$ be a section of
$\operatorname{pr}\colon \mathcal J_{Q,k}\to Q\times\mathbb R^k$, and put
$N=\operatorname{Im}\gamma$. Assume that $N$ is maximally coisotropic. For each
$\mathbf E\in [\mathbf E_h]_{\mathrm{ev}}$, define
$\mathbf E^\gamma\coloneqq(\bigoplus^k \T\operatorname{pr})(\mathbf E\circ\gamma)$, and denote
by $[(\mathbf E_h)^\gamma]_{\mathrm{ev}}$ the corresponding projected class. Then, the
following statements are equivalent:

\begin{enumerate}[(a)]
\item there exist representatives $\mathbf E\in [\mathbf E_h]_{\mathrm{ev}}$ and
$\mathbf E^\gamma\in [(\mathbf E_h)^\gamma]_{\mathrm{ev}}$ such that
$(\bigoplus^k \T\gamma)(\mathbf E^\gamma)$ and $\mathbf E\circ\gamma$ determine the same
$z$-dependent $\bm\eta_{Q,k}$-evolution Hamilton--De Donder--Weyl equations for $h$;

\item there exists a matrix of  functions
$C=(C^\beta_{\alpha})\in \Cinfty(Q\times\mathbb R^k,\operatorname{Mat}_{k\times k})$
satisfying $\sum_{\alpha=1}^k C^\alpha_{\alpha}=0$ and such that
\begin{equation}\label{eq:zdependentHJ-evolution}
{ \dd_Q(h\circ\gamma)}+\sum_{\beta=1}^k \Gamma_\beta\,\gamma^*\theta^\beta
+\sum_{\alpha,\beta=1}^k C^\beta_{\alpha}\,
i_{\partial/\partial z^\beta}d(\gamma^*\theta^\alpha)=0\,,
\end{equation}
where 
$\Gamma_\beta\coloneqq \restr{\d h}{\gamma}\bigl(\T\gamma(\partial/\partial z^\beta)\bigr)$ for
$\beta=1,\dots,k$.
\end{enumerate}
For $k=1$, the trace condition gives $C^1_1=0$, and
\eqref{eq:zdependentHJ-evolution} reduces to
$\dd_Q(h\circ\gamma)+\Gamma\,\gamma^*\theta=0$.
\end{theorem}

\begin{proof}
Take adapted Darboux coordinates $(q^i,p^\alpha_i,z^\alpha)$ on $\mathcal J_{Q,k}$. Hence, one can write $\gamma(q,z)=(q^i,\gamma^\alpha_i(q,z),z^\alpha)$. Set
$U_\alpha^i:=(\partial h/\partial p^\alpha_i)\circ\gamma$ and
$\Gamma_\beta:=\dd h|_\gamma\bigl(\T\gamma(\partial/\partial z^\beta)\bigr)
=(\partial h/\partial z^\beta)\circ\gamma
+\sum_{\alpha=1}^k\sum_{i=1}^n U_\alpha^i\,\partial\gamma^\alpha_i/\partial z^\beta$. By the projected gauge freedom induced by $\ker\chi$, every representative
$\mathbf E^\gamma=(E^\gamma_1,\dots,E^\gamma_k)$ of $[(\mathbf E_h)^\gamma]_{\mathrm{ev}}$
can be written locally as
\[
E^\gamma_\alpha
= \sum_{i=1}^n U_\alpha^i \frac{\partial}{\partial q^i}
+ \sum_{\beta=1}^k\Bigl(\sum_{j=1}^n \gamma^\beta_j U_\alpha^j + C^\beta_{\alpha}\Bigr)
\frac{\partial}{\partial z^\beta}\,,
\qquad \alpha=1,\dots,k\,,
\]
for some matrix $C=(C^\beta_{\alpha})$. Since we are in the evolution case, the condition
$\sum_{\alpha=1}^k \eta^\alpha_{Q,k}(E_\alpha)=0$ implies that the projected gauge term is
traceless, i.e. $\sum_{\alpha=1}^k C^\alpha_{\alpha}=0$.

Assume now (a). Since $(\bigoplus^k \T\gamma)(\mathbf E^\gamma)$ and $\mathbf E\circ\gamma$
determine the same $z$-dependent $\bm\eta_{Q,k}$-evolution HdDW equations for the $k$-vector fields of $h$, their diagonal momentum
equations coincide. Arguing exactly as in the proof of Theorem~\ref{thm:zdependentHJ-gauge}, but using only the maximal
coisotropy condition
\[
\frac{\partial\gamma^\alpha_j}{\partial q^i}
+\sum_{\beta=1}^k \gamma^\beta_i\frac{\partial\gamma^\alpha_j}{\partial z^\beta}
=
\frac{\partial\gamma^\alpha_i}{\partial q^j}
+\sum_{\beta=1}^k \gamma^\beta_j\frac{\partial\gamma^\alpha_i}{\partial z^\beta}\,,
\qquad \alpha=1,\dots,k\,,\qquad i,j=1,\dots,n\,,
\]
one obtains
\[
\frac{\partial(h\circ\gamma)}{\partial q^j}
+\sum_{\beta=1}^k \Gamma_\beta \gamma^\beta_j
+\sum_{\alpha,\beta=1}^k C^\beta_{\alpha}\,
\frac{\partial\gamma^\alpha_j}{\partial z^\beta}=0\,,
\qquad j=1,\dots,n\,.
\]

Using
$\gamma^*\theta^\beta=\sum_{j=1}^n \gamma^\beta_j\,\d q^j$, the previous system is exactly
\eqref{eq:zdependentHJ-evolution}. This proves (b).

Conversely, assume (b). Choose a representative $\mathbf E^\gamma$ of the projected class
with the local form above. By definition of the projected class, there exists
$\mathbf E\in [\mathbf E_h]_{\mathrm{ev}}$ such that
$\mathbf E^\gamma=(\bigoplus^k \T\operatorname{pr})(\mathbf E\circ\gamma)$. By construction,
the $q$-components of $(\bigoplus^k \T\gamma)(\mathbf E^\gamma)$ are the canonical ones
$U_\alpha^i=(\partial h/\partial p^\alpha_i)\circ\gamma$, so the $q$-equations hold.
Moreover,
\[
    \sum_{\alpha=1}^k \eta^\alpha_{Q,k}\bigl(\T\gamma(E^\gamma_\alpha)\bigr)
=\sum_{\alpha=1}^k C^\alpha_{\alpha}=0\,,
\]
so the evolution $z$-equation is also satisfied.

Finally, reversing the previous computation and using again the maximal coisotropy condition,
equation \eqref{eq:zdependentHJ-evolution} yields
\[
\sum_{\alpha=1}^k \bigl(\T\gamma(E^\gamma_\alpha)\bigr)_{p^\alpha_j}
=
-\frac{\partial h}{\partial q^j}\circ\gamma
-\sum_{\alpha=1}^k \gamma^\alpha_j\,\frac{\partial h}{\partial z^\alpha}\circ\gamma\,,
\qquad j=1,\dots,n\,.
\]
Hence, $(\bigoplus^k \T\gamma)(\mathbf E^\gamma)$ and $\mathbf E\circ\gamma$ determine the same
$z$-dependent $\bm\eta_{Q,k}$-evolution HdDW equations for $h$. This proves (a).
\end{proof}



\section{Integrable contact Hamiltonian systems and perspectives for the \texorpdfstring{$k$}{}-contact case}\label{sec:Integrable}

In the symplectic setting, the image of a solution of the Hamilton--Jacobi equation is a Lagrangian submanifold of \(\T^*Q\). In the completely integrable case, complete solutions give rise to Lagrangian foliations invariant under the Hamiltonian flow. In the \(k\)-contact framework, the situation depends on the version of the theory under consideration. In the \(z\)-independent setting, the image of a solution is Legendrian, whereas the \(z\)-dependent theory suggests that the natural geometric object is maximally coisotropic. This motivates the following notions.

\begin{definition}\label{Def:CompleteZindepen}
    A \emph{complete solution of the \(z\)-independent \(k\)-contact classical (resp.\ evolution) Hamilton--Jacobi  problem} for \((\bigoplus^k \T^*Q\times \mathbb R^k,\bm\eta_{Q,k},h)\) is a local bundle diffeomorphism \(\Phi\colon Q\times \mathbb R^{k(n+1)}\to \bigoplus^k \T^*Q\times \mathbb R^k\) such that, for every \(\lambda\in \mathbb R^{k(n+1)}\), the map $\Phi_\lambda\colon Q\to \bigoplus^k \T^*Q\times \mathbb R^k$ given by $q\mapsto \Phi(q,\lambda)$, is a solution of the $z$-independent \(k\)-contact classical (resp. evolution) Hamilton--Jacobi equation.
\end{definition}

In this case, the images of the family \(\{\Phi_\lambda\}\) define locally a foliation by Legendrian submanifolds. The same viewpoint applies in the \(z\)-dependent setting.

\begin{definition}\label{Def:comZdepk}
    A \emph{complete solution of the \(z\)-dependent \(k\)-contact Hamilton--Jacobi (resp.\ evolution) problem} for \((\bigoplus^k \T^*Q\times \mathbb R^k,\bm\eta_{Q,k},h)\) is a local bundle diffeomorphism \(\Phi\colon Q\times \mathbb R^{kn}\times \mathbb R^k\to \bigoplus^k \T^*Q\times \mathbb R^k\) such that, for every \(\lambda\in \mathbb R^{kn}\), the map \(\Phi_\lambda\colon Q\times \mathbb R^k\to \bigoplus^k \T^*Q\times \mathbb R^k\) given by \((q,z)\mapsto \Phi(q,\lambda,z)\), is a solution of the \(z\)-dependent \(k\)-contact Hamilton--Jacobi problem (resp.\ of the \(z\)-dependent \(k\)-contact evolution Hamilton--Jacobi problem). 
\end{definition}

Given such a complete solution, the family \(\mathcal F=\{\mathcal F_\lambda:=\Ima\Phi_\lambda\mid \lambda\in \mathbb R^{kn}\}\subseteq \bigoplus^k \T^*Q\times \mathbb R^k\) defines a foliation by \((n+k)\)-dimensional maximally coisotropic submanifolds. Thus, in the 
\(z\)-dependent case, complete solutions are naturally related to maximally coisotropic foliations. For comparison, let \((\T^*Q\times \mathbb R,\eta,h)\) be a contact Hamiltonian system. Recall that the action-dependent Hamilton--Jacobi equation is \[\dd_Q(h\circ\gamma)+\frac{\partial(h\circ\gamma)}{\partial z}\,\gamma=(h\circ\gamma)\,\Lie_{\partial/\partial z}\gamma.\] A \emph{complete solution of the action-dependent Hamilton--Jacobi problem} for \((\T^*Q\times \mathbb R,\eta,h)\) is a local diffeomorphism \(\Phi\colon Q\times \mathbb R^n\times \mathbb R\to \T^*Q\times \mathbb R\) such that, for each \(\lambda\in \mathbb R^n\), the map \(\Phi_\lambda\colon Q\times \mathbb R\to \T^*Q\times \mathbb R\) given by \((q,z)\mapsto \Phi(q,\lambda,z)\), is a solution of the action-dependent Hamilton--Jacobi equation. Given such a complete solution, the family \(\mathcal F=\{\mathcal F_\lambda:=\Ima\Phi_\lambda\mid \lambda\in \mathbb R^n\}\subseteq \T^*Q\times \mathbb R\) defines a foliation by \((n+1)\)-dimensional coisotropic submanifolds. Therefore, in the contact case, complete solutions are naturally related to coisotropic foliations rather than to Legendrian ones. Moreover, this is a particular case of Definition \ref{Def:comZdepk} for $k=1$. 
This also suggests that the appropriate analogue of complete integrability should be formulated in terms of invariant coisotropic foliations. The standard notion for the contact case is given next.

\begin{definition}
    Let \((M,\eta,h)\) be a contact Hamiltonian system. We say that \((M,\eta,h)\) is an \emph{integrable contact Hamiltonian system in $z$-dependent sense} if there exists a foliation \(\mathcal F\) of \(M\) by \((n+1)\)-dimensional coisotropic submanifolds that is invariant under the flow of the Hamiltonian vector field \(X_h\). 
\end{definition}

This motivates the following definition in the 
\(z\)-dependent 
\(k\)-contact setting.
\begin{definition}
Let \((\mathcal J_{Q,k},\bm\eta_{Q,k},h)\) be a \(k\)-contact Hamiltonian
system. We say that it is an \emph{integrable 
\(
k\)-contact Hamiltonian system in the 
\(z\)-dependent sense} if there exists a foliation \(\mathcal F\) of \(\mathcal J_{Q,k}\) by
\((n+k)\)-dimensional maximally coisotropic submanifolds and an integrable
\(\bm\eta_{Q,k}\)-Hamiltonian \(k\)-vector field \(\mathbf X_h\) associated
with \(h\) such that \(D^{\mathbf X_h}\subset T\mathcal F\).
\end{definition}

\section{Applications} \label{sec:Applications}

This section illustrates the previous HdDW formalisms with several representative examples. In general, we consider a scalar field \(u=u(t,x)\), with \((t,x)\in\mathbb{R}^2\), and configuration space \(Q=\mathbb{R}\) with coordinate \(u\). The two-contact phase space is
\(
\mathcal{J}_{\mathbb{R},2}=\bigoplus\nolimits^2 \T^*Q\times\mathbb{R}^2\,,
\)
with canonical coordinates \((u,p^t,p^x,z^t,z^x)\). The canonical two-contact form is
\[
{\bm \eta}_{\mathbb{R},2}
=
(\dd z^t-p^t\,\dd u)\otimes e_1+(\dd z^x-p^x\,\dd u)\otimes e_2\,,
\]
with Reeb vector fields \(R_t=\partial_{z^t}\) and \(R_x=\partial_{z^x}\).

Let \({\bm Z}\) be a two-vector field on \(\mathcal{J}_{\mathbb{R},2}\) and define the projected two-vector field on \(Q\) by
\(
{\bm Z}^\gamma := (\bigoplus^2 \T\pi_Q)\circ ({\bm Z}\circ\gamma).
\)
In coordinates, the projected dynamics on \(Q\) has the form
\[
(Z^\gamma)_\alpha(u)=\left(\frac{\partial h}{\partial p^\alpha}\circ\gamma\right)(u)\,\frac{\partial}{\partial u}\,,
\qquad \alpha=t,x\,,\qquad \forall u\in Q\,.
\]

Meanwhile, the projection of ${\bm Z}$ relative to \(\pi_{Q\times\mathbb{R}^2}\), namely
\(
{\bm Z}^\gamma_{Q\times \mathbb{R}^2} :=  (\bigoplus^2 \T\pi_{Q\times \mathbb{R}^2})\circ ({\bm Z}\circ\gamma),
\)
has the form
\[
{\bm Z}^\gamma_{Q\times\mathbb{R}^2}
=
\sum_{\beta=1}^k\left[\left(\frac{\partial h}{\partial p^\beta}\circ\gamma\right)(u)\,\frac{\partial}{\partial u}
+\sum_{\alpha=t,x}(Z^\gamma)^\alpha_\beta\circ\gamma\,\frac{\partial}{\partial z^\alpha}\right]\otimes e_\beta\,.
\]

In particular, we analyse the following examples:
\begin{itemize}
    \item the damped telegrapher/Klein--Gordon equation,
    \item a dissipative Hunter--Saxton equation,
    \item a simple dissipative first-order field model,
    \item an EIT-type thermodynamic model.
\end{itemize}

It is worth noting that all sections $\gamma:\mathbb{R}\times \mathbb{R}^2\rightarrow \mathcal{J}_{\mathbb{R},2}$ will be maximally coisotropic, since  condition \eqref{eq:Con1} is automatically satisfied when $\dim Q=1$. 

\subsection{The damped telegrapher/Klein--Gordon equation}

We consider the damped telegrapher/Klein--Gordon equation
\begin{equation}\label{eq:telKG_PDE}
u_{tt}-\kappa\,u_{xx}+\lambda\,u_t+\epsilon\,u=0\,,
\qquad \kappa>0\,,\quad \lambda\geq 0\,,\quad \epsilon\in\mathbb R\,.
\end{equation}
This equation includes, as a particular case, certain damped Klein--Gordon equations  \cite{Malhi_Stanislavova_2018} and the classical telegrapher equation describing the electric potential along a transmission line \cite{RWD_94}. More precisely, if \(u(x,t)\) denotes the voltage at position \(x\) and time \(t\), then an infinitesimal segment of line of length \(\delta x\) is modelled by a series resistance \(R\,\delta x\), a series inductance \(L\,\delta x\), a shunt conductance \(G\,\delta x\), and a shunt capacitance \(C\,\delta x\), where \(R,L,G,C\) are the physical parameters per unit length. 

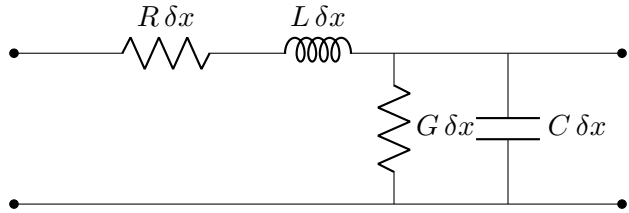
\begin{figure}[h!]
\centering
\begin{circuitikz}[scale=1]
    \draw (0,2) node[circ]{} to[short] (1,2);
    \draw (0,0) node[circ]{} to[short] (6.5,0);
    \draw (1,2) to[R,l=$R\,\delta x$] (3,2)
                 to[L,l=$L\,\delta x$] (5,2)
                 to[short] (6.5,2);
    \draw (5,2) to[R,l=$G\,\delta x$] (5,0);
    \draw (6.5,2) to[C,l=$C\,\delta x$] (6.5,0);
    \draw (6.5,2) to[short] (8,2) node[circ]{};
    \draw (6.5,0) to[short] (8,0) node[circ]{};
\end{circuitikz}
\caption{Infinitesimal element of a transmission line.}
\end{figure}

The corresponding second-order telegrapher equation is
\begin{equation}\label{eq:physical_telegrapher}
LC\,u_{tt}+(RC+LG)\,u_t+RG\,u-u_{xx}=0 \,.
\end{equation}
Let us assume that $L,C>0$. 
Dividing by \(LC\), one obtains
\begin{equation}\label{eq:physical_telegrapher_normalized}
u_{tt}-\frac{1}{LC}u_{xx}+\frac{RC+LG}{LC}u_t+\frac{RG}{LC}u=0 \,.
\end{equation}
Hence, our mathematical model \eqref{eq:telKG_PDE} matches the physical telegrapher equation through the identifications
\begin{equation}\label{eq:parameter_identification}
\kappa=\frac{1}{LC}\,,\qquad \lambda=\frac{RC+LG}{LC}=\frac{R}{L}+\frac{G}{C}\,,\qquad \epsilon=\frac{RG}{LC}\,.
\end{equation}
Therefore, \(\kappa\) is determined by the product \(LC\), while \(\lambda\) encodes the  contribution coming from resistance and leakage, and \(\epsilon\) corresponds to the 0th-order term induced by the simultaneous presence of resistance and conductance. In particular, in the lossless case \(R=G=0\), one has \(\lambda=\epsilon=0\), and \eqref{eq:telKG_PDE} reduces to the wave equation \(u_{tt}-\kappa u_{xx}=0\) with propagation speed \(\sqrt{\kappa}=1/\sqrt{LC}\).

Conversely, given \(\kappa,\lambda,\epsilon\), the relations with the physical parameters are
$LC = 1/\kappa$, $RC+LG = \lambda/\kappa$ and $RG = \epsilon/\kappa$.
Thus, the triple \((\kappa,\lambda,\epsilon)\) determines the combinations \(LC\), \(RC+LG\), and \(RG\), although not the four physical parameters \(R,L,G,C\) separately.

Let us study a two-contact approach to the damped telegrapher/Klein--Gordon equation via the Hamiltonian
\begin{equation}\label{eq:telKG_ham}
h_{\mathrm{tel}}(u,p^t,p^x,z^t,z^x)
=
\frac12\left((p^t)^2-\kappa^{-1}(p^x)^2\right)
+\frac{\epsilon}{2}u^2+\lambda z^t\,.
\end{equation}
It is worth noting that \(h_{\mathrm{tel}}\) is affine in the dissipative variables, and in fact its dependence on them reduces  to \(z^t\). This has relevant consequences for the partial differential equations satisfied by the dissipative variables and $u(t,x)$, as discussed in Section \ref{Sec:SpecialTypes}.

Let \(\psi(t,x)=(u(t,x),p^t(t,x),p^x(t,x),z^t(t,x),z^x(t,x))\) be a solution section of the $z$-independent classical HdDW equations associated with $h_{\rm tel}$. Then, 
\begin{equation}\label{eq:FirstEqu}
u_t=\frac{\partial h_{\mathrm{tel}}}{\partial p^t}=p^t\,,
\qquad
u_x=\frac{\partial h_{\mathrm{tel}}}{\partial p^x}=-\kappa^{-1}p^x\,,
\end{equation}
and the momentum balance
\[
\partial_t p^t+\partial_x p^x
=
-\left(\frac{\partial h_{\mathrm{tel}}}{\partial u}
+p^t\frac{\partial h_{\mathrm{tel}}}{\partial z^t}
+p^x\frac{\partial h_{\mathrm{tel}}}{\partial z^x}\right)
=-(\epsilon u+\lambda p^t)\,.
\]
Note that $h_{\rm tel}$ is a regular Hamiltonian and is affine in the dissipative variables. Then,
eliminating \(p^t=u_t\) and \(p^x=-\kappa u_x\), one recovers \eqref{eq:telKG_PDE}, which  is precisely the particular case of \eqref{eq:second_order_system_linear_z} corresponding to $h_{\rm tel}$. Note that \(z^t\) and \(z^x\) are auxiliary variables introduced by our two-contact formalism and play a secondary role in this example. Since the telegrapher equation is recovered solely from the equations for \(u\), \(p^t\), and \(p^x\) and $h_{\rm tel}$ is affine in the dissipative variables, the $z$-independent standard and evolution HdDW equations lead to the same second-order equation for \(u\) associated with the same Hamiltonian \(h_{\rm tel}\). The difference lies in the partial differential equations satisfied by the auxiliary functions \(z^t(t,x)\) and \(z^x(t,x)\).

This model properly reflects a limitation of the models with an affine dependence on the dissipative variables given in Section \ref{Sec:SpecialTypes}: the resulting second-order PDE does not depend explicitly on the independent variables through the dissipative sector. Nevertheless, a model
\begin{equation}\label{eq:telKG_ham2}
    h'_{\mathrm{tel}}(u,p^t,p^x,z^t,z^x)
    =
    \frac12\left((p^t)^2-\kappa^{-1}(p^x)^2\right)
    +\frac{\epsilon}{2}u^2+\frac{1}{2}\lambda (z^t)^2\,.
\end{equation}
which is quadratic on a dissipative variable, generates the PDE
\begin{equation}\label{eq:telKG}
    u_{tt}-\kappa\,u_{xx}+\lambda z^t(t,x)\,u_t+\epsilon\,u=0\,,
    \qquad \kappa>0\,,\quad \lambda\geq 0\,,\quad \epsilon\in\mathbb R\,,
\end{equation}
where $z^t$ satisfies a balance equation. This gives a geometric mechanism for producing damped Klein--Gordon-type equations with an effective damping coefficient \(\lambda z^t(t,x)\), once the balance equation for \(z^t\) has been solved or constrained appropriately \cite{Malhi_Stanislavova_2018}. The specific values of $z^t$ can be modelled by means of the balance equation on the dissipative variables, where it can be determined by means of $z^x$. This idea is quite general and leads to a simple approach to analysing systems of PDEs with explicit dependence on the independent variables. Nevertheless, we will hereafter stick to the analysis of $h_{\rm tel}$.

\subsubsection{\texorpdfstring{$z$}{z}-independent holonomic sections and projected integrable field}

A $z$-independent holonomic section $\gamma\colon \mathbb{R}\rightarrow \mathcal{J}_{\mathbb{R},2}$ is
\[
\gamma(u)=(u,\gamma_u^t(u),\gamma_u^x(u),\gamma^t(u),\gamma^x(u))\,,
\qquad
\gamma_u^t(u)=\frac{\d\gamma^t}{\d u}(u)\,,\qquad \gamma_u^x(u)=\frac{\d\gamma^x}{\d u}(u)\,.
\]
In both the   ${\bm \eta}_{\mathbb{R},2}$-Hamiltonian and ${\bm \eta}_{\mathbb{R},2}$-evolution descriptions of the damped telegrapher equation, one starts by describing the $z$-independent HdDW equations for a $k$-contact vector field ${\bm Z}$. Due to \eqref{eq:FirstEqu},
the projected two-vector field ${\bm Z}^\gamma$ on $\mathbb{R}$ has components 
\[
(Z^\gamma)_t=\gamma_u^t(u)\,\frac{\partial}{\partial u}\,,
\qquad
(Z^\gamma)_x=-\kappa^{-1}\gamma_u^x(u)\,\frac{\partial}{\partial u}\,.
\]
Let us focus on cases where ${\bm Z}^\gamma$ is integrable so as to derive solutions of  the damped telegrapher equation. Note that this is enough to generate solutions via our Hamilton--Jacobi Theorems \ref{theo:ClassicalHJzindependent} and \ref{theo:evolutionHJzindependent}: no integrability of ${\bm Z}$ is required, which significantly simplifies the applications of our models. As both components of ${\bm Z}^\gamma$ are multiples of $\partial_u$,
the integrability of ${\bm Z}^\gamma$ is equivalent to $[(Z^\gamma)_t,(Z^\gamma)_x]=0$, which yields
\[
\frac{\d}{\d u}\left(\frac{\gamma_u^t(u)}{\gamma_u^x(u)}\right)=0
\quad\Longrightarrow\quad
\frac{\gamma_u^t(u)}{\gamma_u^x(u)}=c\ \text{(constant)}\,,
\]
on any domain in $Q$ where $\gamma_u^x\neq 0$, which we assume hereafter for simplicity. Hence, $\gamma_u^t=c\,\gamma_u^x$. Let $\sigma\colon\mathbb R^2\to Q=\mathbb R$ be an integral section of ${\bm Z}^\gamma$, i.e.
\[
    u(t,x):=\sigma(t,x)\,,\qquad
    u_t=\gamma_u^t(u)\,,\qquad u_x=-\kappa^{-1}\gamma_u^x(u)\,.
\]
With $\gamma_u^t=c\,\gamma_u^x$, define (locally, where $\gamma_u^x\neq0$) that 
\(F(u):=\int^u\frac{\d s}{\gamma_u^x(s)}\).
Then,
\[
\partial_t(F(u(t,x)))=\frac{\gamma_u^t(u(t,x))}{\gamma_u^x(u(t,x))}=c\,,\qquad 
\partial_x(F(u(t,x)))=\frac{-\kappa^{-1}\gamma_u^x(u(t,x))}{\gamma_u^x(u(t,x))}=-\kappa^{-1}\,,
\]
and 
\begin{equation}\label{eq:projected_solution_general}
    F(u(t,x))=c\,t-\kappa^{-1}x+C\,,
    \qquad
    u(t,x)=F^{-1}(c\,t-\kappa^{-1}x+C)\,,
\end{equation}
on any domain where $F$ is invertible and for certain $C\in \mathbb{R}$.  We will not analyse the case $\gamma_u^x=0$. 

The induced lifted solutions (whenever the corresponding $z$-independent Hamilton--Jacobi condition \eqref{eq:HJ_standard_telKG_expanded} holds) are $\psi=\gamma\circ\sigma$:
\[
p^t(t,x)=\gamma_u^t(u(t,x))=c\,\gamma_u^x(u(t,x))\,,\qquad
p^x(t,x)=\gamma_u^x(u(t,x))\,,
\]
\[
    z^t(t,x)=\gamma^t(u(t,x))\,,\qquad z^x(t,x)=\gamma^x(u(t,x))\,,
    \qquad \frac{\d\gamma^t}{\d u}=\gamma_u^t\,,\qquad\frac{\d\gamma^x}{\d u}=\gamma_u^x\,,
\]
and from $\gamma_u^t=c\gamma_u^x$ one gets \begin{equation}\label{ConditionGammaTX}
\gamma^t(u)=c\,\gamma^x(u)+C_0.
\end{equation}

We now distinguish the standard and evolution \(z\)-independent Hamilton--Jacobi equations.

\subsubsection{Hamiltonian  \texorpdfstring{$z$}{}-independent HJ equation and a solvable family}

The standard $z$-independent Hamilton--Jacobi condition is
$
h_{\mathrm{tel}}\circ\gamma=0,
$
that is,
\begin{equation}\label{eq:HJ_standard_telKG_expanded}
\frac12\left((\gamma_u^t(u))^2-\kappa^{-1}(\gamma_u^x(u))^2\right)
+\frac{\epsilon}{2}u^2+\lambda \gamma^t(u)=0\,.
\end{equation}
As in the previous discussion, the projected two-vector field is integrable if and only if
\(
\gamma_u^t=c\,\gamma_u^x
\) 
on any domain where $\gamma_u^x\neq0$. Using also $(\gamma^t)'(u)=\gamma_u^t(u)$, differentiating \eqref{eq:HJ_standard_telKG_expanded} with respect to \(u\) yields the necessary compatibility equation
\begin{equation}\label{eq:compat_ODE}
(c^2-\kappa^{-1})\,\gamma_u^x(u)\,(\gamma_u^x)'(u)+\lambda c\,\gamma_u^x(u)+\epsilon\,u=0\,.
\end{equation}
Unlike the evolution case considered below, \eqref{eq:compat_ODE} is only a necessary condition for the $z$-independent standard Hamilton--Jacobi problem, since one must still impose the original equation \eqref{eq:HJ_standard_telKG_expanded}.

To obtain an explicit family of standard solutions, let us consider the Ansatz
\(
\gamma_u^x(u)=a\,u\,,\, \gamma_u^t(u)=c\,a\,u\,,
\)
with \(a\neq0\) being a constant. Then, \eqref{eq:compat_ODE} reduces to
\begin{equation}\label{eq:a_quadratic}
(c^2-\kappa^{-1})\,a^2+\lambda c\,a+\epsilon=0\,.
\end{equation}
When \(c^2\neq \kappa^{-1}\), the corresponding roots are
\[
a=\frac{-\lambda c\pm\sqrt{\lambda^2 c^2-4(c^2-\kappa^{-1})\epsilon}}{2(c^2-\kappa^{-1})}\,.
\]
Integrating the relations $(\gamma^x)'=\gamma_u^x$,$(\gamma^t)'=\gamma_u^t$, and using \eqref{ConditionGammaTX}, one gets
\[
\gamma^x(u)=\frac{a}{2}u^2+C_1\,,\qquad
\gamma^t(u)=\frac{ca}{2}u^2+cC_1+C_0\,,
\]
for two constants $C_0,C_1$.
Substituting these expressions into \eqref{eq:HJ_standard_telKG_expanded} and using \eqref{eq:a_quadratic}, the terms proportional to \(u^2\) cancel and one is left with
\(
\lambda (cC_1+C_0)=0
\).
Hence, if \(\lambda\neq0\), one must take \(C_0+cC_1=0\), whereas if \(\lambda=0\), the constant \(C_0+cC_1\) is arbitrary. Therefore, under this additional condition, the above Ansatz indeed provides a family of solutions of the standard \(z\)-independent classical ${\bm \eta}_{\mathbb{R},2}$-Hamiltonian Hamilton--Jacobi equation.

The projected integral sections are obtained exactly as before. Since
\[
F(u):=\int^u\frac{\d s}{\gamma_u^x(s)}=\frac1a\ln|u|\,,
\]
one has
\[
F(u(t,x))=c\,t-\kappa^{-1}x+C\,,
\qquad
u(t,x)=u_0\,\exp\!\left(a(c\,t-\kappa^{-1}x)\right)\,,
\]
for a constant \(u_0\neq0\). The induced lifted solutions \(\psi=\gamma\circ\sigma\) are then
\[
p^t(t,x)=c\,a\,u(t,x)\,,\qquad p^x(t,x)=a\,u(t,x)\,,
\]
\[
z^t(t,x)=\frac{ca}{2}u(t,x)^2+cC_1+C_0\,,\qquad
z^x(t,x)=\frac{a}{2}u(t,x)^2+C_1\,.
\]

\subsubsection{$z$-independent evolution for the damped telegrapher/Klein--Gordon equation}

Since the initial Hamiltonian $h_{\rm tel}$ is affine in $z^t,z^x$, we can consider the same Hamiltonian as before, and we obtain the same associated damped telegrapher/Klein--Gordon equation for \(u(x,t)\). Nevertheless, the \(z\)-independent $\bm\eta_{\mathbb{R},2}$-evolution HdDW equations \eqref{eq:k-contact-HdDW-Darboux-coordinates2} read
\begin{equation}\label{eq:evolHdDW_telKG}
\begin{dcases}
u_t=\frac{\partial h_{\mathrm{tel}}}{\partial p^t}\circ\psi=p^t\,,
\qquad
u_x=\frac{\partial h_{\mathrm{tel}}}{\partial p^x}\circ\psi=-\kappa^{-1}p^x\,,\\[2mm]
\partial_t p^t+\partial_x p^x
=
-\left(\frac{\partial h_{\mathrm{tel}}}{\partial u}
+p^t\frac{\partial h_{\mathrm{tel}}}{\partial z^t}
+p^x\frac{\partial h_{\mathrm{tel}}}{\partial z^x}\right)\circ\psi
=-(\epsilon u+\lambda p^t)\,,\\[2mm]
\partial_t z^t+\partial_x z^x
=
\left(p^t\frac{\partial h_{\mathrm{tel}}}{\partial p^t}
+p^x\frac{\partial h_{\mathrm{tel}}}{\partial p^x}\right)\circ\psi
=(p^t)^2-\kappa^{-1}(p^x)^2\,.
\end{dcases}
\end{equation}
The only difference is that now \(z^t,z^x\) satisfy a different partial differential equation. Since \(h_{\mathrm{tel}}\) is affine in the variables \(z^t,z^x\), the projected vector field is the same as in the previous approach, and so is the integrability condition. Hence,
\(
\gamma_u^t=c\,\gamma_u^x
\)  
provided \(\gamma_u^x\neq0\). Consequently, the projected solutions are again given by
\[
F(u(t,x))=c\,t-\kappa^{-1}x+C\,,
\qquad
u(t,x)=F^{-1}(c\,t-\kappa^{-1}x+C)\,,
\]
on any domain where \(F\) is invertible, and the lifted solutions are \(\psi=\gamma\circ\sigma\):
\[
p^t(t,x)=\gamma_u^t(u(t,x))=c\,\gamma_u^x(u(t,x))\,,\qquad
p^x(t,x)=\gamma_u^x(u(t,x))\,,
\]
\[
z^t(t,x)=\gamma^t(u(t,x))\,,\qquad z^x(t,x)=\gamma^x(u(t,x))\,,
\qquad \frac{d\gamma^t}{du}=\gamma_u^t\,,\qquad \frac{d\gamma^x}{du}=\gamma_u^x\,.
\]

The difference with the standard formulation lies in the Hamilton--Jacobi equation. In the evolution setting, the \(z\)-independent Hamilton--Jacobi condition is
\[
\dd(h_{\mathrm{tel}}\circ\gamma)=0\,,
\]
that is,
\[
\frac{\dd}{\dd u}\left[
\frac12\left((\gamma_u^t(u))^2-\kappa^{-1}(\gamma_u^x(u))^2\right)
+\frac{\epsilon}{2}u^2+\lambda \gamma^t(u)
\right]=0\,.
\]
Equivalently, in terms of $\gamma^t,\gamma^x$ and $u$,
\[
\frac{\dd}{\dd u}\left[
\frac12\left(\left(\frac{\d\gamma^t}{\d u}\right)^2-\kappa^{-1}\left(\frac{\d\gamma^x}{\d u}\right)^2\right)
+\frac{\epsilon}{2}u^2+\lambda \gamma^t(u)
\right]=0\,.
\]
After substituting \(\gamma_u^t=c\,\gamma_u^x\) and \((\gamma^t)'=\gamma_u^t\), this becomes precisely \eqref{eq:compat_ODE}. Therefore, the previous computation applies directly in the evolution case.

In particular, the Ansatz
\(
\gamma_u^x(u)=a\,u, \gamma_u^t(u)=c\,a\,u,
\) 
with \(a\) satisfying \eqref{eq:a_quadratic}, yields an explicit family of evolution \(z\)-independent Hamilton--Jacobi solutions. Integrating once,
\[
\gamma^x(u)=\frac{a}{2}u^2+C_1\,,\qquad
\gamma^t(u)=\frac{ca}{2}u^2+C_0+cC_1\,,
\]
and hence
\[
u(t,x)=u_0\,\exp\!\left(a(c\,t-\kappa^{-1}x)\right)\,,
\]
\[
p^t(t,x)=c\,a\,u(t,x)\,,\qquad p^x(t,x)=a\,u(t,x)\,,
\]
\[
z^t(t,x)=\frac{ca}{2}u(t,x)^2+C_0+cC_1\,,\qquad
z^x(t,x)=\frac{a}{2}u(t,x)^2+C_1\,.
\]
Moreover, using \eqref{eq:a_quadratic}, one finds
\(
h_{\mathrm{tel}}\circ\gamma=\lambda (C_0+cC_1),
\) 
which is constant, as required. Thus, in the evolution case no further restriction is needed beyond \eqref{eq:a_quadratic}.

The previous \(z\)-independent constructions provide explicit families of solutions of the standard and evolution Hamilton--Jacobi problems, together with their lifted solutions, but they do not yield complete solutions in the sense of Definition~\ref{Def:CompleteZindepen}. Indeed, for \(Q=\mathbb R\) and \(k=2\), a complete \(z\)-independent solution would require a local bundle diffeomorphism
\(
\Phi\colon \mathbb R\times \mathbb R^4\to \mathcal J_{\mathbb R,2},
\)
whereas the Ansatz \(\gamma_u^x(u)=a\,u\), \(\gamma_u^t(u)=c\,a\,u\) is constrained by \eqref{eq:a_quadratic}. In all the $z$-independent Hamilton--Jacobi equations,  $\gamma$ was obtained in terms of three parameters $c,C_0,C_1$, which may satisfy additional conditions. Hence, we cannot construct a complete solution \(
\Phi\colon \mathbb R\times \mathbb R^4\to \mathcal J_{\mathbb R,2},
\) which requires at least four parameters.

\subsubsection{The \texorpdfstring{$z$}{}-dependent Hamiltonian two-contact Hamilton--Jacobi equation}

Let
\(
\gamma\colon \mathbb R\times \mathbb R^2\to \mathcal J_{\mathbb R,2}
\)
be a section of
\(
\operatorname{pr}\colon \mathcal J_{\mathbb R,2}\to \mathbb R\times \mathbb R^2
\),
written as
\[
\gamma(u,z^t,z^x)=\bigl(u,\gamma_u^t(u,z),\gamma_u^x(u,z),z^t,z^x\bigr)\,.
\]
Recall that since $\dim Q=1$ in this case, the maximally coisotropy condition for $\gamma$ is tautologically satisfied. 
Since
\(
\gamma^*\theta^t=\gamma_u^t\,\dd u
\)
and
\(
\gamma^*\theta^x=\gamma_u^x\,\dd u,
\)
Theorem~\ref{thm:zdependentHJ-gauge} yields the following $z$-dependent Hamilton--Jacobi condition for
\(h_{\mathrm{tel}}\): there exists a matrix
\(
C=(C_\alpha^{\beta})_{\alpha,\beta\in\{t,x\}}
\)
of  functions satisfying
\(
C_t^t+C_x^x=-(h_{\mathrm{tel}}\circ\gamma)
\)
and
\begin{equation}\label{eq:HJ-tel-standard}
\partial_u(h_{\mathrm{tel}}\circ\gamma)+\Gamma_t\,\gamma_u^t+\Gamma_x\,\gamma_u^x
+\sum_{\alpha,\beta\in\{t,x\}} C_\alpha^{\beta}\,\partial_{z^\beta}\gamma_u^\alpha=0\,,
\end{equation}
where
\begin{equation}\label{eq:Gamma-tel-standard}
\Gamma_t=\lambda+\gamma_u^t\,\partial_{z^t}\gamma_u^t-\kappa^{-1}\gamma_u^x\,\partial_{z^t}\gamma_u^x\,,
\qquad
\Gamma_x=\gamma_u^t\,\partial_{z^x}\gamma_u^t-\kappa^{-1}\gamma_u^x\,\partial_{z^x}\gamma_u^x\,.
\end{equation}

Fix \(a\neq0\) and parameters \((\mu,\nu)\in\mathbb R^2\). Consider the family of sections
\begin{equation}\label{eq:gamma-zdep-family}
\gamma_{(\mu,\nu)}(u,z^t,z^x)
=
\bigl(u,a\,z^t-\lambda u+\mu,-a\,z^x+\nu,z^t,z^x\bigr).
\end{equation}
For this family,
\(
\partial_{z^t}\gamma_u^t=a,
\)
\(
\partial_{z^x}\gamma_u^x=-a,
\)
and the mixed derivatives vanish. Moreover,
\(
\Gamma_t=a\gamma_u^t+\lambda
\)
and
\(
\Gamma_x=a\kappa^{-1}\gamma_u^x.
\)
A direct computation shows that
\[
\partial_u(h_{\mathrm{tel}}\circ\gamma_{(\mu,\nu)})+\Gamma_t\,\gamma_u^t+\Gamma_x\,\gamma_u^x
=
\epsilon\,u+a\bigl((\gamma_u^t)^2+\kappa^{-1}(\gamma_u^x)^2\bigr).
\]

Hence, \eqref{eq:HJ-tel-standard} is satisfied by choosing, for instance,
\[
C^{\rm st}=
\begin{pmatrix}
A_{\rm st} & 0\\
0 & B_{\rm st}
\end{pmatrix},
\]
with
\[
A_{\rm st}
=
-\frac12\left(
h_{\mathrm{tel}}\circ\gamma_{(\mu,\nu)}
+\frac{\epsilon}{a}u
+(\gamma_u^t)^2+\kappa^{-1}(\gamma_u^x)^2
\right),
\]
\[
B_{\rm st}
=
-\frac12\left(
h_{\mathrm{tel}}\circ\gamma_{(\mu,\nu)}
-\frac{\epsilon}{a}u
-(\gamma_u^t)^2-\kappa^{-1}(\gamma_u^x)^2
\right).
\]
Indeed,
\(
A_{\rm st}+B_{\rm st}=-(h_{\mathrm{tel}}\circ\gamma_{(\mu,\nu)})
\)
and
\(
aA_{\rm st}-aB_{\rm st}
=
-\epsilon u-a\bigl((\gamma_u^t)^2+\kappa^{-1}(\gamma_u^x)^2\bigr).
\) 
Now, define
\[
\Phi\colon \mathbb R\times\mathbb R^2\times\mathbb R^2\to \mathcal J_{\mathbb R,2}\,,
\qquad
\Phi(u,\mu,\nu,z^t,z^x)
=
\bigl(u,a\,z^t-\lambda u+\mu,-a\,z^x+\nu,z^t,z^x\bigr)\,.
\]
Its inverse is given by
\[
\Phi^{-1}(u,p^t,p^x,z^t,z^x)
=
\bigl(u,p^t-a\,z^t+\lambda u,p^x+a\,z^x,z^t,z^x\bigr)\,,
\]
so \(\Phi\) is a global diffeomorphism. Therefore, \(\Phi\) is a complete solution of the
\(z\)-dependent standard two-contact Hamilton--Jacobi problem for \(h_{\mathrm{tel}}\).

\subsubsection{ \texorpdfstring{$z$}{}-dependent evolution two-contact Hamilton--Jacobi equation}

Let us now consider the $z$-dependent evolution Hamilton--Jacobi formulation. By Theorem~\ref{thm:zdependentHJ-evolution}, the corresponding
Hamilton--Jacobi condition is again \eqref{eq:HJ-tel-standard}, but now the matrix
\(
C=(C_\alpha^\beta)
\)
must satisfy the traceless condition
\(
C_t^t+C_x^x=0
\).
For the same family \eqref{eq:gamma-zdep-family}, the previous computation gives
\[
\partial_u(h_{\mathrm{tel}}\circ\gamma_{(\mu,\nu)})+\Gamma_t\,\gamma_u^t+\Gamma_x\,\gamma^x_u
=
\epsilon\,u+a\bigl((\gamma_u^t)^2+\kappa^{-1}(\gamma_u^x)^2\bigr)\,.
\]
Hence, the $z$-dependent $\bm\eta_{\mathbb{R},2}$-evolution  Hamilton--Jacobi equation is satisfied by taking
\[
C^{\rm ev}=
\begin{pmatrix}
A_{\rm ev} & 0\\
0 & -A_{\rm ev}
\end{pmatrix},
\qquad
A_{\rm ev}
=
-\frac12\left(
\frac{\epsilon}{a}u+(\gamma_u^t)^2+\kappa^{-1}(\gamma_u^x)^2
\right).
\]
Indeed, \(C^{\rm ev}\) is traceless and
\[
aA_{\rm ev}-a(-A_{\rm ev})
=
-\epsilon u-a\bigl((\gamma_u^t)^2+\kappa^{-1}(\gamma_u^x)^2\bigr)\,,
\]
and then \eqref{eq:HJ-tel-standard} holds in the $\bm\eta_{\mathbb{R},2}$-evolution sense.

Therefore, the same map \(\Phi\) above provides a complete solution of the \(z\)-dependent
evolution two-contact Hamilton--Jacobi problem for \(h_{\mathrm{tel}}\).

\subsection{A dissipative Hunter--Saxton equation as a \texorpdfstring{$k$}{}-contact Hamiltonian system}

We consider the dissipative Hunter--Saxton equation
\begin{equation}\label{eq:DissHS}
u_{tx}+u\,u_{xx}+\tfrac12 u_x^2+\mu\,u_x=0\,,\qquad \mu\in\mathbb{R}\,,
\end{equation}
which reduces to the classical Hunter--Saxton equation when $\mu=0$. This equation appears, for instance, in the description of orientation waves in nematic liquid crystals. Differentiating \eqref{eq:DissHS} with respect to $x$, one obtains
\begin{equation}
u_{txx}+2u_xu_{xx}+u\,u_{xxx}+\mu\,u_{xx}=0\,,
\end{equation}
which is also common in the literature \cite{HZ_94}. For further information on dissipative Hunter--Saxton-type equations, see \cite{LY_14,Wei_12,HS_91,HZ_94,Len_07,WY_11,LZ_11}.

Let us introduce the Hamiltonian
\begin{equation}\label{HS_diss_ham}
h_{\rm HS}(u,p^t,p^x,z^t,z^x)=-2\,p^t p^x+2u\,(p^t)^2+\mu\,p^t+2\mu  z^t\,,
\end{equation}
where $\mu\ge0$. Since $h_{\rm HS}$ is affine in the variables $z^t,z^x$ (with trivial dependence on $z^x$) and regular in $\mathcal{J}_{\mathbb{R},2}$, the general remarks from the previous sections concerning the $z$-independent standard and evolution HdDW equations apply.

The $z$-independent standard HdDW equations associated with \eqref{HS_diss_ham} are
\begin{equation}\label{eq:HS_standard_HdDW}
u_t=-2p^x+4u\,p^t+\mu\,,\qquad
u_x=-2p^t\,,
\end{equation}
together with
\begin{equation}\label{eq:HS_standard_momentum}
\partial_t p^t+\partial_x p^x=-2(p^t)^2-2\mu  p^t\,.
\end{equation} Since $h_{\rm HS}$ is regular, the 
elimination of the polymomenta from \eqref{eq:HS_standard_HdDW}--\eqref{eq:HS_standard_momentum} gives
\[
u_{tx}+u\,u_{xx}+\tfrac12 u_x^2+\mu u_x=0\,.
\]
Hence, \eqref{eq:DissHS} is recovered.

\subsubsection{ \texorpdfstring{$z$}{}-independent holonomic sections and projected integrable field}

A $z$-independent holonomic section
\(
\gamma\colon \mathbb R\to \mathcal J_{\mathbb R,2}
\)
is locally written as
\begin{equation}\label{eq:HS_gamma_zind}
\gamma(u)=\bigl(u,\gamma_u^t(u),\gamma_u^x(u),\gamma^t(u),\gamma^x(u)\bigr)\,,
\qquad
\gamma_u^t(u)=\frac{\dd\gamma^t}{\dd u}(u)\,,\qquad
\gamma_u^x(u)=\frac{\dd\gamma^x}{\dd u}(u)\,.
\end{equation}
From \eqref{eq:HS_standard_HdDW}, the projected two-vector field on $Q=\mathbb R$ is
\begin{equation}\label{eq:HS_projected_field}
(Z^\gamma)_t=\bigl(-2\gamma_u^x(u)+4u\,\gamma_u^t(u)+\mu\bigr)\frac{\partial}{\partial u}\,,
\qquad
(Z^\gamma)_x=-2\gamma_u^t(u)\frac{\partial}{\partial u}\,.
\end{equation}
Assume that ${\bm Z}^\gamma$ is integrable and that $\gamma_u^t\neq0$ on the domain under consideration. Since both components are multiples of $\partial/\partial u$, integrability is equivalent to
\(
[(Z^\gamma)_t,(Z^\gamma)_x]=0,
\) 
that is,
there exists a constant $c\in\mathbb R$ such that
\begin{equation}\label{eq:HS_integrability_relation}
\gamma_u^x(u)=(2u+c)\gamma_u^t(u)+\frac{\mu}{2}\,.
\end{equation}

Let $\sigma\colon\mathbb R^2\to Q=\mathbb R$ be an integral section of ${\bm Z}^\gamma$, and write $u(t,x):=\sigma(t,x)$. Then,
\begin{equation}\label{eq:HS_projected_system}
u_t=-2\gamma_u^x(u)+4u\,\gamma_u^t(u)+\mu\,,\qquad
u_x=-2\gamma_u^t(u)\,.
\end{equation}
Using \eqref{eq:HS_integrability_relation}, one gets
\(
u_t=-2c\,\gamma_u^t(u),u_x=-2\gamma_u^t(u)\), hence \( u_t=c\,u_x.
\) 
Define locally
\(
F(u):=\int^u\frac{\d s}{-2\gamma_u^t(s)}
\), which gives $\partial_t(F(u))=c$ and $\partial_x(F(u))=1$, and 
\begin{equation}\label{eq:HS_projected_solution}
F(u(t,x))=c\,t+x+C\,,\qquad
u(t,x)=F^{-1}(c\,t+x+C)\,,
\end{equation}
on any domain where $F$ is invertible. The induced lifted solutions are $\psi=\gamma\circ\sigma$, namely
\[
p^t(t,x)=\gamma_u^t(u(t,x))\,,\qquad p^x(t,x)=\gamma_u^x(u(t,x))\,,
\]
\[
z^t(t,x)=\gamma^t(u(t,x))\,,\qquad z^x(t,x)=\gamma^x(u(t,x))\,.
\]
Moreover, \eqref{eq:HS_integrability_relation} implies
\[
\frac{\dd\gamma^x}{\dd u}(u)=(2u+c)\frac{\dd\gamma^t}{\dd u}(u)+\frac{\mu}{2}\,,
\]
hence
\begin{equation}\label{eq:HS_S_relation}
\gamma^x(u)=(2u+c)\gamma^t(u)-2\int^u \gamma^t(s)\,\d s+\frac{\mu}{2}u+C_0\,,
\end{equation}
for some constant $C_0$. Recall that the Hamilton--Jacobi equation must still be imposed in order to obtain actual solutions of the dissipative Hunter--Saxton equation.
\subsubsection{Standard $z$-independent two-contact Hamilton--Jacobi equation}

The standard $z$-independent Hamilton--Jacobi condition is $h_{\rm HS}\circ\gamma=0$, that is,
\begin{equation}\label{eq:HS_standard_HJ}
-2\,\gamma_u^t(u)\gamma_u^x(u)+2u\,(\gamma_u^t(u))^2+\mu\,\gamma_u^t(u)+2\mu\gamma^t(u)=0\,.
\end{equation}
Using  the integrability condition for the projection \eqref{eq:HS_integrability_relation}, this becomes
\begin{equation}\label{eq:HS_standard_HJ_reduced}
-(u+c)\,(\gamma_u^t(u))^2+\mu \gamma^t(u)=0\,.
\end{equation}
This is already the standard  $z$-independent Hamilton--Jacobi equation written in terms of $\gamma_u^t$ and $\gamma^t$. Differentiating \eqref{eq:HS_standard_HJ_reduced} and using $(\gamma^t)'(u)=\gamma_u^t(u)$, one gets the necessary compatibility equation
\begin{equation}\label{eq:HS_standard_gammaODE}
-(\gamma_u^t(u))^2-2(u+c)\gamma_u^t(u)(\gamma_u^t)'(u)+\mu\gamma_u^t(u)=0\,.
\end{equation}
On any domain where $\gamma_u^t\neq0$, this is equivalent to
\begin{equation}\label{eq:HS_standard_gammaODE_linear}
2(u+c)(\gamma_u^t)'(u)+\gamma_u^t(u)-\mu=0\,.
\end{equation}
Hence, on a connected domain where $u+c$ does not vanish, one has
\begin{equation}\label{eq:HS_general_gamma_t}
\gamma_u^t(u)=\mu+\frac{A}{\sqrt{|u+c|}}\,,\qquad A\in\mathbb R\,.
\end{equation}
Using \eqref{eq:HS_integrability_relation}, this yields
\begin{equation}\label{eq:HS_general_gamma_x}
\gamma_u^x(u)=(2u+c)\left(\mu+\frac{A}{\sqrt{|u+c|}}\right)+\frac{\mu}{2}\,.
\end{equation}
Let $\delta:=\operatorname{sgn}(u+c)$ on the chosen connected domain. For $\mu>0$, one may choose
\begin{equation}\label{eq:HS_general_S_t}
\gamma^t(u)=\mu(u+c)+2A\,\delta\sqrt{|u+c|}+\frac{\delta A^2}{\mu}\,,
\end{equation}
which satisfies both $(\gamma^t)'(u)=\gamma_u^t(u)$ and \eqref{eq:HS_standard_HJ_reduced}.

A particularly simple explicit subfamily is obtained by taking $A=0$. Then
\[
\gamma_u^t(u)=\mu\,,\qquad
\gamma_u^x(u)=\mu(2u+c)+\frac{\mu}{2}\,,
\]
and one may choose
\begin{equation}\label{eq:HS_standard_explicit_S}
\gamma^t(u)=\mu(u+c)\,,\qquad
\gamma^x(u)=\mu u^2+\frac{(2c+1)\mu}{2}\,u+C_1\,.
\end{equation}
Since now $F(u)=-u/(2\mu)$, one has that \eqref{eq:HS_projected_solution} yields
\begin{equation}\label{eq:HS_standard_explicit_u}
u(t,x)=u_0-2\mu (x+c\,t)\,
\end{equation}
for some constant $u_0\in\mathbb R$. The corresponding lifted solution is
\[
p^t(t,x)=\mu\,,\qquad
p^x(t,x)=\mu\bigl(2u(t,x)+c\bigr)+\frac{\mu}{2}\,,
\]
\[
z^t(t,x)=\mu\bigl(u(t,x)+c\bigr)\,,\qquad
z^x(t,x)=\mu u(t,x)^2+\frac{2c+1}{2}\mu\,u(t,x)+C_1\,.
\]

\subsubsection{Evolution $z$-independent two-contact Hamilton--Jacobi equation}

In our $z$-independent evolution formulation, only the balance equation for dissipation variables is different from the classical case for a Hamiltonian that is affine in the dissipation variables, so the projected field is still given by \eqref{eq:HS_projected_field}, and the integrability condition \eqref{eq:HS_integrability_relation} and the projected solutions \eqref{eq:HS_projected_solution} remain unchanged. The $z$-independent evolution HdDW equations are
\begin{equation}\label{eq:HS_evol_HdDW}
\begin{dcases}
u_t=-2p^x+4u\,p^t+\mu\,,\qquad u_x=-2p^t\,,\\[2mm]
\partial_t p^t+\partial_x p^x=-2(p^t)^2-2\mu p^t\,,\\[2mm]
\partial_t z^t+\partial_x z^x=-4p^t p^x+4u\,(p^t)^2+\mu p^t\,.
\end{dcases}
\end{equation}
For a $z$-independent holonomic section $\gamma$, its $\bm\eta_{\mathbb{R},2}$-evolution Hamilton--Jacobi condition is
\(
\dd(h_{\rm HS}\circ\gamma)=0
\),
that is,
\begin{equation}\label{eq:HS_evol_zind_HJ}
\frac{\dd}{\dd u}\left[-2\,\gamma_u^t(u)\gamma_u^x(u)+2u\,(\gamma_u^t(u))^2+\mu\,\gamma_u^t(u)+2\mu \gamma^t(u)\right]=0\,.
\end{equation}
Equivalently,
\begin{equation}\label{eq:HS_evol_zind_HJ_constant}
-(u+c)\,(\gamma_u^t(u))^2+\mu \gamma^t(u)=K\,,
\qquad K\in\mathbb R\,.
\end{equation}
Differentiating \eqref{eq:HS_evol_zind_HJ_constant} in terms of $u$, one recovers exactly \eqref{eq:HS_standard_gammaODE}. Hence, the previous computation applies directly in the evolution setting. In particular, for $\mu>0$ and on a connected domain where $u+c$ does not vanish, one may take
\[
\gamma_u^t(u)=\mu+\frac{A}{\sqrt{|u+c|}}\,,\qquad
\gamma_u^x(u)=(2u+c)\left(\mu+\frac{A}{\sqrt{|u+c|}}\right)+\frac{\mu}{2}\,,
\]
together with
\[
\gamma^t(u)=\mu(u+c)+2A\,\delta\sqrt{|u+c|}+\frac{\delta A^2+K}{\mu}\,,
\]
where $\delta=\operatorname{sgn}(u+c)$ on the chosen connected domain. Note that $K$ is now a new free parameter with respect to the ${\bm \eta}_{\mathbb{R},2}$-Hamiltonian for the dissipative Hunter--Saxton equation. For the explicit subfamily $A=0$, one may choose
\begin{equation}\label{eq:HS_evol_explicit_S}
\gamma^t(u)=\mu(u+c)+\frac{K}{\mu}\,,\qquad
\gamma^x(u)=\mu u^2+\frac{(2c+1)\mu}{2}\,u+C_1\,.
\end{equation}
The projected solution is again
\(
u(t,x)=u_0-2\mu(x+c\,t)
\),
while the lifted one reads
\[
p^t(t,x)=\mu\,,\qquad
p^x(t,x)=\mu\bigl(2u(t,x)+c\bigr)+\frac{\mu}{2}\,,
\]
\[
z^t(t,x)=\mu\bigl(u(t,x)+c\bigr)+\frac{K}{\mu}\,,\qquad
z^x(t,x)=\mu u(t,x)^2+\frac{(2c +1)\mu}{2}\,u(t,x)+C_1\,.
\]
A slightly more general explicit subfamily is obtained by taking $A=1$ and, for simplicity, $K=0$. Then
\[
\gamma_u^t(u)=\mu+\frac{1}{\sqrt{|u+c|}}\,,\qquad
\gamma_u^x(u)=(2u+c)\left(\mu+\frac{1}{\sqrt{|u+c|}}\right)+\frac{\mu}{2}\,.
\]
On a connected domain where $u+c$ does not vanish, define again $\delta:={\rm sgn}(u+c)$. Then, one may choose
\[
\gamma^t(u)=\mu(u+c)+2\delta\sqrt{|u+c|}+\frac{\delta}{\mu}\,.
\]
Moreover, integrating $\gamma_u^x(u)$, one gets
\[
\gamma^x(u)=\mu u^2+\frac{(2c+1)\mu}{2}\,u+\frac{4}{3}|u+c|^{3/2}-2c\,\delta\sqrt{|u+c|}+C_1\,.
\]

Since now
\[
F(u)=\int^u\frac{\d s}{-2\gamma_u^t(s)}
=\int^u\frac{\d s}{-2\left(\mu+\frac{1}{\sqrt{|s+c|}}\right)}\,,
\]
a direct computation yields
\[
F(u)=-\frac{u+c}{2\mu}+\frac{\delta}{\mu^2}\sqrt{|u+c|}-\frac{\delta}{\mu^3}\ln\!\bigl(1+\mu\sqrt{|u+c|}\bigr)\,,
\]
up to an irrelevant additive constant. Hence,  
\(
F(u(t,x))=x+ct+C,
\)
and 
\[
-\frac{u(t,x)+c}{2\mu}+\frac{\delta}{\mu^2}\sqrt{|u(t,x)+c|}-\frac{\delta}{\mu^3}
\ln\!\bigl(1+\mu\sqrt{|u(t,x)+c|}\bigr)=x+ct+C\,.
\]

Therefore, if we define
\[
G_\delta(r):=-\frac{\delta}{2\mu}r^2+\frac{\delta}{\mu^2}r-\frac{\delta}{\mu^3}\ln(1+\mu r)\,,\qquad r>0\,,
\]
then $\sqrt{|u(t,x)+c|}=G_\delta^{-1}(x+ct+C)$ and
\[
u(t,x)=\delta\Bigl(G_\delta^{-1}(x+ct+C)\Bigr)^2-c\,.
\]

The corresponding lifted solution is
\[
p^t(t,x)=\mu+\frac{1}{\sqrt{|u(t,x)+c|}}\,,
\qquad  p^x(t,x)=(2u(t,x)+c)\left(\mu+\frac{1}{\sqrt{|u(t,x)+c|}}\right)+\frac{\mu}{2}\,,
\]
\[
z^t(t,x)=\mu\bigl(u(t,x)+c\bigr)+2\delta\sqrt{|u(t,x)+c|}+\frac{\delta}{\mu}\,,
\]
\[
z^x(t,x)=\mu u(t,x)^2+\frac{(2c+1)\mu}{2}\,u(t,x)+\frac{4}{3}|u(t,x)+c|^{3/2}
-2c\,\delta\sqrt{|u(t,x)+c|}+C_1\,.
\]

\subsubsection{\texorpdfstring{$z$}{}-dependent standard and evolution two-contact Hamilton--Jacobi equations}

Let
\(
\gamma\colon \mathbb R\times\mathbb R^2\to \mathcal J_{\mathbb R,2}
\)
be a section of
\(
\operatorname{pr}\colon \mathcal J_{\mathbb R,2}\to \mathbb R\times\mathbb R^2
\),
written as
\begin{equation}\label{eq:HS_gamma_zdep}
\gamma(u,z^t,z^x)=\bigl(u,\gamma_u^t(u,z),\gamma_u^x(u,z),z^t,z^x\bigr)\,,
\qquad z=(z^t,z^x)\,.
\end{equation}
Since $Q=\mathbb R$ is one-dimensional, the coisotropy requirement is automatic. Moreover,
\[
\gamma^*\theta^t=\gamma_u^t\,\d u\,,\qquad \gamma^*\theta^x=\gamma_u^x\,\d u\,,
\qquad
i_{\partial/\partial z^\beta}\d(\gamma^*\theta^\alpha)=\frac{\partial\gamma_u^\alpha}{\partial z^\beta}\,\d u\,.
\]
For \eqref{HS_diss_ham}, one has
\begin{equation}\label{eq:HS_h_gamma}
h_{\rm HS}\circ\gamma=-2\,\gamma_u^t\gamma_u^x+2u\,(\gamma_u^t)^2+\mu\,\gamma_u^t+2\mu z^t\,.
\end{equation}
Furthermore,
\begin{equation}\label{eq:HS_Gamma_t}
\Gamma_t=2\mu+\bigl(-2\gamma_u^x+4u\gamma_u^t+\mu\bigr)\frac{\partial\gamma_u^t}{\partial z^t}-2\gamma_u^t\frac{\partial\gamma_u^x}{\partial z^t}\,,
\end{equation}
and
\begin{equation}\label{eq:HS_Gamma_x}
\Gamma_x=\bigl(-2\gamma_u^x+4u\gamma_u^t+\mu\bigr)\frac{\partial\gamma_u^t}{\partial z^x}-2\gamma_u^t\frac{\partial\gamma_u^x}{\partial z^x}\,.
\end{equation}

Hence the correct $z$-dependent Hamilton--Jacobi equation, both in the standard and in the evolution formulations, is
\begin{equation}\label{eq:HS_zdep_HJ_correct}
\partial_u(h\circ\gamma)+\Gamma_t\,\gamma_u^t+\Gamma_x\,\gamma_u^x+\sum_{\alpha,\beta\in\{t,x\}}C_\alpha^{\beta}\,\frac{\partial\gamma_u^\alpha}{\partial z^\beta}=0\,,
\end{equation}
for a suitable matrix $C=(C_\alpha^{\beta})_{\alpha,\beta\in\{t,x\}}$. The difference between the two cases lies only in the trace condition:
\begin{equation}\label{eq:HS_zdep_trace_standard}
C_t^t+C_x^x=-(h\circ\gamma)
\qquad\text{in the standard case}\,,
\end{equation}
whereas
\begin{equation}\label{eq:HS_zdep_trace_evol}
C_t^t+C_x^x=0
\qquad\text{in the evolution case}\,.
\end{equation}

Let us now exhibit a genuinely $z$-dependent complete solution. Fix $a\neq0$ and parameters $(\rho,\sigma)\in\mathbb R^2$, and define
\begin{equation}\label{eq:HS_gamma_complete_family}
\gamma_{(\rho,\sigma)}(u,z^t,z^x)=\bigl(u,a\,z^t+\rho,-a\,z^x+\sigma,z^t,z^x\bigr)\,.
\end{equation}
Set
\[
T:=a\,z^t+\rho\,,\qquad X:=-a\,z^x+\sigma\,.
\]
Then $\gamma_u^t=T$, $\gamma_u^x=X$, $\partial_{z^t}\gamma_u^t=a$, $\partial_{z^x}\gamma_u^x=-a$, and the mixed $z$-derivatives vanish. Moreover,
\[
h_{\rm HS}\circ\gamma_{(\rho,\sigma)}=-2TX+2uT^2+\mu T+2\mu z^t\,,
\]
\[
\Gamma_t=2\mu+a(-2X+4uT+\mu),\qquad \Gamma_x=2aT\,.
\]
Therefore
\begin{equation}\label{eq:HS_Xi_gamma}
\Xi_\gamma:=\partial_u(h_{\rm HS}\circ\gamma_{(\rho,\sigma)})+\Gamma_t\,T+\Gamma_x\,X=(2+4au)T^2+((2+a)\mu)\,T\,.
\end{equation}

In the standard case, \eqref{eq:HS_zdep_HJ_correct} is satisfied by choosing
\begin{equation}\label{eq:HS_C_standard}
C^{\mathrm{st}}=
\begin{pmatrix}
A_{\mathrm{st}} & 0\\
0 & B_{\mathrm{st}}
\end{pmatrix},
\qquad
A_{\mathrm{st}}=-\frac12\left((h_{\rm HS}\circ\gamma_{(\rho,\sigma)})+\frac{\Xi_\gamma}{a}\right),
\qquad
B_{\mathrm{st}}=-\frac12\left((h_{\rm HS}\circ\gamma_{(\rho,\sigma)})-\frac{\Xi_\gamma}{a}\right).
\end{equation}
Indeed, $A_{\mathrm{st}}+B_{\mathrm{st}}=-(h_{\rm HS}\circ\gamma_{(\rho,\sigma)})$ and $aA_{\mathrm{st}}-aB_{\mathrm{st}}=-\Xi_\gamma$, so \eqref{eq:HS_zdep_HJ_correct} and \eqref{eq:HS_zdep_trace_standard} hold.

In the evolution case, one may instead take
\begin{equation}\label{eq:HS_C_evol}
C^{\mathrm{ev}}=
\begin{pmatrix}
A_{\mathrm{ev}} & 0\\
0 & -A_{\mathrm{ev}}
\end{pmatrix},
\qquad
A_{\mathrm{ev}}=-\frac{\Xi_\gamma}{2a}\,.
\end{equation}
Then, $C^{\mathrm{ev}}$ is traceless and $aA_{\mathrm{ev}}-a(-A_{\mathrm{ev}})=-\Xi_\gamma$, so \eqref{eq:HS_zdep_HJ_correct} and \eqref{eq:HS_zdep_trace_evol} hold.

Finally, define
\begin{equation}\label{eq:HS_complete_map}
\Phi\colon \mathbb R\times\mathbb R^2\times\mathbb R^2\to \mathcal J_{\mathbb R,2}\,,
\qquad
\Phi(u,\rho,\sigma,z^t,z^x)=\bigl(u,a\,z^t+\rho,-a\,z^x+\sigma,z^t,z^x\bigr)\,.
\end{equation}
Its inverse is
\begin{equation}\label{eq:HS_complete_inverse}
\Phi^{-1}(u,p^t,p^x,z^t,z^x)=\bigl(u,p^t-a\,z^t,p^x+a\,z^x,z^t,z^x\bigr)\,.
\end{equation}
Hence, $\Phi$ is a global diffeomorphism. Therefore, $\Phi$ is a complete solution of the $z$-dependent standard and evolution two-contact Hamilton--Jacobi problem for \eqref{HS_diss_ham}. 

\begin{remark}
The main difference between the $z$-dependent evolution and classical approaches is   encoded in the evolution law for the action variables $z^t,z^x$ and in the trace condition imposed on the matrix $C$ in the corresponding $z$-dependent Hamilton--Jacobi equation.
\end{remark}

\subsubsection{A quadratic solution from the \texorpdfstring{$z$}{}-dependent evolution Hamilton--Jacobi equation}

The $z$-dependent Hamilton--Jacobi equation is substantially more flexible than the $z$-independent one. In particular, it allows one to construct sections solving the Hamilton--Jacobi condition whose projected dynamics yields explicit nonlinear solutions of the dissipative Hunter--Saxton equation.

Consider the section $\widetilde{\gamma}:\mathbb{R}\times\mathbb{R}^2\to \mathcal{J}_{\mathbb{R},2}$ given by
\[
\widetilde{\gamma}(u,z^t,z^x)=\left(u,0,\frac{\mu}{2}-z^t,z^t,z^x\right).
\]
Then, $\widetilde{\gamma}_u^t=0$, $\widetilde{\gamma}_u^x=\mu/2-z^t$, $\partial_{z^t}\widetilde{\gamma}_u^x=-1$, and all other $z$-derivatives of $\widetilde{\gamma}_u^t,\widetilde{\gamma}_u^x$ vanish. Moreover,
\[
h_{\mathrm{HS}}\circ\widetilde{\gamma}=2\mu z^t\,,\qquad \Gamma_t=2\mu\,,\qquad \Gamma_x=0\,.
\]
Hence, the left-hand side of the $z$-dependent evolution Hamilton--Jacobi equation is
\[
\partial_u(h_{\mathrm{HS}}\circ\widetilde{\gamma})+\Gamma_t\widetilde{\gamma}_u^t+\Gamma_x\widetilde{\gamma}_u^x+\sum_{\alpha,\beta\in\{t,x\}}C^\beta_\alpha\,\partial_{z^\beta}\widetilde{\gamma}_u^\alpha=-C^t_x\,.
\]
Therefore, the equation is satisfied by any traceless matrix $C=(C^\beta_\alpha)$ with $C^t_x=0$. For instance, one may take
\[
C_{\mathrm{ev}}=
\begin{pmatrix}
1 & 0\\
-2z^t(\mu/2-z^t) & -1
\end{pmatrix}.
\]

Let $\overline{\bm E}^{\widetilde{\gamma}}$ be the projected representative determined by $\widetilde{\gamma}$. Its $u$-components are
\[
(\overline{E}^{\widetilde{\gamma}})_t(u)=\left(\frac{\partial h_{\mathrm{HS}}}{\partial p^t}\circ\widetilde{\gamma}\right)=2z^t\,,\qquad
(\overline{E}^{\widetilde{\gamma}})_x(u)=\left(\frac{\partial h_{\mathrm{HS}}}{\partial p^x}\circ\widetilde{\gamma}\right)=0\,.
\]
Choosing the representative with $(\overline{E}^{\widetilde{\gamma}})_t(z^t)=1$, $(\overline{E}^{\widetilde{\gamma}})_x(z^t)=0$, $(\overline{E}^{\widetilde{\gamma}})_t(z^x)=0$, and $(\overline{E}^{\widetilde{\gamma}})_x(z^x)=-1$, one obtains the commuting vector fields
\[
(\overline{E}^{\widetilde{\gamma}})_t=2z^t\,\frac{\partial}{\partial u}+\frac{\partial}{\partial z^t}\,,
\qquad
(\overline{E}^{\widetilde{\gamma}})_x=-\frac{\partial}{\partial z^x}\,.
\]
Hence, $\overline{{\bm E}}^{\widetilde{\gamma}}$ is integrable. If $\sigma(t,x)=(u(t,x),z^t(t,x),z^x(t,x))$ is an integral section, then
\[
u_t=2z^t\,,\qquad u_x=0\,,\qquad (z^t)_t=1\,,\qquad (z^t)_x=0\,,\qquad (z^x)_t=0\,,\qquad (z^x)_x=-1\,.
\]
Therefore,
\[
z^t(t,x)=t+c_1\,,\qquad z^x(t,x)=-x+c_2\,,
\]
and
\[
u_t(t,x)=2(t+c_1)\,,\qquad u_x(t,x)=0\,.
\]
Thus
\(
u(t,x)=(t+c_1)^2+c_0
\),
 where $c_0,c_1\in\mathbb{R}$. The lifted solution $\psi=\widetilde{\gamma}\circ\sigma$ is
\[
u(t,x)=(t+c_1)^2+c_0\,,\qquad p^t(t,x)=0\,,\qquad p^x(t,x)=\frac{\mu}{2}-t-c_1\,,
\]
\[
z^t(t,x)=t+c_1\,,\qquad z^x(t,x)=-x+c_2\,.
\]
This satisfies the $z$-dependent evolution HdDW equations. In particular, the projected field yields the quadratic solution
\(
u(t,x)=(t+c_1)^2+c_0
\)
of the dissipative Hunter--Saxton equation
\[
u_{tx}+u\,u_{xx}+\frac12 u_x^2+\mu u_x=0\,.
\]

This example shows that the $z$-dependent Hamilton--Jacobi equation may produce admissible sections and projected dynamics.

\subsection{A remark on the \texorpdfstring{$z$}{}-dependent Hamilton--Jacobi equation}

It is worth stressing that, unlike the $z$-independent case, the $z$-dependent Hamilton--Jacobi equation is formulated on sections $\gamma\colon Q\times\mathbb{R}^2\to J_{Q,2}$ and describes projected classes on $Q\times\mathbb{R}^2$. Hence, the conditions that must be satisfied may be weaker than the $z$-independent condition and allows, potentially, for many more admissible sections. In particular, solving the $z$-dependent Hamilton--Jacobi equation does not by itself reconstruct a solution of the initial dissipative Hunter--Saxton equation unless one also chooses a representative of the projected class and an integral section of the corresponding projected two-vector field.  Nevertheless, one can still exhibit explicit $z$-dependent Hamilton--Jacobi solutions. Let us give an additional example. 
Let $\gamma\colon\mathbb{R}\times\mathbb{R}^2\to \mathcal{J}_{\mathbb{R},2}$ be a section of $\mathrm{pr}\colon \mathcal{J}_{\mathbb{R},2}\to \mathbb{R}\times\mathbb{R}^2$, written as
\[
\gamma(u,z^t,z^x)=\bigl(u,\gamma_u^t(u,z),\gamma_u^x(u,z),z^t,z^x\bigr)\,.
\]
For the Hamiltonian
\(
h_{\mathrm{HS}}
\), 
Theorem~\ref{thm:zdependentHJ-gauge} yields the following $z$-dependent standard Hamilton--Jacobi condition: there exists a matrix $C=(C^\beta_\alpha)_{\alpha,\beta\in\{t,x\}}$ satisfying $C_t^t+C_x^x=-(h_{\mathrm{HS}}\circ\gamma)$ and
\[
\partial_u(h_{\mathrm{HS}}\circ\gamma)+\Gamma_t\gamma_u^t+\Gamma_x\gamma_u^x+\sum_{\alpha,\beta\in\{t,x\}}C^\beta_\alpha\,\partial_{z^\beta}\gamma_u^\alpha=0\,,
\]
where
\[
\Gamma_t=\bigl(-2\gamma_u^x+4u\,\gamma_u^t+\mu\bigr)\partial_{z^t}\gamma_u^t-2\gamma_u^t\,\partial_{z^t}\gamma_u^x+2\mu\,,\qquad 
\Gamma_x=\bigl(-2\gamma_u^x+4u\,\gamma_u^t+\mu\bigr)\partial_{z^x}\gamma_u^t-2\gamma_u^t\,\partial_{z^x}\gamma_u^x\,.
\]

Now, fix $c\in\mathbb{R}$ and consider the section
\[
\widetilde{\gamma}(u,z^t,z^x)=\left(u,\mu,\mu(2u+c)+\frac{\mu}{2},z^t,z^x\right).
\]
Then, $\partial_{z^t}\widetilde{\gamma}_u^t=\partial_{z^x}\widetilde{\gamma}_u^t=\partial_{z^t}\widetilde{\gamma}_u^x=\partial_{z^x}\widetilde{\gamma}_u^x=0$, so $\Gamma_t=2\mu$ and $\Gamma_x=0$. Moreover,
\[
h_{\mathrm{HS}}\circ\widetilde{\gamma}=-2\mu^2(u+c)+2\mu z^t\,,
\qquad
\partial_u(h_{\mathrm{HS}}\circ\widetilde{\gamma})=-2\mu^2\,.
\]
Hence,
\[
\partial_u(h_{\mathrm{HS}}\circ\widetilde{\gamma})+\Gamma_t\widetilde{\gamma}_u^t+\Gamma_x\widetilde{\gamma}_u^x
=-2\mu^2+2\mu\cdot\mu=0\,.
\]
Therefore, the $z$-dependent Hamilton--Jacobi equation is satisfied for any matrix $C$ with trace
\[
C_t^t+C_x^x=2\mu^2(u+c)-2\mu z^t\,.
\]
For instance, one may take
\[
C_{\mathrm{st}}=
\begin{pmatrix}
0 & 0\\
0 & 2\mu^2(u+c)-2\mu z^t
\end{pmatrix}.
\]

The corresponding projected field on $\mathbb{R}\times\mathbb{R}^2$ has $u$-components
\[
u_t=\left(\frac{\partial h_{\mathrm{HS}}}{\partial p^t}\circ\widetilde{\gamma}\right)=-2c\mu\,,
\qquad
u_x=\left(\frac{\partial h_{\mathrm{HS}}}{\partial p^x}\circ\widetilde{\gamma}\right)=-2\mu\,.
\]
Hence
\[
u(t,x)=-2\mu(x+ct)+u_0\,,
\qquad u_0\in\mathbb{R}\,,
\]
which indeed solves the dissipative Hunter--Saxton equation
\[
u_{tx}+u\,u_{xx}+\frac12 u_x^2+\mu u_x=0\,,
\]
since $u_{tx}=u_{xx}=0$ and $\frac12 u_x^2+\mu u_x=2\mu^2-2\mu^2=0$.

This example is intentionally simple: it shows that the $z$-dependent Hamilton--Jacobi formalism does yield explicit solutions, while at the same time illustrating that its class-based nature makes it substantially more flexible than the $z$-independent one in many cases.

\subsection{A simple dissipative first-order field model as a \texorpdfstring{$k$}{}-contact Hamiltonian system}

Let us consider the Hamiltonian function
\begin{equation}\label{camassaham}
h(u,p^t,p^x,z^t,z^x)=\frac12\left(u^2+(p^x)^2\right)+\lambda z^t,
\end{equation}
where $\lambda\ge0$ is a dissipation parameter. Since \(R_t(h)=\lambda\) and \(R_x(h)=0\,,\) 
the Hamiltonian is not invariant under the Reeb flow, which reflects the non-conservative character of the dynamics. Note also that this example is not regular, which will allow us to illustrate the features of non-regular Hamiltonians. We will focus on $z$-dependent Hamilton--Jacobi equations.

The standard HdDW equations associated with \eqref{camassaham} are
\begin{equation}\label{eq:wave_standard_HdDW}
u_t=\frac{\partial h}{\partial p^t}=0\,,\qquad
u_x=\frac{\partial h}{\partial p^x}=p^x\,,
\end{equation}
together with
\begin{equation}\label{eq:wave_standard_momentum}
\partial_t p^t+\partial_x p^x
=
-\left(\frac{\partial h}{\partial u}
+p^t\frac{\partial h}{\partial z^t}
+p^x\frac{\partial h}{\partial z^x}\right)
=-(u+\lambda p^t)\,,
\end{equation}
and
\begin{equation}\label{eq:wave_standard_z}
\partial_t z^t+\partial_x z^x
=
p^t\frac{\partial h}{\partial p^t}+p^x\frac{\partial h}{\partial p^x}-h
=
\frac12\bigl((p^x)^2-u^2\bigr)-\lambda z^t\,.
\end{equation}

The evolution HdDW equations associated with \eqref{camassaham} read
\begin{equation}\label{eq:wave_evol_HdDW}
\begin{dcases}
u_t=\frac{\partial h}{\partial p^t}=0\,,\qquad
u_x=\frac{\partial h}{\partial p^x}=p^x\,,\\[2mm]
\partial_t p^t+\partial_x p^x=-\frac{\partial h}{\partial u}-p^t\frac{\partial h}{\partial z^t}-p^x\frac{\partial h}{\partial z^x}=-u-\lambda p^t\,,\\[2mm]
\partial_t z^t+\partial_x z^x
=
p^t\frac{\partial h}{\partial p^t}+p^x\frac{\partial h}{\partial p^x}
=
(p^x)^2\,.
\end{dcases}
\end{equation}
Thus, the projected first-order relation for $u$ and the momentum balance are the same as in the standard case, while the evolution law for the dissipation variables differs. In particular, the dissipation parameter $\lambda$ disappears from the  $z$-dependent balance equation in the dissipation variables. Significantly, one notes that $p^t$ cannot be written in terms of $u(t,x)$ and its derivatives. It can be set freely, but there is no natural second-order system of PDEs for $u(t,x)$ in this case.

\subsubsection{\texorpdfstring{$z$}{}-dependent standard and evolution two-contact Hamilton--Jacobi equations}

Let
\(
\gamma\colon \mathbb R\times\mathbb R^2\to \mathcal J_{\mathbb R,2}
\)
be a section of
\(
\operatorname{pr}\colon \mathcal J_{\mathbb R,2}\to \mathbb R\times\mathbb R^2
\),
written as
\begin{equation}\label{eq:wave_gamma_zdep}
\gamma(u,z^t,z^x)=\bigl(u,\gamma_u^t(u,z),\gamma_u^x(u,z),z^t,z^x\bigr)\,,
\qquad z=(z^t,z^x)\,.
\end{equation}
Since $Q=\mathbb R$ is one-dimensional, the coisotropy requirement is automatic. Moreover,
\[
\gamma^*\theta^t=\gamma_u^t\,\dd u\,,\qquad
\gamma^*\theta^x=\gamma_u^x\,\dd u\,,\qquad
i_{\partial/\partial z^\beta}\d(\gamma^*\theta^\alpha)
=
\frac{\partial\gamma_u^\alpha}{\partial z^\beta}\,\dd u\,.
\]
For \eqref{camassaham}, one has
\begin{equation}\label{eq:wave_h_gamma}
h\circ\gamma=\frac12\left(u^2+(\gamma_u^x)^2\right)+\lambda z^t\,.
\end{equation}
Furthermore,
\begin{equation}\label{eq:wave_Gamma_t}
\Gamma_t=\lambda+\gamma_u^x\frac{\partial\gamma_u^x}{\partial z^t}\,,
\qquad
\Gamma_x=\gamma_u^x\frac{\partial\gamma_u^x}{\partial z^x}\,.
\end{equation}
Hence, the $z$-dependent Hamilton--Jacobi equation takes the form
\begin{equation}\label{eq:wave_zdep_HJ}
\partial_u(h\circ\gamma)+\Gamma_t\,\gamma_u^t+\Gamma_x\,\gamma_u^x
+\sum_{\alpha,\beta\in\{t,x\}}C_\alpha^{\beta}\,\frac{\partial\gamma_u^\alpha}{\partial z^\beta}=0\,,
\end{equation}
for a suitable matrix $C=(C_\alpha^{\beta})_{\alpha,\beta\in\{t,x\}}$. As in the previous examples, the difference between the standard and evolution formulations is encoded in the trace condition:
\begin{equation}\label{eq:wave_trace_standard}
C_t^t+C_x^x=-(h\circ\gamma)
\qquad\text{in the standard case}\,,
\end{equation}
whereas
\begin{equation}\label{eq:wave_trace_evol}
C_t^t+C_x^x=0
\qquad\text{in the evolution case}\,.
\end{equation}

Let us now exhibit a simple $z$-dependent complete solution. Fix $a\neq0$ and parameters $(\rho,\sigma)\in\mathbb R^2$, and define
\begin{equation}\label{eq:wave_gamma_family}
\gamma_{(\rho,\sigma)}(u,z^t,z^x)=\bigl(u,a\,z^t+\rho,-a\,z^x+\sigma,z^t,z^x\bigr)\,.
\end{equation}
Set
$T:=a\,z^t+\rho$ and $X:=-a\,z^x+\sigma$.
Then,
\[
\gamma_u^t=T,\qquad \gamma_u^x=X\,,\qquad
\partial_{z^t}\gamma_u^t=a\,,\qquad
\partial_{z^x}\gamma_u^x=-a\,,
\]
and the mixed $z$-derivatives vanish. Moreover,
\[
h\circ\gamma_{(\rho,\sigma)}=\frac12(u^2+X^2)+\lambda z^t\,,
\qquad
\Gamma_t=\lambda\,,\qquad
\Gamma_x=-aX\,.
\]
Therefore,
\begin{equation}\label{eq:wave_Xi_gamma}
\Xi_\gamma:=\partial_u(h\circ\gamma_{(\rho,\sigma)})+\Gamma_t\,T+\Gamma_x\,X
=
u+\lambda T-aX^2.
\end{equation}

In the standard case, \eqref{eq:wave_zdep_HJ} is satisfied by choosing
\begin{equation}\label{eq:wave_C_standard}
C^{\mathrm{st}}=
\begin{pmatrix}
A_{\mathrm{st}} & 0\\
0 & B_{\mathrm{st}}
\end{pmatrix},
\qquad
A_{\mathrm{st}}=-\frac12\left((h\circ\gamma_{(\rho,\sigma)})+\frac{\Xi_\gamma}{a}\right),
\qquad
B_{\mathrm{st}}=-\frac12\left((h\circ\gamma_{(\rho,\sigma)})-\frac{\Xi_\gamma}{a}\right).
\end{equation}
Indeed,
\[
A_{\mathrm{st}}+B_{\mathrm{st}}=-(h\circ\gamma_{(\rho,\sigma)})\,,
\qquad
aA_{\mathrm{st}}-aB_{\mathrm{st}}=-\Xi_\gamma\,,
\]
so \eqref{eq:wave_zdep_HJ} and \eqref{eq:wave_trace_standard} hold.

In the evolution case, one may instead take
\begin{equation}\label{eq:wave_C_evol}
C^{\mathrm{ev}}=
\begin{pmatrix}
A_{\mathrm{ev}} & 0\\
0 & -A_{\mathrm{ev}}
\end{pmatrix},
\qquad
A_{\mathrm{ev}}=-\frac{\Xi_\gamma}{2a}\,.
\end{equation}
Then $C^{\mathrm{ev}}$ is traceless and
\[
aA_{\mathrm{ev}}-a(-A_{\mathrm{ev}})=-\Xi_\gamma\,,
\]
so \eqref{eq:wave_zdep_HJ} and \eqref{eq:wave_trace_evol} hold.

Finally, define
\begin{equation}\label{eq:wave_complete_map}
\Phi\colon \mathbb R\times\mathbb R^2\times\mathbb R^2\to \mathcal J_{\mathbb R,2}\,,
\qquad
\Phi(u,\rho,\sigma,z^t,z^x)=\bigl(u,a\,z^t+\rho,-a\,z^x+\sigma,z^t,z^x\bigr)\,.
\end{equation}
Its inverse is given by
\begin{equation}\label{eq:wave_complete_inverse}
\Phi^{-1}(u,p^t,p^x,z^t,z^x)=\bigl(u,p^t-a\,z^t,p^x+a\,z^x,z^t,z^x\bigr)\,.
\end{equation}
Hence $\Phi$ is a global diffeomorphism. Therefore, $\Phi$ is a complete solution for both the standard and evolution $z$-dependent two-contact Hamilton--Jacobi problem associated with  \eqref{camassaham}.  

\begin{remark}
The main difference between the $z$-dependent standard and evolution formulations in this example is not the family of sections  itself, but the trace condition imposed on the matrix $C$ in the corresponding $z$-dependent Hamilton--Jacobi equation. In this sense, both approaches admit the same complete family of sections, while the role of dissipation is encoded differently in the associated projected gauge terms.
\end{remark}
\subsection{\texorpdfstring{$k$}{}-contact Hamilton--Jacobi equations for a covariant EIT-type thermodynamic model}\label{subsec:thermo_EIT_reducedDTT}

This example shows that the $k$-contact formalism naturally accommodates a thermodynamic field-theory interpretation. In the present case, the geometric structure is not used to describe time evolution in the mechanical sense, but rather a  relativistic system of constitutive relations, balance laws, and entropy production equations. In this sense, our application is partially based on the $k$-contact formalism for thermodynamic relativistic systems in \cite{HLM_26}. More precisely, our model may be understood as a covariant version of extended irreversible thermodynamics, where fluxes are regarded as independent non-equilibrium variables, and simultaneously as a reduced divergence-type sector, in the sense that the dynamics is encoded by balance laws governed by a generating Hamiltonian. This section is rather a mathematical formalism showing potential further uses in physics.  

Consider the manifold $M=\mathbb{R}^{k^2+4k+2}$ with coordinates $(\widetilde{S}^\mu,N^\mu,T^{\lambda\mu},P^\mu,\xi,\beta_\lambda,V)$, where $\mu,\lambda=1,\dots,k$. We define $\widetilde{S}^\mu=S^\mu+\xi N^\mu-\sum_{\lambda=1}^k\beta_\lambda  T^{\lambda\mu}$, where the $S^\mu$ are entropy fluxes, $N^\mu$ as particle-number fluxes, $T^{\lambda\mu}$ as generalised stress-energy variables which are considered to be general, namely not necessarily symmetric, $P^\mu$ as pressure-type or dissipative fluxes, and $(\xi,\beta_\lambda,V)$ as configuration thermodynamic variables. Let 
\[\bm \eta_{\rm Ther}=\sum_{\mu=1}^k \left(\d \widetilde{S}^\mu-N^\mu \d\xi+\sum_{\lambda=1}^k T^{\lambda\mu}\d\beta_\lambda-P^\mu \d V\right)\otimes e_\mu\,.\] 
This is of Darboux type after the identifications 
\[z^\mu=\widetilde{S}^\mu\,, \qquad q^i=(\xi,\beta_\lambda,V)\,,\qquad p_i^\mu=(N^\mu,-T^{\lambda\mu},P^\mu)\,\qquad i=1,\ldots,n,\quad \mu,\lambda=1,\ldots,k.
\]
Hence, $(M,{\bm \eta}_{\rm Ther})$ is a co-oriented $k$-contact manifold, with Reeb vector fields $R_\mu=\partial/\partial \widetilde{S}^\mu$, $\mu=1,\dots,k$, and with polarisation $\mathcal{V}=\langle \partial/\partial N^\mu,\partial/\partial T^{\lambda\mu},\partial/\partial P^\mu\rangle_{\lambda,\mu=1,\ldots,k}$ and $n=k+2$. Therefore, $(M,\bm \eta_{\rm Ther},\mathcal{V})$ is a polarised co-oriented $k$-contact manifold. This is not the standard description for thermodynamic systems in \cite{HLM_26}, but both are equivalent, as they are the same up to a change of variables from $\widetilde{S}^\mu$ to $S^\mu$ and leaving invariant remaining  coordinates. The reason for the change is that our present form is in canonical $k$-contact coordinates, which makes it easier to apply the HdDW formalism, which takes a particular form in these coordinates. Moreover, note that this formalism, compatible with the one in \cite{HLM_26}, produces field equations with an energy-momentum tensor, $T^{\mu\nu}$, which is the transposed of the standard one, $T_{\rm Phys}^{\nu\mu}=T^{\mu\nu}$, in physics. Of course, in models with a symmetric energy-momentum tensor, this difference has no relevance.  

Let $h=h(\xi,\beta_\lambda,V,N^\mu,T^{\lambda\mu},P^\mu)$ be a Hamiltonian such that $\partial h/\partial \widetilde{S}^\mu=0$ for every $\mu=1,\dots,k$. This assumption means that the Hamiltonian does not depend on the entropy-flux variables, while the constitutive content is encoded in its dependence on the intensive variables and on the fluxes.

\subsection{Standard and evolution thermodynamic HdDW equations}

Let ${\bm X}=(X_1,\dots,X_k)$ be an $\bm\eta_{\rm Ther}$-Hamiltonian $k$-vector field for $
(M,\bm\eta_{\rm Ther},h)$. The local form of the $k$-contact Hamilton--De Donder--Weyl equations gives
\[
\frac{\partial \xi}{\partial x^\mu}=\frac{\partial h}{\partial N^\mu}\,,\qquad
\frac{\partial \beta_\lambda}{\partial x^\mu}=-\frac{\partial h}{\partial T^{\lambda\mu}}\,,\qquad
\frac{\partial V}{\partial x^\mu}=\frac{\partial h}{\partial P^\mu}\,,\qquad \mu=1,\ldots,k,
\]
together with
\[
\sum_{\mu=1}^k\frac{\partial N^\mu}{\partial x^\mu}=-\frac{\partial h}{\partial \xi}\,,\qquad
\sum_{\mu=1}^k\frac{\partial T^{\lambda\mu}}{\partial x^\mu}=\frac{\partial h}{\partial \beta_\lambda}\,,\qquad
\sum_{\mu=1}^k\frac{\partial P^\mu}{\partial x^\mu}=-\frac{\partial h}{\partial V}\,,\qquad \lambda=1,\ldots,k,
\]
and
\[
\sum_{\mu=1}^k\frac{\partial \widetilde{S}^\mu}{\partial x^\mu}
=
\sum_{\mu=1}^k\left(
N^\mu\frac{\partial h}{\partial N^\mu}
+\sum_{\lambda=1}^k T^{\lambda\mu}\frac{\partial h}{\partial T^{\lambda\mu}}
+P^\mu\frac{\partial h}{\partial P^\mu}
\right)-h\,.
\]
Hence, the $k$-contact Hamiltonian formalism yields a first-order thermodynamic field theory of balance-law type, whose constitutive closure is determined by $h$.

A particularly useful class of Hamiltonians is $h=U(\xi,\beta_\lambda,V)+\Phi(N^\mu,T^{\lambda\mu},P^\mu)$. In this case,
\[
\frac{\partial \xi}{\partial x^\mu}=\frac{\partial\Phi}{\partial N^\mu}\,,\qquad
\frac{\partial \beta_\lambda}{\partial x^\mu}=-\frac{\partial\Phi}{\partial T^{\lambda\mu}}\,,\qquad
\frac{\partial V}{\partial x^\mu}=\frac{\partial\Phi}{\partial P^\mu}\,,\qquad \lambda,\mu=1,\ldots,k,
\]
while
\[
\sum_{\mu=1}^k\frac{\partial N^\mu}{\partial x^\mu}=-\frac{\partial U}{\partial \xi}\,,\qquad
\sum_{\mu=1}^k\frac{\partial T^{\lambda\mu}}{\partial x^\mu}=\frac{\partial U}{\partial \beta_\lambda}\,,\qquad
\sum_{\mu=1}^k\frac{\partial P^\mu}{\partial x^\mu}=-\frac{\partial U}{\partial V}\,\qquad \lambda=1,\ldots,k.
\]
When $U=0$, the purely balance-constitutive system reads $\sum_{\mu=1}^k\partial_\mu N^\mu=0$, $\sum_{\mu=1}^k\partial_\mu T^{\lambda\mu}=0$, and $\sum_{\mu=1}^k\partial_\mu P^\mu=0$ for $\lambda=1,\ldots,k$, which is the simplest covariant EIT-type instance within this framework. Note that the equations concerning $N^\mu$ and the $T^{\lambda\mu}$ are quite ubiquitous in the literature \cite{HLM_26,IsraelStewart1979}. 

If ${\bm E}=(E_1,\dots,E_k)$ is an ${\bm \eta}_{\rm Ther}$-evolution $k$-vector field for the same Hamiltonian $h$, then the equations for $\xi$, $\beta_\lambda$, $V$, $N^\mu$, $T^{\lambda\mu}$, and $P^\mu$ remain the same, while the entropy-flux equation can be obtained from 
\[
\sum_{\mu=1}^k\frac{\partial \widetilde{S}^\mu}{\partial x^\mu}
=
\sum_{\mu=1}^k\left(
N^\mu\frac{\partial h}{\partial N^\mu}
+\sum_{\lambda=1}^k T^{\lambda\mu}\frac{\partial h}{\partial T^{\lambda\mu}}
+P^\mu\frac{\partial h}{\partial P^\mu}
\right).
\]
Thus, as in the previous examples, the distinction between the standard and evolution formalisms affects only the entropy equation.

\subsubsection{The \texorpdfstring{$z$}{}-independent Hamilton--Jacobi problem}

Let $Q_{\mathrm{th}}=\mathbb{R}^{k+2}$ with global coordinates $(\xi,\beta_\lambda,V)$, and let $\gamma\colon Q_{\mathrm{th}}\to M$ be a section of the projection $(\xi,\beta_\lambda,V,\widetilde{S}^\mu,N^\mu,T^{\lambda\mu},P^\mu)\mapsto (\xi,\beta_\lambda,V)$. Writing
\[
\gamma(\xi,\beta,V)=\bigl(\widetilde{S}^\mu(\xi,\beta,V),N^\mu(\xi,\beta,V),T^{\lambda\mu}(\xi,\beta,V),P^\mu(\xi,\beta,V),\xi,\beta_\lambda,V\bigr)\,,
\]
the condition $\gamma^*\bm\eta=0$ is equivalent to $\sum_{\mu=1}^k(\d \widetilde{S}^\mu-N^\mu \d\xi+\sum_{\lambda=1}^k T^{\lambda\mu}\d\beta_\lambda-P^\mu \d V)\otimes e_\mu=0$, namely to
\[
N^\mu=\frac{\partial \widetilde{S}^\mu}{\partial \xi},\qquad
T^{\lambda\mu}=-\frac{\partial \widetilde{S}^\mu}{\partial \beta_\lambda},\qquad
P^\mu=\frac{\partial \widetilde{S}^\mu}{\partial V},\qquad \mu,\lambda=1,\ldots,k.
\]
Hence, the functions $\widetilde{S}^\mu=\widetilde{S}^\mu(\xi,\beta,V)$ form a generating family for the thermodynamic section.

The standard $z$-independent Hamilton--Jacobi equation is simply $h\circ\gamma=0$, i.e.
\[
h\!\left(
\xi,\beta_\lambda,V,
\frac{\partial \widetilde{S}^\mu}{\partial \xi},
-\frac{\partial \widetilde{S}^\mu}{\partial \beta_\lambda},
\frac{\partial \widetilde{S}^\mu}{\partial V}
\right)=0.
\]
It makes sense physically to assume solutions for hydrodynamic systems to be included in Legendrian submanifolds, which forces solutions to be included in submanifolds with Hamiltonian equal to zero  \cite{HLM_26}. In this sense, the $z$-independent Hamilton--Jacobi theory for Hamiltonian $k$-vector fields can always be applied to study thermodynamic systems.

In the evolution case, one gets $\d(h\circ\gamma)=0$, or equivalently
\[
h\!\left(
\xi,\beta_\lambda,V,
\frac{\partial \widetilde{S}^\mu}{\partial \xi},
-\frac{\partial \widetilde{S}^\mu}{\partial \beta_\lambda},
\frac{\partial \widetilde{S}^\mu}{\partial V}
\right)=\mathrm{const}.
\]
Again, solutions with Hamiltonian equal to zero, which are physically motivated, can always be studied in this approach. 

\subsubsection{The \texorpdfstring{$z$}{}-dependent Hamilton--Jacobi problem}

Let $\gamma\colon Q_{\mathrm{th}}\times\mathbb{R}^k\to M$ be a section of the projection
\[
(\widetilde{S}^\mu,N^\mu,T^{\lambda\mu},P^\mu,\xi,\beta_\lambda,V)\longmapsto (\xi,\beta_\lambda,V,\widetilde{S}^\mu),
\]
and denote by $(\widetilde{s}^\mu)$ the induced coordinates on the copy of $\mathbb{R}^k$ in the base. Set $\gamma^*\theta^\alpha=N^\alpha \d\xi-\sum_{\lambda=1}^k T^{\lambda\alpha}\d\beta_\lambda+P^\alpha \d V$, $\alpha=1,\dots,k$, and assume that ${\rm Im}\,\gamma$ is maximally coisotropic. Then the general $z$-dependent Hamilton--Jacobi equation reads
\[
\dd_Q(h\circ\gamma)+\sum_{\beta=1}^k\Gamma_\beta\,\gamma^*\theta^\beta+\sum_{\alpha,\beta=1}^k C^\beta_{\alpha}\,i_{\partial/\partial \tilde{s}^\beta}\dd(\gamma^*\theta^\alpha)=0\,,
\]
where $\Gamma_\beta=\d h|_\gamma(\T\gamma(\partial/\partial \widetilde{s}^\beta))$, for $\beta=1,\dots,k$. In the standard case, the matrix $C=(C^\beta_\alpha)$ satisfies $\sum_{\alpha=1}^k C^\alpha_\alpha=-(h\circ\gamma)$, whereas in the evolution case one has $\sum_{\alpha=1}^k C^\alpha_\alpha=0$.

Expanding this equation in the coordinates $(\xi,\beta_\lambda,V)$, one obtains
\[
\frac{\partial(h\circ\gamma)}{\partial \xi}
+\sum_{\beta=1}^k\Gamma_\beta N^\beta
+\sum_{\alpha,\beta=1}^k C^\beta_\alpha\frac{\partial N^\alpha}{\partial \widetilde{s}^\beta}=0\,,
\]
\[
\frac{\partial(h\circ\gamma)}{\partial \beta_\lambda}
-\sum_{\beta=1}^k\Gamma_\beta T^{\lambda\beta}
-\sum_{\alpha,\beta=1}^k C^\beta_\alpha\frac{\partial T^{\lambda\alpha}}{\partial \widetilde{s}^\beta}=0\,,
\]
\[
\frac{\partial(h\circ\gamma)}{\partial V}
+\sum_{\beta=1}^k\Gamma_\beta P^\beta
+\sum_{\alpha,\beta=1}^k C^\beta_\alpha\frac{\partial P^\alpha}{\partial \widetilde{s}^\beta}=0\,.
\]
These are the $z$-dependent Hamilton--Jacobi equations in thermodynamic variables.

The previous construction gives a $k$-contact realisation of a covariant EIT-type model, since the flux variables $(N^\mu,T^{\lambda\mu},P^\mu)$ are treated as independent non-equilibrium fields and the entropy fluxes $S^\mu$ are incorporated as related to genuine geometric variables. At the same time, since the field equations are balance laws closed by a Hamiltonian potential $h$, the model may be regarded as a reduced divergence-type sector. In particular, the choice $h=\Phi(N^\mu,T^{\lambda\mu},P^\mu)$ yields a purely balance-constitutive system, whereas $h=U(\xi,\beta_\lambda,V)+\Phi(N^\mu,T^{\lambda\mu},P^\mu)$ introduces source terms in the balance equations and allows for non-equilibrium thermodynamic couplings without modifying the underlying $k$-contact geometry.
\section{Conclusions and Outlook}\label{sec:ConclusionsAndOutlook}
In this work, we have developed a Hamilton--Jacobi theory for non-conservative classical field theories within the framework of $k$-contact geometry. A central contribution has been the introduction of evolution $k$-contact $k$-vector fields and their associated Hamilton--De Donder--Weyl equations, which provide a genuine counterpart to the standard $k$-contact Hamiltonian formalism. This has led to two genuinely different Hamilton--Jacobi theories: a $z$-independent approach, based on projections onto the configuration space, and a $z$-dependent approach, where the contact variables are incorporated into the base manifold. In the latter case, the theory is naturally formulated in terms of projected classes of $k$-vector fields, reflecting the gauge freedom coming from $\ker\chi$.

From the geometric viewpoint, the paper clarifies several structural features of $k$-contact Hamiltonian systems. First, unlike the contact case, the map $\chi$ is no longer an isomorphism for $k>1$, which explains the non-uniqueness of Hamiltonian $k$-vector fields associated with a given Hamiltonian function and the role of gauge directions taking values in $\ker\chi$. Second, for Hamiltonians that are affine in the variables $z^\alpha$, the projected equations for the configuration variables are the same for the standard and evolution formulations, so the difference between both theories is encoded only in the equations for the contact variables. Third, in the canonical model, holonomic sections are exactly the Legendrian ones, while in the $z$-dependent setting the natural objects are maximally coisotropic submanifolds. This also motivates the notions of complete solution introduced in both approaches and the $z$-dependent notion of integrable $k$-contact Hamiltonian system in terms of invariant foliations by maximally coisotropic submanifolds.

The general theory has been illustrated through several representative examples, including the damped telegrapher/Klein--Gordon equation, the dissipative Hunter--Saxton equation, a simple dissipative first-order field model, and a covariant EIT-type thermodynamic model. These applications show that the formalism is flexible enough to describe both scalar dissipative PDEs and field-theoretical systems governed by balance laws and constitutive relations. In particular, the damped telegrapher/Klein--Gordon and dissipative Hunter--Saxton examples exhibit explicit $z$-independent and $z$-dependent solutions and show how complete solutions arise naturally in the $z$-dependent setting. The simple dissipative first-order model shows that the same projected dynamics may coexist with different geometric dissipative mechanisms. The thermodynamic example shows that the theory also applies to balance laws and constitutive relations in a genuinely field-theoretical setting. In addition, the analysis of Hamiltonians with quadratic dependence on the dissipative variables shows that the scope of the formalism is broader than the affine case usually considered in the literature.

Some questions remain open. More precisely, it would be desirable to formulate a more complete notion of integrability for $k$-contact Hamiltonian systems, relating invariant foliations, complete solutions, and reconstruction procedures in a systematic way. It would also be interesting to characterise more explicitly when the standard and evolution theories induce the same projected dynamics, and to determine intrinsic criteria ensuring the existence of complete solutions in the $z$-dependent and $z$-independent settings.

Other natural directions concern the extension of the theory to singular or constrained Hamiltonians, reduction procedures, gauge symmetries, higher-order field theories, and systems with nontrivial boundary conditions. It also seems worthwhile to analyse in greater depth the relation between $k$-contact, $k$-symplectic, and multisymplectic formalisms, especially in situations where contactification or dissipative deformations of conservative models play a role. The development of Hamilton--Jacobi theories in these broader settings, together with further physically meaningful applications, will be the subject of future work.

\bibliographystyle{abbrv}
\bibliography{references.bib}

\end{document}